%% file: main.tex
\documentclass[acmsmall,nonacm]{acmart}

\usepackage{amsthm}
\usepackage{wrapfig}
\usepackage{stmaryrd}
\usepackage{thmtools}
\usepackage{array}
\usepackage{epic}
\usepackage{eepic}
\usepackage{epsfig}
\usepackage{multicol}
\usepackage{multirow}
\usepackage{MnSymbol}
\usepackage{booktabs} % For formal tables
\usepackage[textwidth=2cm,textsize=small]{todonotes}
\usepackage{url}
\usepackage{color}
\usepackage[bottom]{footmisc}
\usepackage{graphicx}
\usepackage{mathtools}

\usepackage{algorithm}
\usepackage[noend]{algpseudocode}
\usepackage{setspace}

\allowdisplaybreaks

\input{macros}

\AtBeginDocument{%
	\providecommand\BibTeX{{%
			\normalfont B\kern-0.5em{\scshape i\kern-0.25em b}\kern-0.8em\TeX}}}

\setcopyright{acmcopyright}
\copyrightyear{2018}
\acmYear{2018}
\acmDOI{10.1145/1122445.1122456}

\acmJournal{JACM}
\acmVolume{37}
\acmNumber{4}
\acmArticle{}
\acmMonth{8}

\begin{document}
	
%	\title{Efficient Logspace Classes for Enumeration, Counting, and Uniform Generation}	
	\title{\review{\#NFA admits an FPRAS: Efficient Enumeration, Counting, and Uniform Generation for Logspace Classes}}
          
%	\title{Counting, Uniform Sampling and Enumeration of Strings in Regular Languages}

	\author{Marcelo Arenas}
	%\authornote{}
	\affiliation{%
		\institution{Pontificia Universidad Cat\'olica \& IMFD}
		\city{Santiago}
		\country{Chile}
	}
	\email{marenas@ing.puc.cl}
	
	\author{Luis Alberto Croquevielle}
	%\authornote{}
	\affiliation{%
		\institution{Pontificia Universidad Cat\'olica  \& IMFD}
		\city{Santiago}
		\country{Chile}
	}
	\email{lacroquevielle@uc.cl}
	
	\author{Rajesh Jayaram}
	%\authornote{}
	\affiliation{%
		\institution{Carnegie Mellon University}
		\city{Pittsburgh}
		\country{United States}
	}
	\email{rkjayara@cs.cmu.edu}
	
	\author{Cristian Riveros}
	%\authornote{}
	\affiliation{%
		\institution{Pontificia Universidad Cat\'olica \& IMFD}
		\city{Santiago}
		\country{Chile}
	}
	\email{cristian.riveros@uc.cl}

	\renewcommand{\shortauthors}{Arenas, Croquevielle, Jayaram, and Riveros}
	
	\begin{abstract}
		\input{abstract}

	\end{abstract}


%
% Keywords. The author(s) should pick words that accurately describe the work being
% presented. Separate the keywords with commas.

	\maketitle

	\section{Introduction}
	\label{sec-intro}
	\input{introduction}

	\section{Preliminaries}
	\label{preliminaries}
	\input{preliminaries}
	\section{NLOGSPACE transducers: definitions and our main results}
	\label{nlclass}
	\input{nlclass}

	\section{Applications of the Main Results}
	\label{applications}
	\input{applications}

	\section{Completeness, Self-reducibility, and their Implications for the Class $\rul$}
	\label{cde}
	\input{cde}

	\section{$\snfa$ Admits a Fully Polynomial-Time Randomized Approximation Scheme, and its  Implications to the Class $\rnl$}
	\label{approximate}
	\input{approximate-new}

	\section{Concluding Remarks}
	\label{conclusions}
	\input{conclusions}

	\bibliographystyle{ACM-Reference-Format}
	%ACM-Reference-Format}
	%\bibliography{../extras/references}
	\bibliography{references}

	\appendix
	
	\section{Proofs of Intermediate Results}
	\label{app-lemmas}
	\input{lemmas.tex}

\end{document}

%% file: macros.tex
\newcommand{\review}[1]{{#1}}
\newcommand{\change}[1]{{#1}}

\tikzset{
	rect/.style={
		rectangle,
		rounded corners,
		draw=black, 
		thick,
		text centered},
	rectw/.style={
		rectangle,
		rounded corners,
		draw=white, 
		thick,
		text centered},
	sq/.style={
		rectangle,
		draw=black, 
		thick,
		text centered},
	sqw/.style={
		rectangle,
		draw=white, 
		thick,
		text centered},
	arrout/.style={
		->,
		-latex,
		thick,
	},
	arrin/.style={
		<-,
		latex-,
		thick,
		El         },
	arrd/.style={
		<->,
		>=latex,
		thick,
	},
	arrw/.style={
		thick,
	}
}

\newtheorem{fact}{Fact}

\renewcommand{\epsilon}{\varepsilon}
\renewcommand{\phi}{\varphi}

\newcommand{\sem}[1]{\llbracket#1\rrbracket}

\newcommand{\kc}{\kappa}

\newcommand{\ENUM}{\text{{\rm ENUM}}}
\newcommand{\COUNT}{\text{{\rm COUNT}}}
\newcommand{\GEN}{\text{{\rm GEN}}}

\newcommand{\B}{\text{\tt B}}
\newcommand{\ex}[1]{\mathbb{E}\normalfont\lbrack #1 \rbrack}%expectation

\usepackage[framemethod=TikZ]{mdframed}
\newcounter{Frame}
\newenvironment{Frame}[1][htb]{%
\refstepcounter{Frame}
    \begin{mdframed}[%
        frametitle={#1},
        skipabove=\baselineskip plus 2pt minus 1pt,
        skipbelow=\baselineskip plus 2pt minus 1pt,
        linewidth=1.0pt,
        frametitlerule=true,
    ]%
}{%
    \end{mdframed}
}

\newcommand{\nl}{\text{\sc NL}}
\newcommand{\ul}{\text{\sc UL}}

\newcommand{\ptime}{\text{P}}

\newcommand{\cC}{\mathcal{C}}
\newcommand{\cR}{\mathcal{R}}
\newcommand{\cL}{\mathcal{L}}

\newcommand{\A}{\mathcal{A}}
\newcommand{\cA}{\mathcal{A}}
\newcommand{\cG}{\mathcal{G}}
\newcommand{\fail}{\text{\bf fail}}

\newcommand{\np}{\text{NP}}

\newcommand{\fp}{\text{\sc FP}}
\newcommand{\sharpp}{\#\text{P}}

\newcommand{\pr}{\text{\rm {\bf Pr}}}

\newcommand{\N}{\mathbb{N}} %TODO:Change this back!

\usetikzlibrary{arrows,positioning} 
\tikzset{
	circ/.style={
		circle,
		draw=black, 
		thick,
		text centered,
	},
	circw/.style={
		circle,
		draw=white, 
		thick,
		text centered,
	},
	arrout/.style={
		->,
		-latex,
		thick,
	},
	arrin/.style={
		<-,
		latex-,
		thick,
	},
	arrw/.style={
		-,
		thick,
	},
	arrww/.style={
		-,
		thick,
		draw=white, 
	}
}

\newcommand{\rnl}{\text{\sc Relation\nl}}
\newcommand{\rul}{\text{\sc Relation\ul}}
\newcommand{\nfa}{\text{\rm MEM-NFA}}
\newcommand{\unfa}{\text{\rm MEM-UFA}}
\newcommand{\snfa}{\text{\rm \#NFA}}
\newcommand{\spanl}{\text{\sc SpanL}}
\newcommand{\dnf}{\text{\rm SAT-DNF}}

\newcommand{\shp}{\text{\sc \#P}}
\newcommand{\up}{\text{\sc UP}}

\newcommand{\eps}{\delta}

\newcommand{\trans}[2][]{\raisebox{-1pt}[10pt][0pt]{$\overset{#2}{\underset{^{#1}}{\raisebox{0pt}[3pt][0pt]{$\relbar\mspace{-8mu}\rightarrow$}}}$}}
\newcommand{\Open}[1]{#1~\mkern-12mu\vdash}
\newcommand{\Close}[1]{\dashv~\mkern-10mu#1}
\newcommand{\mspan}[1]{[#1\rangle}
\newcommand{\xset}{\textbf{X}}

\newcommand{\evaleva}{\text{\rm EVAL-eVA}}
\newcommand{\evalueva}{\text{\rm EVAL-UeVA}}
\newcommand{\evalRPQ}{\text{\rm EVAL-RPQ}}

\newcommand{\var}{\operatorname{var}}
\newcommand{\lo}{\operatorname{lo}}
\newcommand{\hi}{\operatorname{hi}}

\newcommand{\evalOBDD}{\text{\rm EVAL-OBDD}}
\newcommand{\evalnOBDD}{\text{\rm EVAL-nOBDD}}

\newcommand{\nA}{A}
\newcommand{\nAun}{\nA_{\emph{unroll}}}
\newcommand{\nQ}{Q}
\newcommand{\nI}{I}
\newcommand{\nF}{F}
\newcommand{\neps}{\epsilon}
\newcommand{\nemptyword}{\lambda}
\newcommand{\ndelta}{\delta}
\newcommand{\nL}{\mathcal{L}}
\newcommand{\nsketch}{\operatorname{sketch}}

\newcommand{\cP}{\mathcal{P}}
\newcommand{\cD}{\mathcal{D}}

%% file: abstract.tex
% !TEX root = main.tex
% !TeX spellcheck = en_US

%A great deal of research in databases has been devoted to the development of efficient query evaluation algorithms. In particular, the development of efficient algorithms for enumerating the answers to a query has been widely investigated in the last few years.
In this work, we study two simple yet general complexity classes, based on logspace Turing machines, which provide a unifying framework for efficient query evaluation in areas like information extraction and graph databases, among others. 
%In this work, we introduce a unifying framework for a significant number of results in this area, which is based on the definition of  two simple complexity classes in terms of logspace Turing Machines.
We investigate the complexity of three fundamental algorithmic problems for these classes: enumeration, counting and uniform generation of solutions, and show that they have several desirable properties in this respect.

Both complexity classes are defined in terms of non-deterministic logspace transducers (NL-transducers). 
For the first class, we consider the case of unambiguous NL-transducers, and we 
%are able to achieve 
prove constant delay enumeration, and both counting and uniform generation of solutions in polynomial time. For the second class, we consider unrestricted NL-transducers, and we obtain polynomial delay enumeration, approximate counting in polynomial time, and polynomial-time randomized algorithms for uniform generation. More specifically, we show that each problem in this second class admits a fully polynomial-time randomized approximation scheme (FPRAS) and a polynomial-time Las Vegas algorithm (with preprocessing) for uniform generation. Remarkably, the key idea to prove these results is to show that the fundamental problem $\snfa$ admits an FPRAS, where $\snfa$ is the problem of counting the number of strings of length $n$ (given in unary) accepted by a non-deterministic finite automaton (NFA).
While this problem is known to be $\shp$-complete and, more precisely, $\spanl$-complete, it was open whether this problem admits an FPRAS. In this work, we solve this open problem, and obtain as a welcome corollary that every function in $\spanl$  admits an FPRAS.

%% file: introduction.tex
% !TEX root = main.tex
% !TeX spellcheck = en_US

%A fundamental problem in the theory of database systems is to understand the computational complexity of potential queries made to the system. 
%As the quantity of data in these systems becomes enormously large, as in many modern applications, even the simplest of queries to retrieve information from a database can become prohibitively expensive to compute. It is crucial, therefore, to classify and describe the classes of queries that admit efficient algorithms to resolve or approximate. In this work, we make a step towards this classification, by introducing two complexity classes, based on log-space Turing machines, which provide a unifying framework for query evaluation in database systems. Our definitions are elementary and natural, but still capture substantially expressive queries.

Arguably, query answering is the most fundamental problem in databases. In this respect, developing 
efficient query answering algorithms, as well as understanding when this cannot be done, is of paramount importance \review{for database theory and applications}. In the most classical view of this problem, one is interested in computing all the answers, or solutions, to a query. However, as the quantity of data becomes enormously large, the number of answers to a query could also be enormous, so computing the complete set of solutions can be prohibitively expensive. In order to overcome this limitation, the idea of enumerating the answers to a query with a {\em small delay} has been recently studied in the database area \cite{segoufin2013enumerating}. More specifically, the idea is to divide the computation of the answers to a query into two phases. In a {\em preprocessing} phase, some data structures are constructed to accelerate the process of computing answers. Then in an {\em enumeration} phase, the answers are enumerated with a small delay between them. In particular, in the case of constant delay enumeration algorithms, the preprocessing phase should take polynomial time, while the time between consecutive answers should be constant.

Constant delay enumeration algorithms allow users to retrieve a fixed number of answers very efficiently, which can give them a lot of information about the solutions to a query. In fact, the same holds if users need a linear or a polynomial number of answers. However, because of the data structures used in the preprocessing phase, these algorithms usually return answers that are very similar to each other \cite{bagan2007acyclic,segoufin2013enumerating,florenzano2018constant}; for example, tuples with $n$ elements where only the first few coordinates are changed in the first answers that are returned. In this respect, other approaches can be used to return some solutions efficiently but improving \review{their variety}. Most notably, the possibility of generating an answer uniformly, at random, is a desirable condition if it can be done efficiently. Notice that returning varied solutions has been identified as an important property not only in databases, but also for algorithms that retrieve information in a broader sense~\cite{AMSW16}.

Efficient algorithms for either enumerating or uniformly generating the answers to a query are powerful tools to help in the process of understanding the answers to a query. But how can we know how long these algorithms should run, and how complete the set of computed answers is? A third tool that is needed then is an efficient algorithm for computing, or estimating, the number of solutions to a query.  %Enumerating, counting and uniformly generating solutions are then important tools when confronting to the problem of answering a query. 
Then, taken together, enumeration, counting and uniform generation techniques form a powerful attacking trident when \review{confronting the problem} of answering a~query.

In this paper, we follow \review{a principled approach} to study the problems of enumerating, counting and uniformly generating the answers to a query. More specifically, we begin by following the guidance of \cite{jerrum1986random}, which urges the use of relations to formalize the notion of solution to a given input of a problem (for instance, to formalize the notion of answer to an input query over an input database). While there are many ways of formalizing this notion, most such formalizations only make sense for a specific kind of queries, e.g. a subset of the integers is well-suited as the solution set for counting problems, but not for sampling problems. \review{We want a general framework, so by following \cite{jerrum1986random}}, we represent a problem as a relation \review{$R \subseteq \{0,1\}^* \times \{0,1\}^*$, and we say that $y$ is a solution for an input $x$ if $(x,y) \in R$.}\footnote{\review{For the sake of presentation, we assume relations to be defined over the binary alphabet $\{0,1\}$. The results of this article also hold if we consider relations defined over an arbitrary finite alphabet $\Sigma$.}} Note that the problem of enumerating the solutions to a given input $x$ corresponds to the problem of enumerating the elements of the set $\{ y \in \{0,1\}^* \mid (x,y) \in R\}$, while the counting and uniform generation problems correspond to the problems of computing the cardinality of $\{ y \in \{0,1\}^* \mid (x,y) \in R\}$ and uniformly generating, at random, a string in this set,~respectively. 

Second, we study two simple yet general complexity classes for relations, based on non-deter\-ministic logspace transducers (NL-transducers), which provide a unifying framework for studying enumeration, counting and uniform generation. More specifically,
%given a finite alphabet $\Sigma$,
an NL-transducer $M$ is a \review{non-deterministic} Turing Machine with input and output alphabet $\{0,1\}$, a read-only input tape, a write-only output tape and a work-tape of which, on input $x \in \{0,1\}^*$, only the first $O(\log(|x|))$ cells can be used. Moreover, a string $y \in \{0,1\}^*$ is said to be an output of $M$ on input $x$, if there exists a run of $M$ on input $x$ that halts in an accepting state with $y$ as the string in the output tape. Finally, assuming that all outputs of $M$ on input $x$ are denoted by $M(x)$, a relation $R \subseteq \{0,1\}^* \times \{0,1\}^*$ is said to be accepted by $M$ if for every input $x$, it holds that $M(x) = \{ y \in \{0,1\}^* \mid (x,y) \in R\}$.

The first complexity class of relations studied in this paper consists of the relations accepted by unambiguous NL-transducers.
More precisely, an NL-transducer $M$ is said to be unambiguous if for every input $x$ and $y \in M(x)$, there exists exactly one run of $M$ on input $x$ that halts in an accepting state with $y$ as the string in the output tape. For this class, we are able to achieve constant delay enumeration, and both counting and uniform generation of solutions in polynomial time. For the second class, we consider (unrestricted) NL-transducers, and we obtain polynomial delay enumeration, approximate counting in polynomial time, and polynomial-time randomized algorithms for uniform generation. More specifically, we show that each problem in this second class admits a fully polynomial-time randomized approximation scheme (FPRAS)~\cite{jerrum1986random}  and a polynomial-time Las Vegas algorithm (with preprocessing) for uniform generation. It is important to mention that the key idea to prove these results is to show that the fundamental problem $\snfa$ admits an FPRAS, where $\snfa$ is the problem of counting the number of strings of length $n$ (given in unary) accepted by a non-deterministic finite automaton (NFA). While this problem is known to be $\shp$-complete and, more precisely, $\spanl$-complete~\cite{alvarez1993very}, it was open whether %this problem 
it admits an FPRAS, and only quasi-polynomial time randomized approximation schemes (QPRAS) were known for it \cite{KSM95,gore1997quasi}. In this work, we solve this open problem, and obtain as a welcome corollary that every function in $\spanl$ admits an FPRAS. Thus, to the best of our knowledge, \review{we identify $\spanl$ as the first} complexity class with a simple and robust definition based on Turing Machines, that contains $\shp$-complete problems and where each problem admits an FPRAS.

\medskip
\noindent {\bf Proviso.} This paper is an extended version of the article {\em ``Efficient Logspace Classes for Enumeration, Counting, and Uniform Generation''} that was published in the 38$^\text{th}$ ACM SIGMOD-SIGACT-SIGAI Symposium on Principles of Database Systems (PODS 2019). For this extended version, we have made many changes. In particular, we have completely reworked the proof that $\snfa$ admits a fully polynomial-time randomized approximation schema, which is the main result of this work. More specifically, we have carefully \review{restructured this proof} for the sake of readability, obtaining a simpler and more efficient algorithm. 
Besides, for the more general class of relations defined in terms of (unrestricted) NL-transducers, this result allows us to prove that each problem in this class admits a polynomial-time Las Vegas uniform generator. Such a randomized algorithm \review{needs a preprocessing phase}, which was not properly made explicit in the conference paper. We have solved this issue by introducing, and studying, the notion of preprocessing polynomial-time Las Vegas uniform generator. 

\medskip
\noindent {\bf Organization of the paper.} The main terminology used in the paper is given in Section \ref{preliminaries}.
In Section \ref{nlclass}, we define the two classes studied in this paper and state our main results.
In Section \ref{applications}, we show how these classes can be used to obtain positive results on query evaluation in information extraction, graph databases, and binary decision diagrams. 
The complete proofs of our results are presented in Sections \ref{cde} and~\ref{approximate}, and Appendix \ref{app-lemmas}. In particular, we explain the algorithmic techniques used to obtain an FPRAS for the $\snfa$ problem in Section \ref{approximate}, where we also provide a detailed proof of this result. 
Finally, some concluding remarks are given in Section \ref{conclusions}.
%Due to the lack of space, the complete proofs of these results are given in the Appendix. 

%% file: preliminaries.tex
%!TEX root = main.tex

Given natural numbers $n \leq m$, we use notation $[n,m]$ for the set $\{n, \ldots, m\}$. Beside, we use $\log(x)$ to refer to the logarithm of $x$ to base $e$. %In this section, we assume that $\Sigma$ is a finite alphabet. 

\subsection{Relations and problems} 
%\noindent \textbf{Relations and problems.} 
%Let $\Sigma$ be a finite alphabet with at least two symbols. 
As usual, $\{0,1\}^*$ denotes the set of all strings over the binary alphabet $\{0,1\}$, $|x|$ denotes the length of a string $x \in \{0,1\}^*$, $x_1 \cdot x_2$ denotes the concatenation of two strings $x_1,x_2 \in \{0,1\}^*$, and $\{0,1\}^n$ denotes the set of all strings $x \in \{0,1\}^*$ such that $|x| = n$.
%As usual, we represent inputs as words $x \in \{0,1\}^*$ and the length of $x$ is denoted by $|x|$.
A problem is represented as a relation $R\subseteq \{0,1\}^*\times\{0,1\}^*$. For every pair $(x,y)\in R$, we interpret $x$ as being the encoding of an input to some problem, and $y$ as being the encoding of a \review{solution
%or witness
to that input}. For each $x\in\{0,1\}^*$, we define the set
%\begin{equation*}
%W_R(x) \ = \ \{y\in\{0,1\}^* \mid (x,y)\in R\},
%\end{equation*}
$W_R(x) = \{y\in\{0,1\}^* \mid (x,y)\in R\}$,
\review{and call it the
%witness
set of solutions for $x$. Also, if $y\in W_R(x)$, we call $y$ a
%witness or
solution} to $x$.

This is a very general framework, so 
%mostly we work with relations that meet two additional properties.
%First, we only work with relations where both the input and the witnesses have a finite encoding. Second, 
we work with $p$-relations \cite{jerrum1986random}. Formally, a relation $R \subseteq \{0,1\}^*\times\{0,1\}^*$ is a $p$-relation if (1) there exists a polynomial $q$ such that $(x,y)\in R$ implies that $|y|\leq q(|x|)$ and (2) there exists a deterministic Turing Machine that receives as input $(x,y)\in\{0,1\}^*\times\{0,1\}^*$, runs in polynomial time and accepts if, and only if, $(x,y)\in R$.
Without loss of generality, from now on we assume that for a $p$-relation $R$, there exists a polynomial $q$ such that $|y| = q(|x|)$ for every $(x,y) \in R$. 
%Notice that this implies that for every $x,y_1,y_2 \in \{0,1\}^*$ such that $(x,y_1) \in R$ and $(x,y_2) \in R$, it holds that $|y_1|=|y_2|$. 
This is not a strong requirement, \review{since all
%witnesses
solutions} can be made to have the same length through~padding.

\subsection{Enumeration, counting and uniform generation}
%\smallskip
%\noindent \textbf{Enumeration, counting and uniform generation.} 
%There are several computational problems associated to a relation $R$. For example, given a relation $R$ and an input $x\in\{0,1\}^*$, we could consider the existence problem of deciding whether there are any witnesses for $x$. 
%In this paper, 
There are several computational problems associated to a relation $R$. For example, given a relation $R$ and an input $x\in\{0,1\}^*$, we could consider the existence problem of deciding whether \review{there are any
%witnesses
solutions} for $x$. 
In this paper,  given a $p$-relation $R$, we are interested in the following problems:
%Given a $p$-relation $R$, we are interested in the following problems:
%We are interested in the problems related to enumeration, counting and uniform generation of solutions. Formally, given a $p$-relation $R$, we define the problems: 
\begin{center}
	\framebox{
		\begin{tabular}{p{1.75cm} p{8.25cm}}
			\textbf{Problem:} & $\ENUM(R)$\\
			\textbf{Input:} &  A word $x \in \{0,1\}^*$ \\
			\textbf{Output:} & Enumerate all $y \in W_R(x)$ 
			without repetitions
		\end{tabular}
	}
\end{center}
\begin{center}
	\framebox{
		\begin{tabular}{p{1.75cm} p{8.25cm}}
			\textbf{Problem:} & $\COUNT(R)$ \\
			\textbf{Input:} &  A word $x \in \{0,1\}^*$ \\
			\textbf{Output:} & The size $|W_R(x)|$
		\end{tabular}
	}
\end{center}
\begin{center}
	\framebox{
		\begin{tabular}{p{1.75cm} p{8.25cm}}
			\textbf{Problem:} & $\GEN(R)$ \\
			\textbf{Input:} &  A word $x \in \{0,1\}^*$\\
			\textbf{Output:} &  Generate uniformly, at random, a word in $W_R(x)$
		\end{tabular}
	}
\end{center}
Given that $|y| = q(|x|)$ for every $(x,y) \in R$, we have that $W_R(x)$ is finite and these three problems are well defined. 
Notice that in the case of $\ENUM(R)$, we do not assume a specific order on words, so that \review{the elements} of $W_R(x)$ can be enumerated in any order (but without repetitions). Moreover, in the case of $\COUNT(R)$, we assume that $|W_R(x)|$ is encoded in binary and, therefore, the size of the output is logarithmic in the size of $W_R(x)$. Finally, in the case of $\GEN(R)$, we generate a word $y \in W_R(x)$ with probability $\frac{1}{|W_R(x)|}$ if the set $W_R(x)$ is not empty; otherwise, we return a special symbol $\bot$ to indicate that $W_R(x) = \emptyset$. 

%\marcelo{We should include Valiant's general perspective of problems as relations, as this gives us a good justification for the problems above.}
%\alberto{Maybe we should make clear that in our case $\GEN(R)$ won't have probabilities $\frac{1}{|%W_R(x)|}$ (maybe omit that part from the definition?)}

\subsection{Enumeration with polynomial and constant delay}
%\smallskip
%\noindent{\bf Enumeration with polynomial and constant delay.} 
An enumeration algorithm for $\ENUM(R)$ is a procedure that receives an input $x \in \{0,1\}^*$ and, during the computation, it outputs each word in $W_R(x)$, one by one and without repetitions. The time between two consecutive outputs is called the delay of the enumeration. In this paper, we consider two restrictions on the delay: \review{polynomial delay and constant delay}. \review{\emph{Polynomial delay enumeration}} is the standard notion of polynomial time efficiency in enumeration algorithms~\cite{johnson1988generating} and is defined as follows. An enumeration algorithm is of polynomial delay if there exists a polynomial $p$ such that for every input $x \in \{0,1\}^*$, the time between the beginning of the algorithm and the initial output, between any two consecutive outputs, and between the last output and the end of the algorithm, is bounded by $p(|x|)$.

\review{\emph{Constant delay enumeration}} is another notion of efficiency for enumeration algorithms that has attracted a lot attention in the recent years~\cite{bagan2006mso,courcelle2009linear,segoufin2013enumerating}. 
This notion has stronger guarantees compared to polynomial delay: the enumeration is done in a second phase after the processing of the input and taking constant-time \review{between two consecutive} outputs in a very precise sense. 
Several notions of \review{constant delay enumeration} has been given, most of them in database theory where it is important to divide the analysis between query and data. In this paper, we want a definition of \review{constant delay} that is agnostic of the distinction between query and data  (i.e. combined complexity) and, for this reason, we use a more general notion of \review{constant delay} enumeration than the one in~\cite{bagan2006mso,courcelle2009linear,segoufin2013enumerating}.

\review{Constant delay enumeration cannot be achieved in general with a standard Turing Machine, because merely moving the head through the tape will take up more than constant time.} So, as it is standard in the literature~\cite{segoufin2013enumerating}, for the notion of \review{constant delay enumeration} we consider enumeration algorithms on Random Access Machines (RAM) with addition and uniform cost measure~\cite{aho1974design}. 
Given a relation $R \subseteq \{0,1\}^* \times \{0,1\}^*$, an enumeration algorithm~$E$ for $R$ has constant delay if $E$ runs in two phases over the input $x$.
\begin{enumerate}
	\item The first phase (precomputation), which does not produce output. 
	\item The second phase (enumeration), which occurs
	immediately after the precomputation phase, where all words in $W_R(x)$ are enumerated without repetitions and satisfying the following conditions, for a fixed constant $c$:
	\begin{enumerate}
		\item \label{cd-1} the time it takes to generate the first output $y$ is bounded by $c \cdot |y|$;
		\item \label{cd-2} the time between two consecutive outputs $y$ and $y'$ is bounded by $c \cdot |y'|$ and does not depend on $y$; and
		\item \label{cd-3} the time between the final element $y$ that is returned and the
		end of the enumeration phase is bounded by $c \cdot |y|$,
	\end{enumerate}
	%for a constant $c$ that is independent of the input $x$.
	%The independence on $x$ in all of these steps is what makes the delay constant.
	% the time taken to generate the first output and from the last
	%            output to the end of this phase is constant in $|\gamma|$ and $|d|$.
	%            Furthermore, we require that the time between two consecutive outputs $o$ and $o'$ is linear in $|o'|$ and does not depend on the size of $\gamma$ or $d$.
\end{enumerate}
We say that $E$ is a constant delay algorithm for $R$ with precomputation phase $f$, if $E$ has constant delay and the precomputation phase takes time $O(f(|x|))$. Moreover, we say that $\ENUM(R)$ can be solved with constant delay 
%enumerated with constant delay, 
if there exists a \review{constant delay} algorithm for $R$ with precomputation phase $p$ for some polynomial~$p$. 

Our notion of \review{constant delay} algorithm differ from the definitions in~\cite{segoufin2013enumerating} in two aspects. 
\review{First, in our definition the input is not divided into some components, so the preprocessing phase must take polynomial time on the size of the entire input. In the case of constant delay algorithms for query answering, the input is usually divided into the data and the query, and the preprocessing phase is only asked to take polynomial time on the size of the data (as the query is usually assumed to be fixed, which is referred to as data complexity~\cite{V82}).}
%we relax the distinction between query and data in the preprocessing phase, allowing our algorithm to take polynomial time in the input (i.e. combined complexity).
Second, our definition of \review{constant delay} is what in \cite{courcelle2009linear,bagan2006mso} is called {\em linear delay in the size of the output}, namely, writing the next output is linear in its size and does not depend on the size of the input. This is a natural assumption, since each output must at least be written down to return it to the user. 
Notice that, given an input $x$ and an output $y$, the notion of \review{polynomial delay} above means polynomial in $|x|$ and, instead, the notion of linear delay from~\cite{courcelle2009linear,bagan2006mso} means linear in $|y|$, i.e., constant in the size of $|x|$.
Thus, we have decided to call the two-phase enumeration from above \review{``constant delay''}, as it does not depend on the size of the input $x$, and the delay is just what is needed to write the output (which is the minimum requirement for such an enumeration algorithm). 

\subsection{Approximate counting and Las Vegas uniform generation with preprocessing}
%\smallskip
%\noindent \textbf{Approximate counting.} 
Given a relation $R \subseteq \{0,1\}^* \times \{0,1\}^*$, the problem $\COUNT(R)$ can be solved efficiently if there exists a polynomial-time algorithm that, given $x \in \{0,1\}^*$, computes $|W_R(x)|$. In other words, if we think of $\COUNT(R)$ as a function that maps $x$ to the value $|W_R(x)|$, then $\COUNT(R)$ can be computed efficiently if $\COUNT(R) \in \fp$, the class of functions that can be computed in polynomial time. As such a condition does not hold for many fundamental problems, we also consider the possibility of efficiently approximating the value of the function $\COUNT(R)$. More precisely, $\COUNT(R)$ is said to admit a fully polynomial-time randomized approximation scheme (FPRAS)~\cite{jerrum1986random} if there exists a randomized algorithm 
$\A : \{0,1\}^* \times (0,1) \to  \N$ and a polynomial $q(u,v)$ such that for every $x \in \{0,1\}^*$ and $\eps \in (0,1)$, it holds that:
\begin{eqnarray*}
\pr(|\A(x,\eps) - |W_R(x)|| \leq \eps \cdot |W_R(x)|) & \geq & \frac{3}{4}
\end{eqnarray*}
and \review{the time needed} to compute $\A(x,\eps)$ is at most $q(|x|,\frac{1}{\eps})$. Thus, $\A(x, \eps)$ approximates the value $|W_R(x)|$ with a relative error of $\eps$, and it can be computed in polynomial time in the size of $x$ and the value $\frac{1}{\eps}$. 

%\smallskip
%\noindent \textbf{Las Vegas uniform generation.} 
The problem $\GEN(R)$ can be solved efficiently if there exists a polynomial-time randomized algorithm that, given $x \in \{0,1\}^*$, generates  an element of $W_R(x)$  with uniform probability distribution (if $W_R(x) = \emptyset$, then it returns~$\bot$). However, as in the case of $\COUNT(R)$, the existence of such a generator is not guaranteed for many fundamental problems, so we also consider a relaxed notion of generation that has a probability of failing in returning a solution. 
More precisely, $\GEN(R)$ is said to admit a preprocessing polynomial-time Las Vegas uniform generator (PPLVUG) if there exists a pair of randomized algorithms $\cP : \{0,1\}^* \times (0,1) \to  (\{0,1\}^* \cup \{ \bot\})$, $\cG : \{0,1\}^* \to  (\{0,1\}^* \cup \{\fail\})$ and a pair of polynomials $q(u,v)$, $r(u)$ such that for every~$x \in \{0,1\}^*$ and $\delta \in (0,1)$:
\begin{enumerate}
\item \label{pplvug-1} The preprocessing algorithm $\cP$ receives as inputs $x$ and $\delta$, \review{and runs in time bounded by} $q(|x|,  \log(1/\delta))$. If $W_R(x) \neq \emptyset$, then $\cP(x,\delta)$ returns a string $\cD$ such that $\cD$ is \emph{good-for-generation} with probability $1-\delta$. If $W_R(x) = \emptyset$, then $\cP(x,\delta)$ returns $\bot$.

\item \label{pplvug-2} The generator algorithm $\cG$ receives as input $\cD$ \review{and runs in time bounded by} $r(|\cD|)$. Moreover, if $\cD$ is good-for-generation then:
\begin{enumerate}
\item \label{pplvug-2-1} $\cG(\cD)$ returns $\fail$ with a probability at most $\frac{1}{2}$, and

\item \label{pplvug-2-2} conditioned on not returning $\fail$, $\cG(\cD)$ returns a truly uniform sample $y \in W_R(x)$, i.e. with a probability $1/|W_R(x)|$ for each $y \in W_R(x)$.
\end{enumerate}
Otherwise, if $\cD$ is not good-for-generation, then $\cG(\cD)$ outputs a string without any guarantee.
\end{enumerate}

In line with the notion of \review{constant delay enumeration} algorithm, we allow the previous concept of uniform generator to have a preprocessing phase. If \review{there is no
%witness
solution} for the input $x$ (that is, $W_R(x) = \emptyset$), then the preprocessing algorithm $\cP$ returns the symbol $\bot$. Otherwise, the invocation $\cP(x,\delta)$ returns a string  $\cD$ in $\{0,1\}^*$, namely, a data structure or ``advice'' for the generation procedure~$\cG$. The output of the invocation $\cP(x,\delta)$ is used by the generator algorithm $\cG$ to \review{produce a
%witness
solution} of $x$ with uniform distribution (that is, with probability $1/|W_R(x)|$). If the output of $\cP(x,\delta)$ is not good-for-generation (which occurs with probability $\delta$), then we have no guarantees on the output of the generator algorithm $\cG$. Otherwise, we know that $\cG(\cD)$ returns an element of $W_R(x)$ with uniform distribution, or it returns $\fail$. Furthermore, we can repeat $\cG(\cD)$ as many times as needed, generating each time a truly uniform sample $y$ from $W_R(x)$ whenever $y \neq \fail$.
   
%The invocation $\cG(x)$ can fail in generating an element of $W_R(x)$, in which case it returns $\fail$. 
Notice that by condition (\ref{pplvug-2-1}), we know that this probability of failing is smaller than $\frac{1}{2}$, so that by invoking $\cG(\cD)$ several times we can make this probability arbitrarily small (for example, the probability that $\cG(\cD)$ returns $\fail$ in 1000 consecutive independent invocations is at most $(\frac{1}{2})^{1000}$). 
%Assume that the invocation $\cG(x)$ does not fail. If $W_R(x) = \emptyset$, then we have by condition 3~(a) that $\cG(x) = \bot$, so the randomized algorithm indicates that there is no witness for $x$ in this case. If $W_R(x) \neq \emptyset$, then we have by conditions (\ref{fpaug-2}) and (\ref{fpaug-3}) that $\cG(x)$ returns an element $y \in W_R(x)$. Moreover, we know by condition 3~(b) that the probability of returning such an element $y$ is $\phi(x)$. Thus, we have a uniform generator in this case, as the probability of returning each element $y \in W_R(x)$ is the same. Finally, we have that $\cG(x)$ can be computed in polynomial time in the size of $x$.
Moreover, we have that $\cP(x,\delta)$ can be computed in time $q(|x|,  \log(1/\delta))$, so we can consider an exponentially small value of $\delta$ such as 
\begin{eqnarray*}
\frac{1}{2^{|x|+1000}}, 
\end{eqnarray*}
and still obtain that $\cP(x,\delta)$ can be computed in time polynomial in $|x|$. Notice that with such a value of $\delta$, the probability of producing a  good-for-generation string $\cD$ is at least 
\begin{eqnarray*}
1 - \frac{1}{2^{1000}},
\end{eqnarray*}
which is an extremely high probability . 
Finally, it is important to notice that the size of $\cD$ is at most $q(|x|,  \log(1/\delta))$, so that $\cG(\cD)$ can be computed in time polynomial in $|x|$ and $\log(1/\delta)$. Therefore, $\cG(\cD)$ can be computed in time polynomial in $|x|$ even if we consider an exponentially small value for $\delta$ such as $1/2^{|x|+1000}$.

%It is important to notice that the notion of polynomial-time Las Vegas uniform generator corresponds to the notion of uniform generator used in~\cite{jerrum1986random}. However, we have decided to use the term ``Las Vegas'' to emphasize the fact that there is a probability of failing in returning a solution. Moreover, 
It is important to notice that the notion of preprocessing polynomial-time Las Vegas uniform generator imposes stronger requirements than the notion of fully polynomial-time almost uniform generator introduced in~\cite{jerrum1986random}. In particular, the latter not only has a probability of failing, but also considers the possibility of generating a solution with a probability distribution that is {\em almost} uniform, that is, an algorithm that generates an string $y \in W_R(x)$ with a probability in an interval $[1/|W_R(x)| - \epsilon, 1/|W_R(x)| + \epsilon]$ for a given error $\epsilon \in (0,1)$.

%% file: nlclass.tex
%!TEX root = main.tex

%\begin{itemize}
%	\item Define here the definition of NL transducers and how to use for enumeration.
%	\item Give the connection with counting and SpanL.
%	\item Some examples (i.e. enumeration and counting of DNF) and discussion.
%\end{itemize}

The goal of this section is to provide simple yet general definitions of classes of relations with good properties in terms of enumeration, counting and uniform generation. 
%Our first goal is to provide a simple yet general definition of a class of relations with good properties in terms of enumeration. 
More precisely, we are first aiming at providing a class $\cC$ of relations  that has a simple definition in terms of Turing Machines and 
%(which is a simpler model than the ERAM) 
such that for every relation $R \in \cC$, it holds that $\ENUM(R)$ can be solved with constant delay, and both $\COUNT(R)$ and $\GEN(R)$ can be solved in polynomial time.
%the solutions for a given input to $R$ can be enumerated with constants delay, and can be counted and uniformly sampled in polynomial time. 
Moreover, as it is well known that such good conditions cannot always be achieved, we are then aiming at extending the definition of $\cC$ to obtain a simple class, also defined in terms of Turing Machines and with good approximation properties. 
%admits a constant delay enumeration algorithm, that is, for every $R \in \cC$, it holds that $R \in \delayc$.  
%Formally, we want to define a class $\mathcal{C}$ such that for all $R\in\mathcal{C}$ we have that $R\in\delayc$. 
It is important to mention that we are not looking for an exact characterization in terms of Turing Machines of the class of relations that admit constant delay enumeration algorithms, as this may result in an overly complicated model. Instead, we are looking for simple yet general classes of relations with good properties in terms of enumeration, counting and uniform generation, and which can serve as a starting point for the systematic study of these three fundamental properties. 
%We are not looking for $R\in\mathcal{C}$ to be a necessary condition for the existence of a constant delay enumeration algorithm, and we don't want an overly complex definition for $\mathcal{C}$ yet, so we will start off with a simple model, and then continue from there.

A key notion that is used in our definitions of classes of relations is that of transducer.
%Given a finite alphabet $\{0,1\}$,
An $\nl$-transducer $M$ is a \review{non-deterministic} Turing Machine with input and output alphabet $\{0,1\}$, a read-only input tape, a write-only output tape where the head is always moved to the right once a symbol is written in it (so that the output cannot be read by $M$), and a work-tape of which, on input $x$, only the first $f(|x|)$ cells can be used, where $f(n) \in O(\log(n))$. A string $y \in \{0,1\}^*$ is said to be an output of $M$ on input $x$, if there exists a run of $M$ on input $x$ that halts in an accepting state with $y$ as the string in the output tape. The set of all outputs of $M$ on input $x$ is denoted by $M(x)$ (notice that $M(x)$ can be empty). Finally, the relation accepted by $M$, denoted by $\cR(M)$, is defined as $\{ (x,y) \in \{0,1\}^*\times\{0,1\}^* \mid y \in M(x)\}$.
%\alberto{Maybe offer some justification as to why we are considering this kind of model? Why NL-transducers?}
%\cristian{I don't believe that we need justification. We propose this, and it has good properties. These are enough justifications for me :).}

\begin{definition}
A relation $R$ is in $\rnl$ if, and only if, there exists an $\nl$-transducer $M$ such that $\cR(M) = R$.
\end{definition}
\review{Cycles are forbidden in $\nl$-transducers to have polynomial-size solutions for each input~\cite{alvarez1993very}. However, we do not need to impose this restriction here as we only work with $p$-relations in this paper (see Section~\ref{preliminaries}).}

The class $\rnl$ should be general enough to contain some natural and
well-studied problems. A first such a problem is the satisfiability of
a propositional formula in DNF. As a relation, this problem can be
represented as follows:
\begin{align*}
\dnf \ = \ \{ (\varphi,\sigma) \mid \varphi \text{ is a \review{propositional formula} in DNF}, \sigma \text{ is a truth assignment and } \sigma(\varphi) = 1 \}.
\end{align*}
Thus, we have that $\ENUM(\dnf)$ corresponds to the problem of enumerating the truth assignments satisfying a propositional formula $\varphi$ in DNF, while $\COUNT(\dnf)$ and $\GEN(\dnf)$ correspond to the problems of counting and uniformly generating such truth assignments, respectively. It is not difficult to see that 
%holds that 
$\dnf \in \rnl$. 
In fact, assume that we are given a propositional formula $\varphi$ of the form $D_1 \vee \cdots \vee D_m$, where each $D_i$ is a conjunction of literals, that is, a conjunction of propositional variables and negation of propositional variables. Moreover, assume that each propositional variable in $\varphi$ is of the form $x\_k$, where $k$ is a binary number, and that $x\_1$, $\ldots$, $x\_n$ are the variables occurring in $\varphi$. Notice that with such a representation, we have that $\varphi$ is a string over the alphabet $\{x,\_,0,1,\wedge,\vee,\neg\}$.\footnote{\review{For the sake of presentation, we consider a non-binary alphabet in this case, although it is easy to see how $\dnf$ can be represented by using the binary alphabet.}}  We define as follows an $\nl$-transducer $M$ such that $M(\varphi)$ is the set of truth assignments satisfying $\varphi$. On input $\varphi$, the $\nl$-transducer $M$ non-deterministically chooses a disjunct $D_i$, which is represented by two indexes indicating the starting and ending symbols of $D_i$ in the string $\varphi$. Then it checks whether $D_i$ is satisfiable, that is, whether $D_i$ does not contain complementary literals. Notice that this can be done in logarithmic space by checking for every $j \in \{1, \ldots, n\}$, whether $x\_j$ and $\neg x\_j$ are both literals in $D_i$. If $D_i$ is not satisfiable, then $M$ halts in a non-accepting state. Otherwise, $M$ returns a satisfying truth assignment of $D_i$ as follows. A truth assignment for $\varphi$ is represented by a string of length $n$ over the alphabet $\{0,1\}$, where the $j$-th symbol of this string is the truth value assigned to variable $x\_j$. Then for every $j \in \{1, \ldots, n\}$, if $x\_j$ is a conjunct in $D_i$, \review{then $M$ writes} the symbol 1 in the output tape, and if $\neg x\_j$ is a conjunct in $D_i$, \review{then $M$ writes} the symbol 0 in the output tape. Finally if neither $x\_j$ nor $\neg x\_j$ is a conjunct in $D_i$, then $M$ non-deterministically chooses a symbol $b \in \{0,1\}$, and it writes $b$ in the output tape. 
%\cristian{This looks nice and very pedagogic, but we should shorten it a bit.}

Given that $\COUNT(\dnf)$ is a \review{$\shp$-complete problem \cite{PB83}},
%It is well-known that $\COUNT(\dnf)$ is a $\shp$-complete problem \cite{DHK05}, so 
we cannot expect $\COUNT(R)$ to be solvable in polynomial time for every $R \in \rnl$. However, $\COUNT(\dnf)$ admits an FPRAS \cite{KL83}, so we can still hope for $\COUNT(R)$ to admit an FPRAS for every $R \in \rnl$. It turns out that proving such a result involves providing an FPRAS for another natural and fundamental problem: $\snfa$. More specifically, $\snfa$ is the problem of counting the number of words of length $k$ accepted by a non-deterministic finite automaton without epsilon transitions (NFA), where $k$ is given in unary (that is, $k$ is given as a string $0^k$). It is known that $\snfa$ is $\shp$-complete~\cite{alvarez1993very}, but it is open whether it admits an FPRAS; in fact, the best randomized approximation scheme known for $\snfa$ runs in time $n^{O(\log(n))}$~\cite{KSM95}. In our notation, this problem is represented by the following relation:
\begin{align*}
\nfa = \{ ((\nA,0^k), w) \mid \nA \text{ is an NFA with alphabet } \{0,1\},
 w \in \{0,1\}^k \text{ and } w \text{ is accepted by } \nA\},
\end{align*}
that is, we have that $\snfa = \COUNT(\nfa)$. It is easy to see that $\nfa \in \rnl$. Hence, we give a positive answer to the open question of whether $\snfa$ admits an FPRAS by proving the following general result about $\rnl$.
\begin{theorem}\label{theo-rnl}
If $R \in \rnl$, then $\ENUM(R)$ can be solved with polynomial delay, $\COUNT(R)$ admits an FPRAS, and $\GEN(R)$ admits a PPLVUG.
\end{theorem}
It is worth mentioning a fundamental consequence of this result in computational complexity. The \review{class of functions} $\spanl$ was introduced in \cite{alvarez1993very} to provide a characterization of some functions that are hard to compute. More specifically,
%given a finite alphabet $\{0,1\}$,
a function $f : \{0,1\}^* \to \N$ is in $\spanl$ if there exists an $\nl$-transducer $M$ with input alphabet $\{0,1\}$ such that $f(x) = |M(x)|$ for every $x \in \{0,1\}^*$. 
The complexity class $\spanl$ is contained in $\shp$, and it is a hard class in the sense that if $\spanl \subseteq \fp$, 
%every function in $\spanl$ can be computed in polynomial time, 
then $\ptime = \np$~\cite{alvarez1993very}, where $\fp$ is the class of functions that can be computed in polynomial time. In fact, $\spanl$ 
%and it 
has been instrumental in proving that some functions are difficult to compute~\cite{alvarez1993very,HV95,ACP12,LM13}.
It is not difficult to see that $\snfa$ belongs to $\spanl$~\cite{alvarez1993very}. 

Given functions $f,g : \{0,1\}^* \to \N$, 
%where $\{0,1\}$ is a finite alphabet. A 
$f$ is said to be parsimoniously reducible to $g$ in polynomial-time
if 
%a polynomial-time parsimonious reduction from $f$ to $g$ is a 
there exists a 
polynomial-time computable function $h : \{0,1\}^* \to \{0,1\}^*$ such that, for every $x \in \{0,1\}^*$, it holds that $f(x) = g(h(x))$. 
It is known that $\snfa$ is $\spanl$-complete under polynomial-time parsimonious reductions, which in particular implies that if $\snfa$ can be computed in polynomial time, then $\ptime = \np$~\cite{alvarez1993very}. Moreover, given that $\snfa$ admits an FPRAS and parsimonious reductions preserve the existence of FPRAS, we obtain the following corollary from Theorem \ref{theo-rnl}.
\begin{corollary}
Every function in $\spanl$ admits an FPRAS.
\end{corollary}
Although some classes 
%of functions 
$\cC$ containing $\shp$-complete functions and for which every $f \in \cC$ admits an FPRAS have been identified before \cite{saluja1995descriptive,AMR17}, to the best of our knowledge this is the first such a class with a simple and robust definition based on Turing Machines.
%\marcelo{This discussion about $\spanl$ could be moved to the next section.}

A tight relationship between the existence of an FPRAS and the existence of a schema for almost uniform generation was proved in~\cite{jerrum1986random}, for the class of relations that are {\em self-reducible}. Thus, one might wonder whether the existence of a  PPLVUG  for $\GEN(R)$ in Theorem \ref{theo-rnl} is just a corollary of our FPRAS for $\COUNT(R)$ along with the result in~\cite{jerrum1986random}.
%Interestingly, the answer to this question is no, as the notion of PPLVUG asks for a uniform generator without any distributional error $\epsilon$, whose existence cannot be inferred from the results in \cite{jerrum1986random}.
\review{However, as the notion of PPLVUG asks for a uniform generator without any distributional error $\epsilon$, it is not clear how to infer its existence from the results in \cite{jerrum1986random}.}
%whose existence cannot be inferred from the results in \cite{jerrum1986random}.
%Interestingly, the answer to this question is no,
Thus, we prove in Section \ref{approximate} that $\COUNT(R)$ admits an FPRAS and $\GEN(R)$ admits a PPLVUG, for a relation $R \in \rnl$, without utilizing the aforementioned result from~\cite{jerrum1986random}.

A natural question at this point is whether a simple syntactic restriction on the definition of $\rnl$ gives rise to a class of relations with better properties in terms of enumeration, counting and uniform generation.
%&In what follows, we show that $\rnl$ meets the criteria mentioned before, as it has a simple definition in terms of Turing Machines and good approximation properties. However, to obtain a class where witnesses can be enumerated with constant delay, counted in polynomial time, and uniformly generated also in polynomial time, we need to impose a restriction on $\nl$-transducers. 
Fortunately, the answer to this question comes by imposing a natural and well-studied restriction on Turing Machines, \review{which allows the definition of a class} that contains many natural problems.
More precisely, we consider the notion of $\ul$-transducer, where the letter ``U'' stands for ``unambiguous''.  Formally, $M$ is a $\ul$-transducer if $M$ is an $\nl$-transducer such that for every input $x$ and $y \in M(x)$, there exists exactly one run of $M$ on input $x$ that halts in an accepting state with $y$ as the string in the output tape. Notice that this notion of transducer is based on well-known classes of decision problems (e.g. $\up$~\cite{V76} and $\ul$ \cite{RA00}) adapted to our case, namely, adapted to problems defined as relations.
\begin{definition}
A relation $R$ is in $\rul$ if, and only if, there exists a $\ul$-transducer $M$ such that $\cR(M) = R$.
\end{definition}
For the class $\rul$, we obtain the following result. 
\begin{theorem}\label{theo-rul}
If $R \in \rul$, then $\ENUM(R)$ can be solved with constant delay, there exists a polynomial-time algorithm for $\COUNT(R)$, and there exists a polynomial-time randomized algorithm for  $\GEN(R)$.
\end{theorem}
In particular, it should be noticed that given $R \in \rul$ and an input $x$, the solutions for $x$ can be enumerated, counted and uniformly generated efficiently. \review{In the following section, we provide examples of relations in this class.}

Classes of problems definable by machine models and that can be enumerated with constant delay have been proposed before. In~\cite{AmarilliBJM17}, it is shown that if a problem is definable by a d-DNNF circuit, then the solutions of an instance can be listed with linear preprocessing and constant delay enumeration. Still, to the best of our knowledge, \review{$\rul$ is the first such a class} with a simple and robust definition based on Turing~Machines.

{%\color{blue}
\paragraph{On the relationship of $\rnl$ and $\rul$ with known complexity classes.}
It is well-known that a function $f$ is in $\sharpp$ if and only if
there exists a $p$-relation $R$ such that $f = \COUNT(R)$ (recall the
definition of $p$-relation from Section~\ref{preliminaries}). In the
same way, there exists a tight relationship between $\spanl$ and
$\rnl$, as it is easy to see that a function $f$ is in $\spanl$ if and
only if there exists a relation $R \in \rnl$ such that $f
= \COUNT(R)$. Hence, the reader may wonder why it is necessary to
introduce $\rnl$ and $\rul$, considering further that such classes are
defined in terms of well-known Turing Machine models. The key issue to
consider here is that function complexity classes, such as $\sharpp$
and $\spanl$, are not appropriate to state results about the
enumeration and uniform generation problems. For instance, it would
not be correct to state that every function in $\spanl$ admits a
PPLVUG, as the definition of $\spanl$ does not provide a unique notion
of solution for an input, which is the object to be generated in this
case. In this respect, we introduce $\rnl$ and $\rul$ to have a
unified framework to study the counting, enumeration and
uniform generation problems. The definition of such classes should not
be considered as a contribution of this paper. In fact, they
should only be seen as our way of following the guidance
of \cite{jerrum1986random}, which, as mentioned before, urges the use
of relations to formalize the notion of solution for an input of a problem.
}

%% file: applications.tex
% !TEX root = main.tex
% !TeX spellcheck = en_US

Before providing the proofs of Theorems~\ref{theo-rnl} and~\ref{theo-rul},
%In the next section, we give a sketch proof of Theorem~\ref{theo-rnl} and~\ref{theo-rul}. Before that, 
we give 
%here 
some 
%examples and 
implications of these results. In particular, we show
%, by using 
how $\nl$- and $\ul$-transducers can be used to obtain positive results on query evaluation in areas like information extraction, graph databases, and binary decision diagrams. 

\subsection{Information extraction}

In~\cite{fagin2015document}, the framework of document spanners was proposed as a formalization of ruled-based information extraction. In this framework, the main data objects are documents and spans. 
Formally, given a finite alphabet $\Sigma$, a document is a string $d = a_1 \ldots a_n$ and a span is pair $s = \mspan{i,j}$ with $1 \leq i \leq j \leq n+1$.
A span represents a continuous region of the document
$d$, whose content is the substring of $d$ from positions $i$ to
$j-1$. Given a finite set of variables $\xset$, a mapping $\mu$ is a function  from 
%variables 
$\xset$ 
to the spans of $d$.  

Variable set automata (VA) are one of the main formalisms to specify sets of mappings over a document. Here, we use the notion of extended VA (eVA) from~\cite{florenzano2018constant} to state our main results. 
%Given the lack of space, 
We  only recall the main definitions, and we refer the reader to
% (see~\cite{florenzano2018constant,fagin2015document} for more intuition and further details). 
\cite{florenzano2018constant,fagin2015document} for more intuition and further details.
An eVA is a tuple $\cA = (Q, q_0, F, \delta)$ such that $Q$ is a finite set of states, $q_0$ is the initial state, and $F$ is the final set of states. Further, $\delta$ is the transition relation consisting of letter transitions $(q,a,q')$, or variable-set transitions $(q, S, q')$, where $S \subseteq \{\Open{x}, \Close{x} \mid x \in \xset\}$ and $S \neq \emptyset$. The symbols $\Open{x}$ and $\Close{x}$ are called markers, and they are used to denote that variable $x$ is \review{opened or closed by $\cA$}, respectively.
A run $\rho$ over a document $d = a_1  \,\cdots\, a_n$ is a sequence of the form:
$q_0 \ \trans{X_1} \ p_0 \ \trans{a_1} \ q_1 \ \trans{X_2} \ p_1 \ \trans{a_2} \ \ldots \ \ \trans{a_n} \ q_n \ \trans{X_{n+1}} \ p_n$
where each $X_i$ is a (possible empty) set of markers, $(p_i,
a_{i+1}, q_{i+1}) \in \delta$, and $(q_i, X_{i+1}$, $p_i) \in \delta$
whenever $X_{i+1} \neq \emptyset$, and $q_i = p_i$ otherwise (that is, when $X_{i+1} = \emptyset$). 
We say that a run $\rho$ is valid if for every $x \in \xset$ there exists exactly one pair $\mspan{i, j}$ such that $\Open{x}\, \in X_i$ and $\Close{x}\, \in
X_j$.  A valid run 
$\rho$ naturally defines a mapping $\mu^{\rho}$ that maps $x$ to the only span $[i,
j\rangle$ such that $\Open{x}\, \in X_i$ and
$\Close{x}\, \in X_j$.  We say that $\rho$ is
accepting if $p_n \in F$.  Finally, the semantics $\sem{\cA}(d)$ of $\cA$ over
$d$  is defined as the set of all mappings
$\mu^{\rho}$ where $\rho$ is a valid and accepting run of $\cA$ over~$d$.

In~\cite{Freydenberger17,maturana2018document}, it was shown that the decision problem related to query evaluation, namely, given an eVA $\cA$ and a document $d$ deciding whether $\sem{\cA}(d) \neq \emptyset$, is $\np$-hard. For this reason, in~\cite{florenzano2018constant} a subclass of eVA is considered in order to recover polynomial-time evaluation. An eVA $\cA$ is called functional if every accepting run is valid. Intuitively, a functional eVA does not need to check validity of the run given that it is already known that every run that reaches a final state will be~valid. 

For the query evaluation problem of functional eVA (i.e. to compute $\sem{\cA}(d)$), one can naturally associate the following relation:
\[\evaleva \ = \ \{ ((\cA,d), \mu) \mid \text{$\cA$ is a functional eVA,} \text{ $d$ is a document, and $\mu \in \sem{\cA}(d)$} \} \]
%By using $\nl$-transducers, one can easily 
It is not difficult to show that $\evaleva$ is in $\rnl$. Hence, by Theorem~\ref{theo-rnl} we get the following results.
\begin{corollary}\label{cor:rnl-spanners}
\begin{sloppypar}
	$\ENUM(\evaleva)$ can be enumerated with polynomial delay, $\COUNT(\evaleva)$ admits an FPRAS, and $\GEN(\evaleva)$ admits a PPLVUG.
\end{sloppypar}
\end{corollary}
In~\cite{florenzano2018constant}, it was shown that every functional RGX or functional VA (not necessarily extended) can be converted in polynomial time into an functional eVA. Therefore, Corollary~\ref{cor:rnl-spanners} also holds for these more general classes.
Notice that in~\cite{freydenberger2018joining},
%already 
a \review{polynomial delay enumeration} algorithm for  
%algorithm to enumerate 
$\sem{\cA}(d)$ \review{was provided}.
%with polynomial delay and, 
Thus, only the results about $\COUNT(\evaleva)$ and  $\GEN(\evaleva)$ can be considered as new.
 
Regarding efficient enumeration and exact counting, a \review{constant delay} algorithm with polynomial preprocessing \review{was given in~\cite{florenzano2018constant}} for the class of deterministic functional eVA. Here, we can easily extend these results for a more general class, that we called unambiguous functional eVA. Formally, we say that an eVA is unambiguous if for every two valid and accepting runs $\rho_1$ and $\rho_2$, it holds that $\mu^{\rho_1} \neq \mu^{\rho_2}$. In other words, each output of an unambiguous eVA \review{is witnessed} by exactly one run. As in the case of $\evaleva$, we can define the relation $\evalueva$, by restricting the input to unambiguous functional eVA. By using $\ul$-transducers and Theorem~\ref{theo-rul}, we can then extend the results in~\cite{florenzano2018constant} for the unambiguous case. 
\begin{corollary}\label{cor:rnl-spanners-u}
\begin{sloppypar}
	$\ENUM(\evalueva)$ can be solved with constant delay, there exists a polynomial-time algorithm for $\COUNT(\evalueva)$, and there exists a polynomial-time randomized algorithm for $\GEN(\evalueva)$.
\end{sloppypar}
\end{corollary}
Notice that this result gives a \review{constant delay} algorithm with polynomial preprocessing for the class of unambiguous functional eVA. Instead, the algorithm in~\cite{florenzano2018constant} has linear preprocessing over 
%the 
documents, restricted to the case of deterministic eVA. This leaves open whether there exists a \review{constant delay} algorithm with linear preprocessing over 
%the 
documents for the unambiguous case. 

\subsection{Query evaluation in graph databases}

Enumerating, counting, and generating paths are relevant tasks for query evaluation in graph databases \cite{angles2017foundations}.
Given a finite set $\Sigma$ of labels, a 
%property 
graph database $G$ is a pair $(V, E)$ where $V$ is a finite set of vertices and $E \subseteq V \times \Sigma \times V$ is a finite set of labeled edges.
% labeled.
% with properties. 
Here, 
%a labeled graph represents a graph database where 
\review{vertices} represent pieces of data and edges 
%are the 
specify relations between 
%these pieces of data
them~\cite{angles2017foundations}. 
One of the core query languages for posing queries on graph databases
%, i.e. graph patterns, 
are regular path queries (RPQ). 
%Given a finite alphabet $\Sigma$, which representes a set of labels,
% set of properties $\Sigma$ 
An RPQ is a triple $(x, R, y)$ where $x,y$ are variables and $R$ is a regular expression over $\Sigma$. As usual, we denote by $\cL(R)$ all the strings over $\Sigma$ that conform to $R$.  Given an RPQ $Q= (x, R, y)$, a graph database $G = (V,E)$, and \review{vertices} $u, v \in V$,  one would like to retrieve, count, or uniformly generate all paths\footnote{Notice that the standard semantics for RPQs is to retrieve pair of \review{vertices}. Here we consider a less standard semantics based on paths which is also relevant for graph databases~\cite{ACP12,LM13,angles2017foundations}.} in $G$ going from $u$ to $v$ \review{that satisfy $Q$}.
Formally, a path from $u$ to $v$ in $G$ is a sequence of vertices and 
%properties 
labels of the form
$
\pi \ = \   v_0, p_1, v_1, p_2, \ldots, p_n, v_n
$,
such that $(v_i, p_{i+1}, v_{i+1}) \in E$, $u = v_0$, and $v=v_n$. 
%We say that 
A path $\pi$ is said to satisfy
% like above satisfies 
$Q = (x, R, y)$ if the string $p_1 p_2 \,\cdots\, p_n \in \cL(R)$. The length of $\pi$ is defined as $|\pi| = n$. Clearly, between $u$ and $v$ there can be an infinite number of paths \review{that satisfy $Q$}. For this reason, one usually wants to retrieve all paths between $u$ and $v$ of at most a certain length $n$, namely, one usually considers the set
$
\sem{Q}_n(G,u,v)
$ of all paths $\pi$ from $u$ to $v$ in $G$ such that $\pi$ satisfies $Q$ and $|\pi| = n$. This naturally defines the following relation representing the problem of 
%query
\review{evaluating an RPQ} over a graph database:
$$
\evalRPQ \ = \ \{ ((Q, 0^n,G, u,v), \pi) \mid \pi \in \sem{Q}_n(G,u,v)\}.
$$
Using this relation, fundamental problems 
%related to query evaluation of 
for 
RPQs such as
% naturally follow like 
enumerating, counting, or uniform generating paths can be naturally represented. 
%By using $\nl$-transducers one can easily show that $\evalRPQ$ is in $\rnl$.
It is not difficult to show that $\evalRPQ$ is in $\rnl$, from which the following corollary can be obtained by using Theorem~\ref{theo-rnl}.
% we get the following results.
%Although 
%in combined complexity was open.
\begin{corollary}\label{cor:rnl-graphs}
\begin{sloppypar}
\review{$\ENUM(\evalRPQ)$ can be enumerated with polynomial delay,
$\COUNT(\evalRPQ)$ admits an FPRAS, and
 $\GEN(\evalRPQ)$ admits a PPLVUG.}
\end{sloppypar}
\end{corollary}
It is important to mention that giving a \review{polynomial delay enumeration} algorithm for $\evalRPQ$ is straightforward, but the existence of an FPRAS and a PPLVUG for $\evalRPQ$ was not known before when queries are part of the input (that is, in combined complexity~\cite{V82}).

\subsection{Binary decision diagrams}

Binary decision diagrams 
%(OBDDs) 
are an abstract representation \review{of Boolean functions} which are widely used in computer science and have found many applications in areas like formal verification~\cite{bryant1992symbolic}. A binary decision diagram (BDD) is a directed acyclic graph $D = (V, E)$ where each \review{vertex} $v$ is labeled with a variable $\var(v)$ and has at most two edges going to children $\lo(v)$ and $\hi(v)$. Intuitively, $\lo(v)$ and $\hi(v)$ represent the next \review{vertices} when $\var(v)$ takes values $0$ and $1$, respectively. $D$ contains only two terminal, or sink \review{vertices}, labeled by $0$ or $1$, and one initial \review{vertex} called $v_0$. We assume that every path from $v_0$ to a terminal \review{vertex} does not repeat variables.
Then given an assignment $\sigma$ from the variables in $D$ to $\{0,1\}$, we have that $\sigma$ naturally defines a path from $v_0$ to a terminal \review{vertex} $0$ or $1$. In this way, $D$ defines \review{a Boolean function} that gives a value in $\{0,1\}$ to each assignment $\sigma$; in particular, 
%that from assignments $\sigma$ to $\{0,1\}$, where 
$D(\sigma) \in \{0,1\}$ corresponds to the sink \review{vertex} reached by starting from $v_0$ and following the values in~$\sigma$. 
For Ordered BDDs (OBDDs), we also have a linear order $<$ over the variables in $D$ such that, for every $v_1, v_2 \in V$ with $v_2$ a child of $v_1$, it holds that $\var(v_1) < \var(v_2)$. Notice that not necessarily all variables appear in a path from the initial \review{vertex} $v_0$ to a terminal \review{vertex} $0$ or $1$. Nevertheless, the promise in an OBDD is that variables will appear following the order $<$. 

An OBDD $D$ defines the set of assignments $\sigma$ such that $D(\sigma) = 1$. Then $D$ can be considered as a succinct representation of the set $\{\sigma \mid D(\sigma)=1\}$, and one would like to enumerate, count and uniformly generate assignments given $D$. This motivates the relation:
$$
\evalOBDD \ = \ \{ (D, \sigma) \mid D(\sigma) = 1\}.
$$
Given $(D, \sigma)$ in $\evalOBDD$, there is exactly one path in $D$ that witnesses $D(\sigma) = 1$. Therefore, one can easily show that $\evalOBDD$ is in $\rul$. \review{By Theorem~\ref{theo-rul}, we obtain that}:
% with an $\ul$ transducer.
\begin{corollary}
\begin{sloppypar}
	$\ENUM(\evalOBDD)$ can be enumerated with constant delay, there exists a polynomial-time algorithm for $\COUNT(\evalOBDD)$, and there exists a polynomial-time randomized algorithm for $\GEN(\evalOBDD)$.
\end{sloppypar}
\end{corollary}
The above results are well known. Nevertheless, they show how easy and \review{direct it is to} use $\ul$-transducers to realize the good algorithmic properties that a data structure like OBDD has. 

Some non-deterministic variants of BDDs have been studied in the literature~\cite{DBLP:journals/corr/abs-1811-02944}.
%, called non-deterministic OBDD (nOBDD). 
In particular, an nOBDD extends an OBDD with vertices $u$ without variables (i.e. $\var(u) = \bot$) and without labels on its children. Thus, an nOBDD is non-deterministic in the sense that given an assignment $\sigma$, there can be several paths that bring $\sigma$ from the initial \review{vertex} $v_0$ to a terminal \review{vertex} with labeled $0$ or $1$. Without lost of generality, nOBDDs are assumed to be consistent in the sense that, for each $\sigma$, all paths of $\sigma$ in $D$ can reach $0$ or $1$, but not both.

As in the case of OBDDs, we can define a relation $\evalnOBDD$ that pairs an nOBDD $D$ with an assignment $\sigma$ that evaluate $D$ to~$1$ (i.e. $D(\sigma) = 1$). Contrary to OBDDs, an nOBDD looses the single witness property, and now an assignment $\sigma$ can have several paths from the initial \review{vertex} to the $1$ terminal \review{vertex}. Thus, it is not clear whether $\evalnOBDD$ is in $\rul$. Still one can easily show that $\evalnOBDD \in \rnl$, from which the following results follow.

\begin{corollary}\label{cor:rnl-obdd}
\begin{sloppypar}
	$\ENUM(\evalnOBDD)$ can be solved with polynomial delay, $\COUNT(\evalnOBDD)$ admits an FPRAS, and
	 $\GEN(\evalnOBDD)$ admits a PPLVUG.
\end{sloppypar}
\end{corollary} 
It is important to stress that the existence of an FPRAS and a PPLVUG for $\evalnOBDD$ was not known~before,
%open, 
and one can easily show this by using $\nl$-transducers and then applying Theorem~\ref{theo-rnl}.

%% file: cde.tex
% !TEX root = main.tex
% !TeX spellcheck = en_US

%\begin{sloppypar}
The goal of this section is to establish the good algorithmic properties of $\rul$, that is, to prove Theorem \ref{theo-rul}. To this end, we start by introducing a simple notion of reduction for the classes $\rnl$ and $\rul$, which will allow for much simpler proofs.
%\end{sloppypar}

A natural question to ask is which notions of ``completeness'' and ``reduction'' are appropriate for our framework. Notions of reductions for relations have been proposed before, in particular in the context of search problems~\cite{DGP09}. However, we \review{do not intend to} discuss them here; instead, we use an idea of completeness that is very restricted, but that turns out to be useful in this context.

Let $\mathcal{C}$ be a complexity class of relations and $R,S\in\mathcal{C}$, and recall that $W_R(x)$ is defined as the set \review{of
%witnesses
solutions} for input $x$, that is, $W_R(x) = \{ y \mid (x,y) \in R\}$. We say $R$ is reducible to $S$ if there exists a function $f:\{0,1\}^*\to\{0,1\}^*$, computable in polynomial time, such that for every $x \in \{0,1\}^*$: $W_R(x)=W_S(f(x))$. Also, if $T$ is reducible to $S$ for every $T\in\mathcal{C}$, we say $S$ is complete for $\mathcal{C}$. Notice that this definition is very restricted, since the notion of reduction requires the set \review{of
%witnesses
solutions} to be exactly the same for both relations (it is not sufficient that they have the same size, for example). The \review{benefit of this} kind of reduction is that it preserves all the properties of efficient enumeration, counting and uniform generation that we introduced in Sections \ref{preliminaries} and \ref{nlclass}, as stated in Proposition~\ref{prop_reduction_preserves_properties}.

\begin{proposition}\label{prop_reduction_preserves_properties}
If a relation $R$ can be reduced to a relation $S$, then:
\begin{itemize}
\item If $\ENUM(S)$ can be solved with constant (resp. polynomial) delay, then $\ENUM(R)$ can be solved with constant (resp. polynomial) delay.
\item If there exists a polynomial-time algorithm (resp. an FPRAS) for $\COUNT(S)$, then there exists a polynomial-time algorithm (resp. an FPRAS) for $\COUNT(R)$.
\item If there exists a polynomial-time randomized algorithm (resp. a PPLVUG) for $\GEN(S)$, then there exists a polynomial-time randomized algorithm (resp. a PPLVUG) for $\GEN(R)$.
\end{itemize}
\end{proposition}
\begin{proof}
	We go into some detail, but the idea of the proof is very simple. Because our notion of reduction is so strong, all efficient algorithms for $S$ apply immediately for $R$, provided we add a preprocessing phase \review{where we compute a function reducing from $R$ to $S$}. Since that takes only polynomial time, it preserves the overall complexity of all the types of algorithms we have~discussed. 
	
	Now, with more detail and formality. Since $R$ can be reduced to $S$, there exist a polynomial $p(u)$ and a function $f$ such that $W_S(f(x))=W_R(x)$ for every input string $x$, and $f(x)$ can be computed in time $p(|x|)$.
	First, suppose $\ENUM(S)$ can be solved with constant (resp. polynomial) delay, so there is an algorithm $\mathcal{E}$ that enumerates $W_S(f(x))$ with constant (resp. polynomial) delay and with precomputation phase of time $q(|f(x)|)$ for some polynomial $q$. Now, consider the following procedure for $\ENUM(R)$ on input $x$. First, we compute $f(x)$ in time $p(|x|)$. Then, we run $\mathcal{E}(f(x))$, which enumerates \review{all
        %witnesses
        solutions} in $W_S(f(x))$, that is, it enumerates \review{all
        %witnesses
        solutions} in $W_R(x)$. So, the precomputation time of the procedure takes time $p(|x|)+q(|f(x)|)\leq p(|x|)+q(p(|x|))$, which is \review{polynomial in $|x|$}. The enumeration phase is the same as for $\mathcal{E}(f(x))$, so it has constant (resp. polynomial) delay. We conclude that $\ENUM(R)$ can be solved with constant (resp. polynomial)~delay.
	
	Now, suppose there exists a polynomial-time algorithm $\A$ for $\COUNT(S)$, let $q$ be the polynomial that characterizes its complexity, and consider the following procedure for $\COUNT(R)$ on input $x$. First, we construct $f(x)$ in time $p(|x|)$. Next, we run $\A(f(x))$, which computes $|W_S(f(x))|$, that is, it computes $|W_R(x)|$. So, the procedure calculates $|W_R(x)|$ and takes time $p(|x|)+q(|f(x)|)\leq p(|x|)+q(p(|x|))$, which is \review{polynomial in $|x|$}. We conclude that $\COUNT(R)$ has a polynomial-time algorithm. The proof for the case of an FPRAS is completely analogous.
	
	Finally, suppose there exists a polynomial-time randomized algorithm $\cG$ for $\GEN(S)$, and let $q$ be the polynomial that characterizes its complexity. Now, consider the following procedure for $\GEN(R)$ on input $x$. First, we construct $f(x)$ in time $p(|x|)$. Next, we run $\cG(f(x))$, which \review{outputs a
        %witness
        solution} from $W_S(f(x))$, \review{that is, a
        %witness
        solution} from $W_R(x)$, uniformly at random. So, the procedure generates an element from $W_R(x)$ uniformly at random and takes time $p(|x|)+q(|f(x)|)\leq p(|x|)+q(p(|x|))$, which is \review{polynomial in $|x|$}. We conclude that $\GEN(R)$ has a polynomial-time randomized algorithm. The proof for \review{the case of a PPLVUG} is completely analogous.
\end{proof}

Therefore, by finding a complete relation $S$ for a class $\mathcal{C}$ under the notion of reduction just defined, we can study the aforementioned problems for $S$ knowing that the obtained results will extend to every relation in the class $\mathcal{C}$. In what follows, we identify complete problems for the classes $\rnl$ and $\rul$, and use them first to establish the good algorithmic properties of 
%both. Moreover, we prove that the identified problems are self-reducible, which will be useful for establishing some of the results of this section.
$\rul$. Moreover, we prove that
\review{the identified problems are self-reducible~\cite{jerrum1986random}}, which will be useful for establishing some of the results of this
section as well as for some of the results proved in Section \ref{approximate} for the class $\rnl$.

\subsection{Complete problems for $\rnl$ and $\rul$}

The notion of reduction just defined is useful for us because $\rnl$ and $\rul$ admit natural complete problems under this notion. These complete relations are defined in terms of NFAs and we call them $\nfa$ and $\unfa$. We already introduced $\nfa$ in Section \ref{nlclass}, and we now define $\unfa$ as 
\begin{multline*}
\unfa \ = \ \{ ((\nA,0^k), w) \mid \nA \text{ is an unambiguous NFA}\\ 
\text{with alphabet } \{0,1\}, w \in \{0,1\}^k \text{ and } w \text{ is accepted by } \nA\},
\end{multline*}
where an NFA is said to be unambiguous if there exists exactly one accepting run for every string accepted by~it.

\review{Recall from Section \ref{nlclass} that $\nfa \in \rnl$. Besides, it is easy to see that $\unfa \in \rul$.} To see \review{why these relations are complete} for our classes, consider the following. Take a relation $R$ in $\rnl$ (the case for $\rul$ is the same). We know there is an $\nl$-transducer $M$ that characterizes it. Run now $M$ on some given input $x$. Since \review{$M$ works in logarithmic space}, there is only a polynomial number of different configurations that $M$ can ever be in (\review{polynomial in $|x|$}).
Hence, we can consider the set of possible configurations as the states of an NFA $\nA_x$, which then has only polynomial size. The transitions of $\nA_x$ are determined by the transitions between the configurations of $M$. Moreover, a symbol output by the transducer $M$ is interpreted as a symbol read by the automaton $\nA_x$. In this way, $\nA_x$ accepts exactly the language $W_R(x)$. We formalize this idea in the following result.

\begin{proposition}\label{prop_complete_problems}
\begin{sloppypar}
$\nfa$ is complete for $\rnl$ and $\unfa$ is complete for $\rul$.
\end{sloppypar}
\end{proposition}
	We will prove the result only for the case of $\rul$ and $\unfa$, as the other case is completely analogous. The following lemma is the key ingredient in our argument. The proof of this lemma is given in Appendix \ref{app-lemma_ufa_from_relation}.
	\begin{lemma}\label{lemma_ufa_from_relation}
		Let $R$ be a relation in $\rul$.
                %defined on an alphabet $\{0,1\}$.
                Then there exists a polynomial-time algorithm that, given $x \in \{0,1\}^*$, produces an unambiguous NFA $\nA_x$ such that $y \in W_R(x)$ if and only if $y$ is accepted by $\nA_x$. 
	\end{lemma}

\begin{proof}[Proof of Proposition \ref{prop_complete_problems}]
	Let $R$ be a relation in $\rul$ and $x$ be a string in $\{0,1\}^*$. We know by Lemma~\ref{lemma_ufa_from_relation} that we can construct in polynomial time an unambiguous NFA $\nA_x$  such that 
	%$ W_R(x)=\mathcal{L}(\nA_x)$. 
	$y \in W_R(x)$ if and only if $y$ is accepted by $\nA_x$. Now, since $R$ is a $p$-relation, there exists a polynomial $q$ such that $|y|=q(|x|)$ for all $y\in W_R(x)$. Thus, 
	%given that $W_R(x)=\mathcal{L}(\nA_x)$, 
	we have that all words accepted by $\nA_x$ have the same length $q(|x|)$. We conclude that $W_R(x) = W_\unfa\left(\left(\nA_x, 0^{q(|x|)}\right)\right)$. Since this works for every $R\in\rul$ and every input $x$, by definition of completeness we deduce that $\unfa$ is complete for $\rul$.
\end{proof}

\review{In Section \ref{nlclass}, we show that $\dnf \in \rnl$. Thus, a
fundamental question is whether $\dnf$ is complete for the class
$\rnl$ under the notion of reduction considered in this work. Notice
that if this holds, then we will obtain that $\COUNT(\dnf)$ is
$\spanl$-complete under the notion of polynomial-time parsimonious
reduction (introduced in Section \ref{nlclass}). However,
$\COUNT(\dnf)$ is only known to be $\spanl$-complete under
polynomial-time Turing reductions, and it is unknown whether
$\COUNT(\dnf)$ is complete for $\spanl$ under polynomial-time
parsimonious reductions.
%such a concept of reduction;
In fact, it is not even known whether $\COUNT(\dnf)$ is
$\spanl$-complete under some notion of reduction that preserves the
existence of an FPRAS, so that the existence of an FPRAS for
$\COUNT(\dnf)$ cannot be used to infer the existence of an FPRAS for $\snfa$. Hence, we leave as an open problem whether $\dnf$ is
complete for $\rnl$ in the sense studied in this~article.}

\subsection{$\nfa$ and $\unfa$ are self-reducible}\label{proof_self_reducibility}

Self-reducibility is a property of many natural relations, and it plays a key role in proving \review{some important results, like the tight relationship between counting and uniform generation 
established in \cite{jerrum1986random}.} There are different ways of formalizing this concept, and they can get rather technical, but the intuition is pretty straightforward. We say that a (decision) problem is self-reducible if it can be solved by referring to smaller instances of the same problem. 
For example, SAT is self-reducible. Given a propositional formula $\phi$, consider its satisfiability problem. We can easily reduce that problem to smaller instances of SAT as follows. Take the first variable of $\phi$ and replace it by $0$ to get a new formula $\phi_0$. Do the same with $1$ to get a new formula $\phi_1$. Notice that $\phi$ is satisfiable if and only if $\phi_0$ or $\phi_1$ is satisfiable. Moreover, both $\phi_0$ and $\phi_1$ have one less variable than $\phi$, so they are smaller instances.

Now, self-reducibility does not imply the existence of a polynomial-time solution for a problem, as SAT well illustrates. It is true that the instances get smaller, until they eventually become trivially easy to solve. But the number of instances is multiplied, so recursively \review{applying self-reducibility can} lead to an exponential number of smaller instances to solve. Rather than a solution method, self-reducibility is thought of as a structural feature of a problem.

Now, definitions (and proofs) of self-reducibility can get very technical, partly because they have to formalize the notion of ``smaller instance''. Hence, they crucially depend on the way that problems are encoded.\footnote{Thus, saying something like ``SAT is self-reducible'' is slightly inaccurate. We need to specify the way in which the problem, inputs and solutions are encoded, before we can assert something like that.} We now state the main result of this subsection.

\begin{proposition}\label{prop_self_reducibility}
$\nfa$ and $\unfa$ are self-reducible.
\end{proposition}

\review{To see the intuition behind this result, consider first a deterministic finite automaton (DFA) $D$ over the alphabet $\{0,1\}$, and suppose it accepts a string $w = 0 \cdot w'$, where $w' \in \{0,1\}^*$. Then assuming that $q_0$ is the initial state of $D$, we know that the accepting run for $w$ moves from $q_0$ to a state $q_1$ by reading symbol $0$, and then it continues processing $w'$ from $q_1$. Now, if we change the initial state to $q_1$ to get a new DFA $D_{0}$, then $D_{0}$ accepts the string $w'$. In other words, if $\mathcal{L}(D)$ is the language accepted by $D$, then we have that:
\begin{eqnarray*}
\mathcal{L}(D) &=&
%\bigcup_{s\in \{0,1\}} \{s\cdot w' \mid w'\in\mathcal{L}(D_s)\},
\{0\cdot w' \mid w'\in\mathcal{L}(D_0)\} \cup
\{1\cdot w' \mid w'\in\mathcal{L}(D_1)\},
\end{eqnarray*}
where DFA $D_1$ is defined in the same way as $D_{0}$. Besides, notice that if the length of the strings to be accepted by $D$ is given as a parameter, as in the case of $\nfa$, then we can assume $D$ does not contain any cycles, and each automaton $D_i$ ($i = 0,1$) can be made smaller than $D$ by removing $q_0$ and updating the transition function of $D$ accordingly. Hence, the above equality shows that the language accepted by $D$ can be defined in terms of the languages accepted by smaller deterministic finite automata. The same idea can be applied to an NFA $N$, although constructing each NFA $N_i$ ($i = 0,1$) is a little more complicated as there can be several transitions from a state that read the same symbol. Intuitively, this shows that $\nfa$ is self-reducible. The precise definition of self-reducibility (with all its technicalities) and the complete proof of Proposition \ref{prop_self_reducibility} can be found in Appendix~\ref{app-prop_self_reducibility}.}

\subsection{Establishing the good algorithmic properties of $\rul$}\label{sec-good-alg-rul}
Theorem \ref{theo-rul} is a consequence of Propositions \ref{prop_reduction_preserves_properties} and \ref{prop_complete_problems}, and the following result.
\begin{proposition}\label{prop-all-rul}
\begin{sloppypar}
$\ENUM(\unfa)$ can be solved with constant delay, there exists a polynomial-time algorithm for $\COUNT(\unfa)$, and there exists a polynomial-time randomized algorithm for $\GEN(\unfa)$.
\end{sloppypar}
\end{proposition}
%\noindent 
To sum up all the results just mentioned:
%, the big picture is: 
$\unfa$ is complete for $\rul$, it has good algorithmic properties, and our notion of reduction (and completeness) preserves all the algorithmic properties we have discussed. In what follows, we prove each of the three results stated in Proposition \ref{prop-all-rul}.
% one by one. 

\subsubsection{$\ENUM(\unfa)$ can be solved with constant delay}
\label{sec-enum-cd}

\indent We now provide a sketch of the \review{constant delay} algorithm. The idea is conceptually simple. Remember what we want: to output all strings \review{of a certain length accepted by an unambiguous NFA}, without repetition. We may use a preprocessing phase of polynomial time, but afterwards, there can be at most linear time between one string and the next. 

Now, let $(\nA,0^k)$ be the input, and consider Figure \ref{fig_ufa} with $k=3$ as an example. \review{To do constant delay enumeration}, we do a depth-first traversal of the NFA, starting from the initial state. As we traverse the NFA, we read the symbols from the transitions, and store them in a partial string. When the partial string reaches length $k$, if we happen to be in a final state, we output the string.

\begin{figure}[ht]
	\centering
	\begin{tikzpicture}[scale=0.15]
	\tikzstyle{every node}+=[inner sep=0pt]
	\draw [black] (15.4,-20.8) circle (3);
	\draw (15.4,-20.8) node {$q_0$};
	\draw [black] (29.1,-12.3) circle (3);
	\draw (29.1,-12.3) node {$q_1$};
	\draw [black] (29.1,-29.1) circle (3);
	\draw (29.1,-29.1) node {$q_2$};
	\draw [black] (43.1,-12.3) circle (3);
	\draw (43.1,-12.3) node {$q_3$};
	\draw [black] (43.1,-29.1) circle (3);
	\draw (43.1,-29.1) node {$q_4$};
	\draw [black] (57.7,-12.3) circle (3);
	\draw (57.7,-12.3) node {$q_F$};
	\draw [black] (57.7,-12.3) circle (2.4);
	\draw [black] (57.7,-29.1) circle (3);
	\draw (57.7,-29.1) node {$q_5$};
	\draw [black] (17.95,-19.22) -- (26.55,-13.88);
	\fill [black] (26.55,-13.88) -- (25.61,-13.88) -- (26.13,-14.73);
	\draw (21.31,-16.05) node [above] {$0$};
	\draw [black] (32.1,-12.3) -- (40.1,-12.3);
	\fill [black] (40.1,-12.3) -- (39.3,-11.8) -- (39.3,-12.8);
	\draw (36.1,-11.8) node [above] {$0$};
	\draw [black] (46.1,-12.3) -- (54.7,-12.3);
	\fill [black] (54.7,-12.3) -- (53.9,-11.8) -- (53.9,-12.8);
	\draw (50.4,-11.8) node [above] {$0,1$};
	\draw [black] (17.97,-22.35) -- (26.53,-27.55);
	\fill [black] (26.53,-27.55) -- (26.11,-26.7) -- (25.59,-27.56);
	\draw (21.25,-25.45) node [below] {$1$};
	\draw [black] (31.02,-26.8) -- (41.18,-14.6);
	\fill [black] (41.18,-14.6) -- (40.28,-14.9) -- (41.05,-15.54);
	\draw (35.55,-19.26) node [left] {$0$};
	\draw [black] (45.07,-26.84) -- (55.73,-14.56);
	\fill [black] (55.73,-14.56) -- (54.83,-14.84) -- (55.58,-15.5);
	\draw (50.94,-22.15) node [right] {$1$};
	\draw [black] (46.1,-29.1) -- (54.7,-29.1);
	\fill [black] (54.7,-29.1) -- (53.9,-28.6) -- (53.9,-29.6);
	\draw (50.4,-29.6) node [below] {$0$};
	\draw [black] (32.1,-29.1) -- (40.1,-29.1);
	\fill [black] (40.1,-29.1) -- (39.3,-28.6) -- (39.3,-29.6);
	\draw (36.1,-29.6) node [below] {$1$};
	\draw [black] (8.8,-20.8) -- (12.4,-20.8);
	\fill [black] (12.4,-20.8) -- (11.6,-20.3) -- (11.6,-21.3);
	\end{tikzpicture}
	\caption{Unambiguous NFA $\nA$.}
	\label{fig_ufa}
\end{figure}
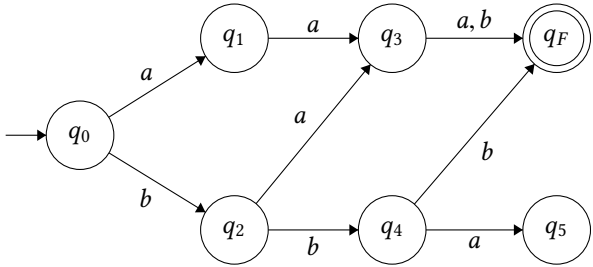

Basically, that is all you have to do, but there are a few technicalities remaining. First of all, we mentioned depth-first traversal even though we are not analyzing a graph, but an NFA. The clarification is simple. We will use the preprocessing phase to get a labeled directed acyclic graph (DAG) \review{$\nAun$} from $\nA$ and $k$, and do the depth-first traversal on \review{$\nAun$}. The DAG \review{$\nAun$} is obtained by first unrolling $\nA$ in the following way:
\begin{enumerate}
    \item Cluster all final states of $\nA$ into a single final state. This is easy to do: create a new state $q_F$, make it the unique final state, and create an $\varepsilon$-transition from all previous final states to the new one.
    \item Remove all $\varepsilon$-transitions (this is a standard procedure for an NFA).
    \item Unroll the NFA $k+1$ times. That is, for each state $q$ create $k+1$ copies $\{(q, i)\}_{i=0}^k$, and for each transition $q \xrightarrow{a} p$ in $A$, create the transitions $\{(q, i)  \xrightarrow{a} (p, i+1)\}_{i=0}^{k-1}$ in the unrolled automaton. Keep a unique initial state $(q_0,0)$ and a unique final state $(q_F,k)$.
    \item Remove all nodes that are not a part of an accepting run from the initial to the final state.
\end{enumerate}

See Figure \ref{fig_dag} for an example of this kind of transformation. It is easy to see that this can be done in polynomial time and that it produces a new NFA (alternatively, a labeled DAG) \review{$\nAun$} that is still unambiguous and accepts the same words of length $k$ as $\nA$. Since there are no $\varepsilon$-transitions, each string of length $k$ accepted by $\nA$ can be interpreted as a path of length $k$ in \review{$\nAun$} from the initial to the final state, and vice versa. Thus, a depth-first traversal of \review{$\nAun$} \review{will go through} all words of length $k$ accepted by $k$, and no more.

\begin{figure}[ht]
	\centering
	\begin{tikzpicture}[scale=0.16]
	\tikzstyle{every node}+=[inner sep=0pt]
	\draw [black] (18.1,-30.9) circle (3);
	\draw (18.1,-30.9) node {{\small $(q_0,0)$}};
	\draw [black] (32.1,-21) circle (3);
	\draw (32.1,-21) node {{\small $(q_1,1)$}};
	\draw [black] (32.1,-39.8) circle (3);
	\draw (32.1,-39.8) node {{\small $(q_2,1)$}};
	\draw [black] (49.6,-21) circle (3);
	\draw (49.6,-21) node {{\small $(q_3,2)$}};
	\draw [black] (49.6,-39.8) circle (3);
	\draw (49.6,-39.8) node {{\small $(q_4,2)$}};
	\draw [black] (65.1,-31.2) circle (3.7);
	\draw (65.1,-31.2) node {{\small $(q_F,3)$}};
	\draw [black] (65.1,-31.2) circle (3.1);
	\draw [black] (20.63,-32.51) -- (29.57,-38.19);
	\fill [black] (29.57,-38.19) -- (29.16,-37.34) -- (28.62,-38.18);
	\draw (24.1,-35.85) node [below] {$1$};
	\draw [black] (35.1,-39.8) -- (46.6,-39.8);
	\fill [black] (46.6,-39.8) -- (45.8,-39.3) -- (45.8,-40.3);
	\draw (40.85,-40.3) node [below] {$1$};
	\draw [black] (61.306,-30.899) arc (-96.27524:-151.64288:13.727);
	\fill [black] (61.31,-30.9) -- (60.57,-30.31) -- (60.46,-31.31);
	\draw (54.14,-29.14) node [below] {$1$};
	\draw [black] (52.17,-38.25) -- (61.73,-32.45);
	\fill [black] (61.73,-32.45) -- (60.79,-32.44) -- (61.31,-33.3);
	\draw (57.95,-35.85) node [below] {$1$};
	\draw [black] (34.14,-37.6) -- (47.56,-23.2);
	\fill [black] (47.56,-23.2) -- (46.64,-23.44) -- (47.38,-24.12);
	\draw (40.32,-28.94) node [left] {$0$};
	\draw [black] (20.55,-29.17) -- (29.65,-22.73);
	\fill [black] (29.65,-22.73) -- (28.71,-22.79) -- (29.29,-23.6);
	\draw (24.15,-25.45) node [above] {$0$};
	\draw [black] (35.1,-21) -- (46.6,-21);
	\fill [black] (46.6,-21) -- (45.8,-20.5) -- (45.8,-21.5);
	\draw (40.85,-20.5) node [above] {$0$};
	\draw [black] (52.594,-20.96) arc (84.38562:27.69626:13.481);
	\fill [black] (63.21,-28.11) -- (63.28,-27.17) -- (62.4,-27.64);
	\draw (59.75,-22.7) node [above] {$0$};
	\end{tikzpicture}
	\caption{Graph \review{$\nAun$} obtained from $\nA$.}
	\label{fig_dag}
\end{figure}

The second technicality concerns the following question: is the enumeration truly repetition-free? It is, for the following reason. Each path of length $k$ from the initial state to the final state is only traversed once (by definition of a depth-first traversal of a graph). Moreover, each one of those paths corresponds to a different string since $A$ and \review{$\nAun$} are both unambiguous automata. 

Finally, \review{does the enumeration phase really have constant delay? That is, does it take time $O(k)$ between one solution and the next?
The answer is yes.} Notice that it takes time $O(k)$ to traverse from the initial state to the final state, and from one final state visit to the next, because the traversal is depth-first. Also, recall that we removed (in step (4) above) all nodes that were not part of an accepting run from the initial to the final state. Thus, there is no time wasted: each traversal to the final state produces a new string that we can output.

\subsubsection{There exists a polynomial-time algorithm for $\COUNT(\unfa)$}

\review{Consider the graph \review{$\nAun$} as defined in Section \ref{sec-enum-cd}}. As we already pointed out, the number of paths of length $k$  from the initial to the final state is exactly what we want to compute: the number of strings of length $k$ accepted  by $\nA$ (the unambiguity assumption is crucial here). Since \review{$\nAun$} is a DAG, we know the number of paths between two given nodes can be computed exactly in polynomial time by dynamic programming. Hence, there exists a polynomial-time algorithm for $\COUNT(\unfa)$.

\subsubsection{There exists a polynomial-time randomized algorithm for $\GEN(\unfa)$}

{%\color{blue}
Consider an input $(\nA, 0^k)$ for the problem $\GEN(\unfa)$. Moreover, as for the case of $\COUNT(\unfa)$, consider the graph \review{$\nAun$} defined in Section \ref{sec-enum-cd}, which can also be seen as an automaton. Assume that $(q_0,0)$ is the initial state of \review{$\nAun$}, and that $\{(q_1,1)$, $\ldots$, $(q_\ell,1)\}$ is the set of states in \review{$\nAun$} reachable from $(q_0,0)$ by following an edge with label 0. Let $N_0$ be the number of strings of length $k$ that start with the symbol 0 and are accepted by $\nA$. Then we can compute $N_0$ in polynomial time by using the counting algorithm mentioned in the previous section, starting from each one of the states in  $\{(q_1,1)$, $\ldots$, $(q_\ell,1)\}$. Notice that this algorithm works properly as $\nA$ is an unambiguous NFA. In the same way, we can compute in polynomial time the number $N_1$ of strings of length $k$ that start with the symbol 1 and are accepted by $\nA$. Given $N_0$ and $N_1$, the first symbol $w_1$ of the string $w = w_1 \cdots w_k$ to be generated is chosen according to the probabilities:
$$
\pr(w_1 = 0) = \frac{N_0}{N_0+N_1} \quad \text{ and } \quad \pr(w_1 = 1) = \frac{N_1}{N_0+N_1}.
$$
Then the algorithm continues in the same way choosing $w_2$, $\ldots$, $w_k$. It is easy to prove that this algorithm generates uniformly, at random, a string accepted by $\nA$ of length $k$.

Notice that the previous idea is essentially the same as the one in \cite{jerrum1986random}, that is, we use the fact that the relation $\unfa$ is self-reducible and its counting problem can be solved efficiently.
%For the details, we refer to \cite{jerrum1986random}, but we will include a clarifying note here.
However, a clarifying note should be included here.
%Notice that
We claim a polynomial-time randomized algorithm for $\GEN(\unfa)$, while an almost-uniform generator is claimed
%while it is described
in \cite{jerrum1986random}. Our result is stronger for two reasons. First of all, we have a stronger counting result (exact polynomial-time algorithm instead of an FPRAS) to use as the basis of our uniform generation algorithm. Second, the computational model considered in \cite{jerrum1986random} 
%works with a computational model 
(the Probabilistic Turing Machine) is a bit different from the one
considered in this work.
%little restricted. 
It cannot, for example, simulate a Bernoulli experiment with a success
probability of exactly $\frac{1}{3}$. Essentially, it makes it
impossible to get an exact uniform generation algorithm. We are less
strict with our computational model, so we are able to get a
polynomial-time randomized algorithm for $\GEN(\unfa)$.
}

%\subsubsection{There exist a polynomial-time algorithm for $\COUNT(\unfa)$ and a polynomial-time randomized algorithm for $\GEN(\unfa)$}
%
%Consider \review{$\nAun$} as defined in the previous paragraphs. As we already pointed out, the number of paths of length $k$  from the initial to the final state is exactly what we want to compute: the number of strings of length $k$ accepted  by $\nA$ (the unambiguity assumption is crucial here). Since \review{$\nAun$} is a DAG, we know the number of paths between two given nodes can be computed exactly in polynomial time by dynamic programming. Hence, there exists a polynomial-time algorithm for $\COUNT(\unfa)$.
%
%To prove that there exists a polynomial-time randomized algorithm for $\GEN(\unfa)$, the idea is essentially the same as the one in \cite{jerrum1986random}, that is, we use the fact that the relation $\unfa$ is self-reducible and its counting problem can be solved efficiently. For the details, we refer to \cite{jerrum1986random}.

%% file: approximate-new.tex
% !TEX root = main.tex
% !TeX spellcheck = en_US

%\cristian{This introduction of the section should be rewritten depending on where it goes this section at the end.}

The goal of this section is to provide a proof of Theorem \ref{theo-rnl}, which considers the class $\rnl$ defined in terms of $\nl$-transducers. Given that we showed in Proposition \ref{prop_complete_problems} that $\nfa$ is complete for $\rnl$, we have by Propositions \ref{prop_reduction_preserves_properties} that Theorem \ref{theo-rnl} is a consequence of the following result.
\begin{theorem}\label{theo-all-rnl}
$\ENUM(\nfa)$ can be solved with polynomial delay, $\COUNT(\nfa)$ admits an FPRAS, and $\GEN(\nfa)$ admits a PPLVUG.
\end{theorem}

{%\color{blue}
The existence problem for $\nfa$ has as input an NFA $\nA$ and a value $k$ given in unary (as the string $0^k$), and the question to answer is whether $W_{\nfa}((\nA,0^k)) \neq \emptyset$ (that is, whether there are any solutions for $(\nA,0^k)$ according to the relation $\nfa$). It is easy to prove such a task can be solved in polynomial time, as the nonemptiness problem for NFA can be solved in polynomial time. Moreover, we proved in 
%As mentioned in Section \ref{sec-good-alg-rul}, we have that $\nfa$ is a self-reducible relation~\cite{jerrum1986random}. 
Section \ref{proof_self_reducibility} that $\nfa$ is a self-reducible relation.
%(that is, for a given  input $(\nA,0^k)$, decide whether there are any witnesses) can be solved in polynomial time.
With all that, a polynomial delay algorithm for $\ENUM(\nfa)$ can be derived from the folklore result that such an enumeration algorithm exists for a self-reducible relation, if the associated existence problem for this relation can be solved in polynomial time (a precise statement of this result can be found in Lemma 4.10
%Theorem 4.9 from
in \cite{schmidt2009enumeration}). In this section, we focus on the remaining part of the proof of Theorem \ref{theo-all-rnl}. More specifically,  we provide an algorithm that approximately counts the number of words of a given length accepted by an NFA, where this length is given in unary. This constitutes an FPRAS for $\COUNT(\nfa)$, as formally stated in the following theorem:}
\begin{theorem}\label{theo-snfa-fpras}
$\snfa$ (and, thus, $\COUNT(\nfa)$) admits a fully polynomial-time randomized approximation scheme.
\end{theorem}
\review{The algorithm mentioned in this theorem} works by simultaneously counting and doing uniform generation of
%witnesses,
solutions.
\review{Then} its existence 
%us our algorithms 
not only  gives us an FPRAS for $\COUNT(\nfa)$, but also a PPLVUG for $\GEN(\nfa)$, as formally stated in the following theorem:
%\end{sloppypar}
\begin{theorem}\label{theo-snfa-pplvug}
$\GEN(\nfa)$ admits a preprocessing polynomial-time Las Vegas uniform generator.
\end{theorem}
%As shown in the following sections, $\snfa$ plays a key role in the study of classes of relations with good properties in terms of counting, uniform generation and enumeration. In particular, understanding the complexity of $\snfa$ is of outermost importance in this study.
In the rest of this section, we prove Theorems \ref{theo-snfa-fpras} and \ref{theo-snfa-pplvug}. More specifically, we start by providing in Section \ref{sec-approximate-technique} an overview of the algorithmic techniques used in the proof of Theorem \ref{theo-snfa-fpras}. Then we present in 
Section~\ref{sec-approximate-template} the template for the FPRAS for $\snfa$, whose main components are given in Sections~\ref{subsec:sets} and \ref{subsec:sampling}. A complete version of the FPRAS for $\snfa$ is finally given in Section \ref{subsec:algo}, where its correctness and polynomial-time complexity are established.  Moreover, the proof of Theorem~\ref{theo-snfa-pplvug} is also given in Section \ref{subsec:algo}.

\subsection{An overview of the algorithmic techniques} \label{sec-approximate-technique} 

\input{approximate-technique}

\subsection{The algorithm template} \label{sec-approximate-template}

\input{approximate-template}

\subsection{Computing an estimate for a set of vertices} \label{subsec:sets}

\input{approximate-sets}

\subsection{Uniform sampling from a vertex} \label{subsec:sampling}

\input{approximate-sampling}

\subsection{Bounding the probability of breaking the main assumption} \label{subsec:bounding}

\input{approximate-bounding}

\subsection{The main algorithm, its correctness and its complexity} \label{subsec:algo}

\input{approximate-algo}

%% file: approximate-technique.tex
% !TEX root = main.tex
% !TeX spellcheck = en_US

We start by providing a high-level overview of our FPRAS for the $\snfa$ problem. To this end, we first set the necessary terminology to refer to this counting problem. 

\begin{sloppypar}
A non-deterministic finite automaton (NFA) $\nA$ over the alphabet \review{$\{0,1\}$} is given as a tuple $(\nQ, \{0,1\}, \Delta, \nI, \nF)$, where $\nQ$ is a finite set of states, $\Delta \subseteq \nQ \times \{0,1\} \times \nQ$ is the transition relation, $\nI \subseteq \nQ$ is a set of initial states and $\nF \subseteq \nQ$ is a set of final states. The language of the strings in $\{0,1\}^*$ that are accepted by $\nA$ is denoted by $\nL(\nA)$. Moreover, given a natural number $n$, the language $\nL_n(\nA)$ is defined as $\nL(\nA) \cap \{0,1\}^n$. With this terminology we define the counting problem $\snfa$ as follows. \review{The input of $\snfa$ is an NFA $\nA$ with $m$ states over the alphabet $\{0,1\}$  and a natural number $n$, and the task is to return $|\nL_n(\nA)|$. Here, $n$ is given in unary (that is, $n$ is given as the string $0^n$)}
\footnote{As mentioned before, it is known that $\snfa$ belongs~to~$\shp$. Notice that the fact that $n$ is given in unary is necessary to show this property. If $n$ is given as a binary number, then the value $|\nL_n(\nA)|$ can be double exponential in the size $O(\log n)$ of this input, since $|\nL_n(\nA)|$ can be equal to $2^n$. Hence, $\snfa$ cannot be in $\shp$ if the input $n$ is given as a binary number, as if a function $f : \{0,1\}^* \to \N$ is in $\shp$, then there exists a polynomial $p(u)$ such that for every $w \in \{0,1\}^*$, it holds that $f(w) \leq 2^{p(|w|)}$.}.
\end{sloppypar}

%We start by providing a high-level overview of our FPRAS for the $\snfa$ problem. 
To illustrate the difficulty of $\snfa$, we first consider the simpler problem of counting the number of strings $|\cL_n(A)|$ of length $n$ contained in the language $\cL(A)$ accepted by a \textit{deterministic finite automaton} (DFA) $A$. Note that if $w \in \cL_n(A)$, there is exactly one accepting path in the DFA for $w$. So to count $|\cL_n(A)|$, one can simply compute the total number of paths of length $n$ in the DFA, which can be done in polynomial time by a dynamic program. However, if $\cL_n(A)$ is instead the language accepted by an NFA, then $w \in \cL_n(A)$ can have exponentially many accepting paths, and so counting paths does not lead to a good estimate of $|\cL_n(A)|$ for an NFA.

One natural approach to overcome the aforementioned issue is to design an algorithm to estimate the \textit{ambiguity} of the NFA. For instance, the following procedure produces an unbiased estimator of  $|\cL_n(A)|$. First, sample a random path of length $n$ in the NFA, and let $w$ be the string accepted on that path. Second, count the number of accepting paths $P_w$ that $w$ has in the NFA, and also count the total number of paths $P$ of length $n$ in the NFA. Repeat this process $N$ times, and report the average value of $P/P_w$. The resulting estimator is indeed unbiased. However, the number of paths $P_w,P_{w'}$ can differ by an exponential factor for different strings $w,w'$, thus the variance of this estimator is exponential. Therefore, this algorithm requires exponentially many samples to obtain a good estimate. Several other similar estimators exist (see e.g. \cite{KSM95}), which all unfortunately do not lead to polynomial time algorithms for the general $\snfa$ problem.

The basic approach of our FPRAS is to incrementally estimate, for each state $q$ in the NFA, the number of distinct strings $w$ for which there is a path of length $\alpha$ from the starting states to $q$ labeled by $w$. Call this set of strings $\nL(q^\alpha)$. Our high level approach is similar to dynamic programming. Namely, to estimate $|\nL(q^\alpha)|$, we first estimate $|\nL(p^{\alpha-1})|$ for each state $p$ such that there is a transition from $(p,a,q)$ in the NFA, where $a \in \{0,1\}$. However, one cannot simply declare
\begin{eqnarray*}
	|\nL(q^\alpha)| & = & \sum_{p \,:\, (p ,a, q) \in \Delta} |\nL(p^{\alpha-1})|, 
\end{eqnarray*}
because a single string $w$ can be in many of the sets $\nL(p^{\alpha-1})$, which would result in over-counting. Therefore, we must also estimate the \textit{intersections} of the sets $\nL(p^{\alpha-1})$. This is challenging, as these sets themselves can be exponentially large, so we cannot afford to write them down. Moreover, there are \review{$2^m$} possible sets which can arise as the intersection of sets of the form $\nL(q^\alpha)$ for $q \in Q$, thus we cannot store an estimate of each.
 The main insight of our FPRAS is to \textit{sketch} the intermediate states  $\nL(p^{\alpha-1})$ of the dynamic program, by replacing the set $\nL(p^{\alpha-1})$ with a small (polynomial-sized) uniformly sampled set $S(p^{\alpha-1}) \subseteq \nL(p^{\alpha-1})$. Here, the sketch $S(p^{\alpha-1})$ acts as a compact representation of the (possibly) larger set $\nL(p^{\alpha-1})$. For instance, to see how such a sketch could be useful, if there were exactly two preceding states $(p_1,a, q)$ and $(p_2,a, q)$, to estimate the relative size of the intersection $|\nL(p_1^{\alpha-1}) \cap \nL(p_2^{\alpha-1})|/|\nL(p_1^{\alpha-1})|$, it will suffice to use the approximation $\tilde{I} = |S(p_1^{\alpha-1}) \cap \mathcal{L}(p_2^{\alpha-1})|/|S(p_1^{\alpha-1})|$. Notice that the quantity $ |S(p_1^{\alpha-1}) \cap \mathcal{L}(p_2^{\alpha-1})|$ can be computed in time polynomial in $|S(p_1^{\alpha-1})|$, by checking for each $w \in S(p_1^{\alpha-1})$ if $w$ is contained in $ \mathcal{L}(p_2^{\alpha-1})$, which can be accomplished in polynomial time by a membership query for NFAs.  If $N(p_1^{\alpha-1}),N(p_1^{\alpha-1})$ are our estimates of $|\nL(p_1^{\alpha-1})|,|\nL(p_2^{\alpha-1})|$, then we can therefore obtain an estimate of $|\nL(q^{\alpha})|$ by $N(q^{\alpha}) = N(p_1^{\alpha-1}) +  N(p_2^{\alpha-1}) - \tilde{I}\cdot N(p_1^{\alpha-1})$, avoiding the issue of overcounting the intersection.

%If $s_{acc}$ is the accepting state of the NFA, this will eventually result in our final approximation of the size of $|U_{s_{acc}}^n|$, which is indeed the number of strings of length $n$ accepted by the NFA. 

The main technical hurdle that remains is to determine how to uniformly sample a string $w$ from a set $\nL(q^\alpha)$ to construct our sketches $S(p^{\alpha})$. This is accomplished by sampling the string $w$ bit by bit. We first partition $\nL(q^\alpha)$ into the set of strings with last bit equal to $0$ and $1$. We then estimate the size of both partitions, and choose a partition with probability proportional to its size. Finally, we store the bit corresponding to the sampled partition, append it to a \textit{suffix} $w'$ of the string $w$ and then recurse onto the next bit. In essence, we sample a string $w$ by growing a suffix of $w$. 

To estimate the size of the partitions, we use our sketches $S(p^\beta)$ of $\nL(p^\beta)$ for all $\beta \leq \alpha$ and states $p$.  Unfortunately, because of the error in estimating the sets $|\nL(p^\beta)|$, there will be some error in the distribution of our sampler. To correct this, and avoid and exponential propagation of this error, we \review{use} a rejection sampling technique of Jerrum, Valiant, and Vazirani \cite{jerrum1986random}, which normalizes the distribution and results in a perfectly uniform sample. This allows for our construction of the sketches $S(q^\alpha)$, and also gives an algorithm for \review{the} \textit{uniform generation} of strings of length $n$ from an~NFA.

%% file: approximate-template.tex
% !TEX root = main.tex

%As mentioned in Section \ref{nlclass}, we consider the approximation problem $\snfa$.
%Recall that 
The input of $\snfa$ is an NFA $\nA = (\nQ, \{0,1\}, \Delta, \nI, \nF)$ 
%on the alphabet $\Sigma = \{0,1\}$ 
with $m$ states, a string $0^n$ that represents a natural number $n$ given in unary, and an error $\neps \in (0,1)$. The problem then is to return a value $N$ such that $N$ is a $(1 \pm \neps)$-approximation of $|\nL_n(\nA)|$, that is,
\begin{eqnarray*}
(1 - \neps)|\nL_n(\nA)| \ \leq \ N \ \leq \ (1 + \neps)|\nL_n(\nA)|.
\end{eqnarray*} 
%where $\nL(\nA)$ is the set of strings accepted by $\nA$ and $\nL_n(\nA) = \{w \in \{0,1\}^* \mid w \in \nL(\nA)$ and $|w| = n \}$. 
Besides, such an approximation should be returned in time polynomial in $m$, $n$ and $\frac{1}{\neps}$.

Our algorithm for approximating $|\nL_n(\nA)|$ first involves the construction of a labelled directed acyclic graph from the NFA $\nA$.
We call this 
%directed acyclic 
graph $\nAun$, as it is obtained by unrolling $n$ times the NFA $\nA$.
%by unrolling every path of length $k \leq n$ such that it touches $k$ distinct vertices. We obtain $N_{\unroll}$ by unrolling the NFA $N$ $n$ times. 
%Recall that $\nA =(\nQ, \Sigma, \Delta, \nI, \nF)$. Then 
Specifically, for every state $q \in Q$ create $n+1$ copies $q^0, q^1,\ldots,q^n$ of $q$, and include them as \review{vertices} of $\nAun$. Moreover, for every transition $(p,b,q)$ in $\Delta$,
% and $b \in \{0,1\}$, 
create the edge $(p^\alpha, b, q^{\alpha+1})$ in $\nAun$, for every $\alpha \in [0,n-1]$. We refer to the set $Q^\alpha = \{q^\alpha \mid q \in Q\}$ as the $\alpha$-th layer of $\nAun$. Furthermore, for every set $P \subseteq Q$, we denote by $P^\alpha$ the copy of $P$ in the $\alpha$-th layer of $\nAun$. This means that $\nI^0$ refers to the initial states of $\nA$ at the first layer, and $F^n$ refers to the final states of $\nA$ at the last layer.
For the sake of presentation, 
we will use the terms \textit{vertex} and \textit{state} interchangeably to refer to the vertices of $\nAun$. Moreover, 
%for the purpose of the algorithm, 
from now on we assume that $\nAun$ is pruned, that is, for every $q \in Q$ and every $\alpha \in [0, n]$, there exists a path from some vertex of $\nI^0$ to $q^\alpha$. In other words, all states in $\nAun$ are connected to some initial state. The pruning of $\nAun$ can be done in a pre-processing step in polynomial-time in $nm$,
% over $\nAun$, 
without changing the overall time of the algorithm. We remark that in the remainder of the section, whenever we state that we run a procedure for $p^\alpha$ with $p \in Q$ and $\alpha$ some layer, it is implicitly assumed that all pruned states $p^\alpha$ have already been removed. Thus, for the remainder, we will not consider the pruned states at any point, since they cannot be used to derive words of length $n$ in the language.
%\marcelo{Couldn't be the case that for some $q \in Q$ and $\alpha \neq \beta$, it holds that $q^\alpha$ is in $\nAun$ but $q^\beta$ is not in $\nAun$? How is this pruning done? We need to explain this (in particular, I think we need to assume that $F \subsetneqq Q$ to do he construction).}

Given a state $q$ and a layer $\alpha$, we define $\nL(q^\alpha)$ as the set of all strings $w$ such that there exists a path labeled with $w$ from some vertex in $\nI^0$ to $q^\alpha$. Notice that $|w| = \alpha$ for every $w \in \nL(q^\alpha)$, and also that $\nL(q^\alpha) \neq \emptyset$ since $\nAun$ is pruned. 
We extend this notation to every set $P \subseteq Q$, namely, $\nL(P^\alpha) = \bigcup_{q \in P} \nL(q^\alpha)$. 
The sets of strings $\nL(P^\alpha)$ will be crucial for our algorithm. Indeed, finding an approximation for $|\nL_n(\nA)|$ is reduced to finding an estimate for $|\nL(F^n)|$, where $F^n$ represents the set of final states of $A$ 
%the final states in
at the last layer.   

\review{
The components of 
%the main 
our approximation algorithm are as follows. Fix the value $\kc = \lceil \frac{nm}{\neps}\rceil$ and assume that $n \geq 2$ and $m \geq 2$ (if $n \leq 1$ or $m \leq1$, then the problem can be easily solved in polynomial time). Then for each layer $\alpha$ and each state $q$ with $q^\alpha$ in $\nAun$, store a number $N(q^\alpha)$ and a set $S(q^\alpha) \subseteq \nL(q^\alpha)$ such that:
\begin{itemize}
	\item $N(q^\alpha)$ is a $(1\pm \kc^{-2})^\alpha$-approximation of $|\nL(q^\alpha)|$, and 
	\item $S(q^\alpha)$ is a uniform sample from $\nL(q^\alpha)$ of size $2\kc^7$. 
\end{itemize}

For the first requirement, we mean that 
\begin{eqnarray*}
(1 - \kc^{-2})^\alpha |\nL(q^\alpha)| \ \leq \ N(q^\alpha) \ \leq \ (1+ \kc^{-2})^\alpha|\nL(q^\alpha)|.
\end{eqnarray*} 
In particular, if $\alpha = 0$, we should have that $N(q^\alpha) = |\nL(q^\alpha)|$.
For the last requirement, we mean that each $w  \in S(q^\alpha)$ is a uniform and independent sample from $\nL(q^\alpha)$. 
%Since the samples will be uniform and independent, 
Given this condition on the samples,
it is possible that we will obtain duplicates 
%samples 
of a given $w \in \nL(q^\alpha)$. Besides, if $|\nL(q^\alpha)| < 2\kc^7$, then we know that $S(q^\alpha)$ has to contain duplicate elements. 
Therefore, we allow $S(q^\alpha)$ to be a multiset (meaning that the strings $w$ in $S(q^\alpha)$ are not necessarily distinct). The number $N(q^\alpha)$ and the set $S(q^\alpha)$ can be understood as a ``sketch'' of $\nL(q^\alpha)$ that will be used to compute other estimates for $\nAun$. 
%In the following, we say that (C1) holds for a level $\alpha$ if  $N(q^\alpha)$ is a $(1\pm \kc^{-2})^\alpha$-approximation of $|\nL(q^\alpha)|$ for all  $q \in Q$. Similarly,  we say that (C2) holds for a level $\alpha$ if we have a set $S(q^\alpha)$ containing $2\kc^7$ uniform samples from $\nL(q^\alpha)$, for all $q \in Q$.

%\rajesh{It might be good to just have the uniform samples $S(q^{\alpha})$ be called a sketch, and the estimate N be called an "estimate". That way, when you say "sketch", it is formally defined what you are referring to.}
%\marcelo{Given the notation used in the following section (in particular, the definition of $\nsketch[\alpha]$), I would prefer to keep the notation as it is.}
}

The algorithm proceeds like a dynamic programming algorithm, computing $N(q^\alpha)$ and $S(q^\alpha)$ for every state $q^\alpha$ in $\nAun$ in a breadth-first search ordering. We first compute $N(q^0),S(q^0)$ for all states $q^0$ at layer $0$. 
Then, given $\bigcup_{\beta=0}^{\alpha-1} \bigcup_{p \in Q} \{N(p^\beta),S(p^\beta)\}$, we compute $N(q^\alpha),S(q^\alpha)$ for each vertex $q^\alpha$. So the value $N(q^\alpha)$ and the set $S(q^\alpha)$ are computed layer by layer. The final estimate for $|\nL(F^n)|$ is $N(F^n)$.
% (to be defined soon). 
We summarize this algorithmic template in Algorithm \ref{fig:FPRAStemplate}.
\begin{figure}[t]
	\begin{Frame}[\textbf{Algorithm \ref{fig:FPRAStemplate}}: \ Algorithmic Template for our FPRAS  	  ]
		\label{fig:FPRAStemplate}
		\begin{enumerate}
			\item Construct the labelled directed acyclic graph $\nAun$ from an input NFA $\nA$ and string $0^n$, where $A = (\nQ, \{0,1\}, \Delta, \nI, \nF)$.
			\item For layers $\alpha=0,1,\dots,n$ and states $q \in Q$:
			\begin{enumerate}
				\item Compute $N(q^\alpha)$ given $\bigcup_{\beta=0}^{\alpha-1} \bigcup_{p \in Q} \{N(p^\beta),S(p^\beta)\}$. For $\alpha=0$, the value $N(q^\alpha)$ is computed without any additional information.
				\item Call a subroutine to sample polynomially many uniform elements from $\nL(q^\alpha)$ using the value $N(q^\alpha)$ and the elements $\bigcup_{\beta=0}^{\alpha-1} \bigcup_{p \in Q} \{N(p^\beta),S(p^\beta)\}$. 
				\item Let $S(q^\alpha) \subseteq \nL(q^\alpha)$ be the multiset of uniform samples obtained.
			\end{enumerate}
			\item Return $N(F^n)$ given $\bigcup_{\beta=0}^{n} \bigcup_{p \in Q} \{N(p^\beta),S(p^\beta)\}$.
		\end{enumerate}
	\end{Frame}
\end{figure}
\review{
For the rest of this section, we show how to instantiate the template of our algorithm.
For a layer~$\alpha$,  we show in Section~\ref{subsec:sets} how to compute the estimate $N(q^\alpha)$ given estimates $N(q^\beta)$ and sets $S(q^\beta)$ for all $\beta < \alpha$. In fact, for this we need to assume a strong condition (introduced in the next section), which states that the samples in our set $S(q^\alpha)$ satisfy good concentration properties. 
%This implements step $(2a)$ of the template shown in Algorithm \ref{fig:FPRAStemplate}. 
Next, given  $N(q^\alpha)$ and the prior estimates $N(q^\beta)$ and sets $S(q^\beta)$, we demonstrate in Section~\ref{subsec:sampling}
how to generate a uniform sample from the set $\nL(q^\alpha)$, proving how to compute $S(q^\alpha)$. 
In particular, again, we will show that the strong condition used as an induction hypothesis holds for the sets $S(q^\alpha)$ with exponentially large probability over $\kc$ (Section~\ref{subsec:bounding}). 
In the last section we put all pieces together and show the correctness of the algorithm. 
}

%% file: approximate-sets.tex
% !TEX root = main.tex
% !TeX spellcheck = en_US

Recall that the input of the problem is an NFA $\nA = (\nQ, \{0,1\}, \Delta, \nI, \nF)$ with $m$ states and a string $0^n$, 
%\rajesh{It might be good to explain why we get it in unary -- some people may not understand right away and get confused.}
%\marcelo{I added a footnote in the preliminaries to explain why $n$ is given in unary.}
and that we assume that $m \geq 2$ and $n \geq 2$.
\review{
Then fix a layer $\alpha$, and define  a sketch data structure such that $\nsketch[\alpha] := \{N(p^\beta),S(p^\beta)\}_{p\in Q, \beta \leq \alpha}$. Moreover, assume that $\nsketch[\alpha]$ has already been computed. In particular, $N(p^\beta)$ is a $(1\pm \kc^{-2})^\beta$-approximation of $|\nL(p^\beta)|$, and 
$S(p^\beta)$ is a uniform sample from $\nL(p^\beta)$ of size $2\kc^7$ for each $\beta \leq \alpha$. }
The goal of this section is twofold; we first show how to compute an estimate of $|\nL(P^\alpha)|$ for every $P \subseteq Q$, which is denoted by $N(P^\alpha)$, and then 
%In fact, the only purpose of $\nsketch[\alpha]$ is for deriving $N(P^\alpha)$ for any $P$. 
we show how to compute an estimate for $N(q^{\alpha+1})$. 
\review{These values $N(P^\alpha)$ will play a crucial role for computing not only $N(q^{\alpha+1})$, but also the set of uniform samples $S(q^{\alpha+1})$ and
%, moreover, 
the final estimate $N(F^n)$ for $|\nL(F^n)|$ (see Sections~\ref{subsec:sampling} and~\ref{subsec:bounding}).}
%To meet these goals, we assume in this section that each $S(p^\beta)$ is a (perfect) uniform sample of size $2 \kc^7$ of $\nL(p^\beta)$, where $k = \lceil \frac{nm}{\eps}\rceil$. Then, in the next section, we show how to compute this set.  

Let $P$ be a non-empty subset of $Q$, and suppose that we want to find an estimate $N(P^\alpha)$ for $|\nL(P^\alpha)|$. If $\nA$ is deterministic, then the sets $\{\nL(p^\alpha) \mid p \in P\}$ are disjoint, and then we can easily compute the size of $|\nL(P^\alpha)|$ as $\sum_{p \in P} |\nL(p^\alpha)|$. Unfortunately, given that $\nA$ can be non-deterministic, this sum will over-approximate the size of $|\nL(p^\alpha)|$, and we need to find a way to deal with the intersections of the sets $\{\nL(p^\alpha) \mid p \in P\}$. For this, fix a total order $\prec$ over the set $P$, and consider the following way to compute $|\nL(P^\alpha)|$:
% as follows:
\begin{align}\label{eq-Palpha}\tag{$\dagger$}
|\nL(P^\alpha)| \ = \ \sum_{p\in P} |\nL(p^\alpha)| \cdot \frac{|\nL(p^\alpha) \setminus \bigcup_{q \in P \,:\, q \prec p} \nL(q^\alpha)|}{|\nL(p^\alpha)|}
\end{align}
With the ratio $|\nL(p^\alpha) \setminus \bigcup_{q \in P \,:\, q \prec p} \nL(q^\alpha)| / |\nL(p^\alpha)|$, we removed from $\nL(p^\alpha)$ its intersection with all sets $\nL(q^\alpha)$ such that $q \prec p$.
In fact, one can easily check that $|\nL(P^\alpha)| = \sum_{p\in P} |\nL(p^\alpha) \setminus \bigcup_{q \in P\,:\, q \prec p} \nL(q^\alpha)|$ and, thus, equation~\eqref{eq-Palpha} trivially holds. We call the above ratio the \emph{intersection rate} of $p^\alpha$ in $P$ given $\prec$ (or just the intersection rate of $p^\alpha$).

Inspired by equation~\eqref{eq-Palpha}, we can estimate $|\nL(P^\alpha)|$ by using $N(p^\alpha)$ to estimate $|\nL(p^\alpha)|$ and 
%the set 
$S(p^\alpha)$ to estimate the intersection rate of $p^\alpha$. More precisely, we define the estimate $N(P^\alpha)$ for $|\nL(P^\alpha)|$ as follows:
\begin{align}\tag{$\ddagger$}
N(P^\alpha) \ = \ \sum_{p\in P} N(p^\alpha) \cdot \frac{|S(p^\alpha) \setminus \bigcup_{q \in P\,:\, q \prec p} \nL(q^\alpha)|}{|S(p^\alpha)|} \label{eq:remove-intersection}
\end{align}
It is important to note that $N(P^\alpha)$ can be computed in polynomial time in the size of $\nsketch[\alpha]$. Indeed, the set $S(p^\alpha) \setminus \bigcup_{q \in P \,:\, q \prec p} \nL(q^\alpha)$ can be computed by iterating over each string $w \in S(p^\alpha)$ and checking whether $w \in \nL(\{q^\alpha \mid q \in P$ and $q \prec p\})$.
% or not. 
Given that verifying if a string is in 
%$\bigcup_{q \in P \,:\, q \prec p} \nL(q^\alpha)$ 
$\nL(\{q^\alpha \mid q \in P$ and $q \prec p\})$ can be done in polynomial time, computing $N(P^\alpha)$ takes polynomial time as well. We call the ratio $|S(p^\alpha) \setminus \bigcup_{q \in P\,:\, q \prec p} \nL(q^\alpha)|/|S(p^\alpha)|$ the \emph{estimate of the intersection rate~of~$p^\alpha$}.  

To show that $N(P^\alpha)$ is a good estimate for $|\nL(P^\alpha)|$, we need that the estimate of the intersection rate is a good approximation of the real intersection rate in each layer. By a good approximation, we mean that the following condition holds at level $\alpha$:
$$
\mathcal{E}(\alpha) \ \ := \ \  \forall q \in Q \ \forall P \subseteq Q. \quad  
\bigg| 
\frac{|\nL(q^\alpha) \setminus \bigcup_{p \in P} \nL(p^\alpha)|}{|\nL(q^\alpha)|} 
-
\frac{|S(q^\alpha) \setminus \bigcup_{p \in P} \nL(p^\alpha)|}{|S(q^\alpha)|}
\bigg| \ < \ \frac{1}{\kc^{3}}
$$
This condition is crucial for the next results and most of our analysis in this and next section will assume that this condition holds. 
Towards the end, in Section~\ref{subsec:algo} we will show that, by Hoeffding's inequality, the condition $\mathcal{E}(\alpha)$ holds for all layers $\alpha$ with exponentially high probability \review{over $\kc$}. %In addition, we define
%\begin{align*}
%\mathcal{P}(\alpha)\ \ := \ \  \forall q \in Q.  \quad   (1-k^{-2})^\alpha %|\mathcal{L}(q^\alpha)| \leq   N(q^\alpha) \leq (1+k^{-2})^\alpha %|\mathcal{L}(q^\alpha)| 
%\end{align*}
%Unfortunately, we need to postpone this proof to Section~\ref{subsec:algo} when all the pieces of the main algorithm are together. 
%Notice that $\mathcal{P}(\alpha)$ formally states that the condition $(C1)$ holds at the level $\alpha$. 
\review{
Next, we prove that, if condition $\mathcal{E}(\alpha)$ holds, then
%we can show that 
$N(P^\alpha)$ is a good estimate for $|\nL(P^\alpha)|$.}
\review{
\begin{proposition}\label{prop:estimate-sets}
	Assume that $\mathcal{E}(\alpha)$ holds and $N(p^\alpha)$ is a $(1\pm \kc^{-2})^\alpha$-approximation of $|\nL(p^\alpha)|$ for every $p \in Q$. Then $N(P^\alpha)$ is a $(1\pm \kc^{-2})^{\alpha+1}$-approximation of $|\nL(P^\alpha)|$ for every $P \subseteq Q$.
\end{proposition}}
\begin{proof}
	\review{Given that condition $\mathcal{E}(\alpha)$ holds, we know that for each $p \in P$}: 
	\begin{align*}
	&\frac{|\nL(p^\alpha) \setminus \bigcup_{q \in P \,:\, q \prec p} \nL(q^\alpha)|}{|\nL(p^\alpha)|}- \kc^{-3} \ <\\
	&\hspace{100pt} \frac{|S(p^\alpha) \setminus \bigcup_{q \in P \,:\, q \prec p} \nL(q^\alpha)|}{|S(p^\alpha)|} \ <\\ 
	&\hspace{200pt} \frac{|\nL(p^\alpha) \setminus \bigcup_{q \in P \,:\, q \prec p} \nL(q^\alpha)|}{|\nL(p^\alpha)|} + \kc^{-3}
	\end{align*}
	Moreover, given that $N(p^\alpha)$ is a $(1\pm \kc^{-2})^\alpha$-approximation of $|\nL(p^\alpha)|$, it holds that:
	\begin{align*}
		(1 - \kc^{-2})^{\alpha} |\nL(p^\alpha)| \ \leq \ N(p^\alpha) \ \leq \ (1 + \kc^{-2})^{\alpha} |\nL(p^\alpha)|.
	\end{align*}
	Putting these two bounds together, we obtain the following bounds from the definition of $N(P^\alpha)$ in equation \eqref{eq:remove-intersection}:
	\begin{multline*}
	(1 - \kc^{-2})^{\alpha}  \sum_{p \in P}   \bigg(|\nL(p^\alpha) \setminus \bigcup_{q \in P \,:\, q \prec p} \nL(q^\alpha)| -\kc^{-3} |\nL(p^\alpha)|\bigg) \ < \ N(P^\alpha) \ <\\
	(1 + \kc^{-2})^{\alpha} \sum_{p \in P} \bigg( |\nL(p^\alpha) \setminus \bigcup_{q \in P \,:\, q \prec p} \nL(q^\alpha)| + \kc^{-3} |\nL(p^\alpha)| \bigg).	
	\end{multline*}
	Recall from the discussion of the intersection rate that $|\nL(P^\alpha)| = \sum_{p\in P} |\nL(p^\alpha) \setminus \bigcup_{q \in P \,:\, q \prec p} \nL(q^\alpha)|$. 
	%Since $\nL(P^\alpha) = \bigcup_{p \in P} \nL(p^\alpha)$, 
	Moreover, given that $\nL(p^\alpha) \subseteq \nL(P^\alpha)|$ and $|P| \leq m \leq \kc$, we have that $\sum_{p\in P} |\nL(p^\alpha)| \leq \sum_{p\in P} |\nL(P^\alpha)| = |P| \cdot |\nL(P^\alpha)| \leq \kc \cdot |\nL(P^\alpha)|$.
	%we also have that $\sum_{p\in P} |\nL(p^\alpha)| \leq k \cdot |\nL(P^\alpha)|$. 
	Replacing both statements in the previous inequality, we obtain
	\begin{align*}
	(1 - \kc^{-2})^{\alpha} ( |\nL(P^\alpha)| -\kc^{-3} \cdot \kc |\nL(P^\alpha)|) \ < \ N(P^\alpha) \ < \ (1 + \kc^{-2})^{\alpha}( |\nL(P^\alpha)| +\kc^{-3} \cdot \kc |\nL(P^\alpha)|),
	\end{align*}
	which is equivalent to
	\begin{align*}
	(1 - \kc^{-2})^{\alpha+1} |\nL(P^\alpha)| \ < \ N(P^\alpha) \ < \ (1 + \kc^{-2})^{\alpha+1} |\nL(P^\alpha)|.
	\end{align*}
	This concludes the proof of the proposition.
\end{proof}
With the estimates of $|\nL(P^\alpha)|$ for every $P\subseteq Q$ at the $\alpha$-th layer, we are ready to give a good estimate for the size $|\nL(q^{\alpha+1})|$ of \review{a} single vertex in the next layer $\alpha+1$. Let $q^{\alpha+1}$ be an arbitrary vertex at layer $\alpha+1$. For $b \in\{0,1\}$, define the set of vertices $R_b = \{p^\alpha \in Q^\alpha \mid (p^\alpha, b, q^{\alpha+1})$ is an edge in $\nAun\}$, namely, the set of all vertices in the $\alpha$-th layer from which $q^{\alpha+1}$ can be reached  
%that can reach $q^{\alpha+1}$ 
by reading symbol $b$. Notice that sets $R_0$ and $R_1$ partition $\cL(q^{\alpha+1})$ in the following sense: 
\begin{align} \label{eq-qalpha-part}
\cL(q^{\alpha+1}) \ \ = \ \ \cL(R_0)\cdot\{0\} \ \uplus\ \cL(R_1)\cdot\{1\},
%\cL(q^{\alpha+1}) \ = \ \{ w \in \cL(q^{\alpha+1}) \mid w = v0 \text{ with } v \in \cL(R_0)\}
%\text{\,$\biguplus$\,} \{ w \in \cL(q^{\alpha+1}) \mid w = v1 \text{ with } v \in \cL(R_1)\}.
\end{align}
where given two sets $S_1, S_2$ of strings, $S_1 \cdot S_2$ is defined as the set consisting of the concatenation of each string of $S_1$ with each string of $S_2$ (in particular, $\cL(R_b)\cdot\{b\} =  \{ w \in \{0,1\}^* \mid w = v \cdot b$ with $v \in \cL(R_b)\}$ for $b \in \{0,1\}$). 
Equation \eqref{eq-qalpha-part} implies that $|\cL(q^{\alpha+1})| = |\cL(R_0)| + |\cL(R_1)|$. Notice that, if we assume $\mathcal{E}(\alpha)$ holds, then by Proposition~\ref{prop:estimate-sets} we have that $N(R_b)$ is a $(1\pm \kc^{-2})^{\alpha+1}$-approximation of $|\nL(R_b)|$ for $b \in \{0,1\}$, from which we obtain that 
$N(q^{\alpha+1}) = N(R_0) + N(R_1)$ is a $(1\pm \kc^{-2})^{\alpha+1}$-approximation of $|\nL(q^{\alpha+1})|$.
% (assuming that ($\star$) holds).
Therefore, we can derive an estimate $N(q^{\alpha+1})$ for $|\cL(q^{\alpha+1})|$ by using previous estimates $\{p^{\alpha}\}_{p\in Q}$.

\review{
Note that the computation of $N(q^{\alpha+1})$ is deterministic, by assuming that $\mathcal{E}(\beta)$ holds for all $\beta \leq \alpha$. Specifically, the estimates $N(q^{0})$ are exact for the initial layer. Next, for each layer $\alpha$ we assume that $\mathcal{E}(\alpha)$ holds and we can compute $N(q^{\alpha+1})$ by using  $\{N(p^{\alpha})\}_{p \in Q}$ (in fact, by using $\{N(P^{\alpha})\}_{P \subseteq Q}$). Then, we asumme that $\mathcal{E}(\alpha+1)$ holds and so on. Therefore, by filling the sets $\{S(p^{\beta})\}_{p\in Q,\beta \leq \alpha}$ with uniform samples and assuming that $\mathcal{E}(\beta)$ holds for all $\beta \leq \alpha$ we can compute each estimate $N(q^{\alpha+1})$. Moreover, we can guarantee that it is a $(1\pm \kc^{-2})^{\alpha+1}$-approximation of $|\nL(q^{\alpha+1})|$.
We summarize this fact in the following proposition.
}
\review{
\begin{proposition}\label{prop:estimate-sets2}
\begin{sloppypar}
Assume that $\mathcal{E}(\beta)$ holds for all $\beta \leq \alpha$. Then $N(p^{\alpha+1})$ is a $(1\pm \kc^{-2})^{\alpha+1}$-approximation of $|\nL(p^{\alpha+1})|$ for every $p \in Q$.
\end{sloppypar}
\end{proposition}
}

%By assuming that $\mathcal{E}(\alpha)$ holds and $N(p^\alpha)$ is a good estimate, we can compute the $N(p^{\alpha+1})$ for the next layer $\alpha+1$. Furthermore, in the next section we show how to compute the set $S(p^{\alpha+1})$, which leads to compute the next layer $\alpha+2$, and so on. 
\review{
After all, at some point we will reach the last layer $n$ and we would like to compute the $(1\pm \neps)$-approximation for $|\nL_n(\nA)|$. 
For this, we can use $N(F^n)$ for estimating $|\nL_n(\nA)|$, which reach the ultimate goal of our algorithm.
\begin{proposition}\label{prop:final-estimate}
	If $\mathcal{E}(\beta)$ holds for all $\beta \leq n$, then $N(F^n)$ is a $(1\pm \neps)$-approximation for $|\nL_n(\nA)|$. 
\end{proposition}
\begin{proof}
	Assume that $N(F^n)$ is a $(1\pm \kc^{-2})^{n+1}$-approximation of $|\nL(F^n)| = |\mathcal{L}_n(\nA)|$, that is,
	\begin{eqnarray*}
	&(1 - \kc^{-2})^{n+1}|\mathcal{L}_n(\nA)| \ \leq \ N(F^n) \ \leq \ (1 + \kc^{-2})^{n+1}|\mathcal{L}_n(\nA)|.&
	\end{eqnarray*}
	But we have that:
	\begin{eqnarray*}
		(1 + \kc^{-2})^{n+1} & \leq & \bigg(1 + \bigg(\frac{\neps}{mn}\bigg)^{2}\bigg)^{n+1}\\
		& = & \bigg[\bigg(1 + \bigg(\frac{1}{(\frac{nm}{\neps})^{2}}\bigg)\bigg)^{(\frac{nm}{\neps})^{2}}\bigg]^\frac{(n+1)\neps^{2}}{n^{2} m^{2}} \\
		& \leq & e^{\frac{\neps^{2}}{{m}^{2}}}\\
		& \leq & 1 + 2 \frac{\neps^{2}}{{m}^{2}} \quad \text{since $e^x \leq (1+2x)$ for $x \in [0,1]$}\\
		& = & 1 + \neps \cdot \frac{2 \neps}{{m}^{2}}\\
		& \leq & 1 + \neps \quad \quad \text{since $m \geq 2$ and $\neps \in (0,1)$}
	\end{eqnarray*}
	and we also have that:
	\begin{eqnarray*}
		(1  - \kc^{-2})^{n+1} & \geq & \bigg(1 - \bigg(\frac{\neps}{mn}\bigg)^{2}\bigg)^{n+1}\\
		& = & \bigg[\bigg(1 - \bigg(\frac{1}{(\frac{nm}{\neps})^{2}}\bigg)\bigg)^{(\frac{nm}{\neps})^{2}}\bigg]^\frac{(n+1)\neps^{2}}{n^{2} m^{2}} \\
		& \geq & (e^{-2})^\frac{\neps^2}{{m}^{2}} \quad \text{since $\bigg(1-\frac{1}{x}\bigg)^x \geq e^{-2}$ for $x \geq 2$}\\
		& \geq & 1 - \frac{2 \neps^{2}}{{m}^{2}} \quad \ \ \text{since $e^{-x} \geq 1-x$ for $x \geq 0$}\\
		& = & 1 - \neps \cdot \frac{2 \neps}{{m}^{2}}\\
		& \geq & 1 - \neps \quad \quad \ \ \text{since $m \geq 2$ and $\neps \in (0,1)$}.
	\end{eqnarray*}
	Thus, we conclude that:
	\begin{equation*}
	(1 - \neps)|\mathcal{L}_n(\nA)| \ \leq \ N(F^n) \ \leq \ (1 + \neps)|\mathcal{L}_n(\nA)|.
	\end{equation*} 	
\end{proof}

% by adding the estimates for $|\cL(R_0)|$ and $|\cL(R_1)|$ we can derive the estimate $N(q^{\alpha+1})$ for $|\cL(q^{\alpha+1})|$ satisfying property (\ref{mprop1}) of the sketch data structure at level $\alpha+1$. 

In the following section, we show how to compute the set $S(q^{\alpha+1})$ using $\nsketch[\alpha]$, namely, how to generate a uniform sample from $\cL(q^{\alpha+1})$. Specifically, we show that assuming $\mathcal{E}(\beta)$ holds for all $\beta \leq \alpha$, we can obtain uniform samples from the sets $\nL(q^{\alpha+1})$ such that property $\mathcal{E}(\alpha+1)$ will hold with high probability. 
}

%% file: approximate-sampling.tex
% !TEX root = main.tex
% !TeX spellcheck = en_US

To carry out our main approximation algorithm, we must implement the algorithm template given in Algorithm \ref{fig:FPRAStemplate}, whose input is assumed to be an NFA $\nA = (\nQ, \{0,1\}, \Delta, \nI, \nF)$ with $m$ states and a string $0^n$, where $m \geq 2$ and $n \geq 2$. In the previous section, we implemented Step 2~(a) of this algorithm and, thus, the goal of this section is to implement the sampling subroutine in Step 2~(b).
% of Algorithm \ref{fig:FPRAStemplate}.
This procedure is based on a sample technique proposed in \cite{jerrum1986random}, but modified to suit our setting. 
%For a string $x \in \{0,1\}^*$, let $|x|$ be the length (in bits) of $x$.

\review{
Take a state $q \in Q$ and layer $\alpha \leq n$, and assume that for all layers $\beta < \alpha$ the condition $\mathcal{E}(\beta)$ holds. Notice that by Proposition~\ref{prop:estimate-sets2}, once we have $\mathcal{E}(\beta)$ and estimates for all levels $\beta < \alpha$, we immediately get the estimates $N(p^\alpha)$ for the level $\alpha$ as well.
}

The procedure to sample a uniform element of the set  $\nL(q^{\alpha})$ is as follows. We initialize a string $w^\alpha$ to be the empty string. Then we construct a sequence of strings $w^\alpha$, $w^{\alpha - 1}$, $\ldots$, $w^1$, $w^0$, where each element $w^\beta$ is of the form $b_\beta \cdot w^{\beta+1}$ with $b_\beta \in \{0,1\}$, and we define the result of the sample procedure to be $w^0$. In other words, we sample a string $w^0$ of $\nL(q^{\alpha})$ by building a \textit{suffix} of the sample, bit by bit. To ensure that $w^0$ is an element of $\nL(q^{\alpha})$ chosen uniformly, we also consider a sequence of sets $P^\alpha$, $P^{\alpha-1}$, $\ldots$, $P^1$, $P^0$ constructed as follows. The first set is  $P^\alpha = \{q^\alpha\}$.
Then we consider the set of vertices at layer $\alpha-1$ that can reach the set $P^\alpha$ by reading letter $b$, namely, for $b \in \{0,1\}$ define:
%construct the set: 
$$
P^\alpha_b \ = \  \{p^{\alpha-1} \in Q^{\alpha-1} \mid \text{there exists } r^\alpha \in P^\alpha \text{ such that } %p^{\alpha-1} \xrightarrow{b} r^{\alpha}\}.
(p^{\alpha-1}, b, r^{\alpha}) \text{ is an edge in } \nAun\}.
$$
Notice that, although we use the superscript $\alpha$, the set $P^\alpha_b$ is a subset of vertices in the $(\alpha-1)$-th layer.
\review{Similar to the previous section}, the sets $P^\alpha_0$ and $P^\alpha_1$ induce a partition \review{of the} set $\cL(P^{\alpha})$ in the following sense:
%as follows: 
$$
\cL(P^{\alpha})\  \ = \ \ \cL(P^\alpha_0)\cdot\{0\} \ \uplus \ \cL(P^\alpha_1)\cdot\{1\}
$$
Therefore, our sampling algorithm \review{estimates} the size $N(P^\alpha_b)$ of $\cL(P^\alpha_b)$ for $b \in \{0,1\}$, and \review{chooses} one of $P^\alpha_0$, $P^\alpha_1$ with probability proportional to its size, namely, $N(P^\alpha_0)/(N(P^\alpha_0) + N(P^\alpha_1))$ and $N(P^\alpha_1)/(N(P^\alpha_0) + N(P^\alpha_1))$.
%, respectively. 
Say we choose $P^\alpha_b$. Then we define $b_{\alpha-1} = \review{b}$, append the bit $b_{\alpha-1}$ as a prefix of $w^\alpha$ to obtain $w^{\alpha - 1} = b_{\alpha-1} \cdot w^\alpha$, define $P^{\alpha - 1}$ as $P^\alpha_b$, and continue with the recursion on $w^{\alpha-1}$ and $P^{\alpha - 1}$. Hence, we have that $P^{\beta}$ is the set of vertices such that there exists a path labeled by $w^\beta$ that connects some state of $P^{\beta}$ with $q^\alpha$.
Notice that $\nL(P^\beta) \neq \emptyset$ for every layer $\beta$ (and, thus, $P^\beta \neq \emptyset$). 
Indeed, given that $\nAun$ is pruned, we know that $\cL(P^\alpha) = \cL(\{q^\alpha\}) \neq \emptyset$.  By induction, if for some level $\beta$ we have that $P^{\beta}_0 = \emptyset$ (similar when $P^{\beta}_1 = \emptyset$), then $\nL(P^{\beta}_{1}) \cdot \{1\}  = \nL(P^{\beta})$ and the next level $\beta-1$ will be chosen with probability $1$. In particular, $\nL(P^{\beta}_{1}) = \nL(P^{\beta-1}) \neq \emptyset$.

Since there could be an error in estimating the sizes of the partitions, it may be the case that some items were chosen with slightly larger probability than others. To remedy this and obtain a perfectly uniform sampler, at every step of the algorithm we store the probability with which we chose a partition. Thus at the end, we have computed exactly the probability $\phi$ with which we sampled the string $w$. We can then reject this sample with probability proportional to $\phi$, which gives a perfect sampler. As long as no string is too much more likely than another to be sampled, the probability of rejection will be a constant, and we can simply run our sampler $O(\log(\frac{1}{\mu}))$-times to get a sample with probability $1-\mu$ for every $\mu > 0$. 

\begin{sloppypar}
%For the sake of presentation, we first assume that we have perfect estimates of the sizes of the partitions in question.
This procedure then is given in Algorithm \ref{fig:sample}. We call it with the initial parameters \textbf{Sample}$(\alpha, \{q^\alpha\},\, \nemptyword,\, \varphi_0)$, where $\nemptyword$ is the empty string, corresponding to the goal of sampling a uniform element of $\nL(P^\alpha) = \nL(q^\alpha)$. Here, $\varphi_0$ is a value that we will later choose. 
%Specifically, $\varphi_0$ will be a constant times a $(1 \pm \neps)$-approximation of~$|\nL(q^\alpha)|$.
%\noindent
%Note that, as defined above, 
Notice that at every step 
%$\beta$ 
of Algorithm \ref{fig:sample}, we have that $|\nL(P^\beta)|$ is precisely the number of strings in $\nL(q^\alpha)$ which have the suffix $w^\beta$, as $\nL(P^\beta)$ is the set of strings $x$ such that $x \cdot w^\beta \in \nL(q^\alpha)$.
Observe then that the set $P^\beta$ depends on the random string $w^\beta$, so in fact we could write $P^\beta_{w^\beta}$ instead of $P^\beta$. For notational simplicity we omit the subscript, and it is then \review{understood} that $P^\beta$ is a function of $w^\beta$. 
%Moreover, both $W^j,w^j$ are only defined in the context of a specific call to \texttt{Sample}$(\{s_i^\alpha\}, \emptyset, e^{-4}/R_{s_i^\alpha} )$ for some fixed state $s_i^\alpha$. Thus to avoid confusion, within our analysis we will use the state $s_i^\alpha$ to a specific state, thereafter, and $W^j,w^j$ will then be defined in the context of a specific call to \texttt{Sample}$(\{s_i^\alpha\}, \emptyset, e^{-4}/R_{s_i^\alpha} )$. 
\end{sloppypar}

\begin{figure}[t]
	\begin{Frame}[\textbf{Algorithm \ref{fig:sample}: \ Sample$(\beta,P^\beta,w^\beta,\varphi)$}]
		\label{fig:sample}
		\begin{enumerate}%[topsep=0pt,itemsep=-1ex,partopsep=1ex,parsep=1ex] 
			\item If $\beta = 0$, then with probability $\varphi$ return  $w^0$, otherwise return \fail. 
			\item Else, compute the set 
			%$P^\beta_b \ = \  \{p^{\beta-1} \mid \exists r^\beta \in P^\beta. \ \ p^{\beta-1} \xrightarrow{b} r^{\beta}\}$ 
			$P^\alpha_b \ = \  \{p^{\alpha-1} \in Q^{\alpha-1} \mid$ there exists $r^\alpha \in P^\alpha$ such that $(p^{\alpha-1}, b, r^{\alpha})$ is an edge in $\nAun\}$ for every $b \in \{0,1\}$.
			\item  Choose a partition $b \in \{0,1\}$ with probability ${\displaystyle p_b = \frac{N(P^\beta_b)}{N(P^\beta_0) + N(P^\beta_1)}}$.
			\item Set $P^{\beta-1} = P^\beta_b$, and $w^{\beta-1} = b \cdot w^\beta$. 
			\item Return \textbf{Sample}$(\beta-1, P^{\beta-1} ,  w^{\beta-1}, \frac{\varphi}{p_b})$.
		\end{enumerate}
	\end{Frame}
\end{figure}

To get some intuition of Algorithm~\ref{fig:sample}, assume for the moment that we can compute each $p_b$ exactly, namely, $p_b=|\nL(P^\beta_b)|/|\nL(P^\beta)|$. Now the probability of choosing a given element $x \in \nL(q^\alpha)$ can be computed as follows. Ignoring for a moment the possibility of returning \fail, we have that $w^0$ is the string returned by \textbf{Sample}$(\alpha, \{q^\alpha\},\, \nemptyword,\, \varphi_0)$. Thus, \review{the probability we choose} $x$ is:
\begin{eqnarray*}
	%\pr(w^0 = x)  & = & \frac{|W^{\alpha-1}|}{|\nL(P^{\alpha})|} \cdot \frac{|W^{\alpha - 2}|}{|W^{\alpha-1}|} \cdot \frac{|W^{\alpha - 3}|}{|W^{\alpha - 2}|} \cdot \cdots \cdot \frac{|W^{1}|}{|W^{2}|} \cdot \frac{1}{|W^{ 1}|}\\
	%& = & \frac{1}{|\nL(P^{\alpha})|}.
	\pr(w^0 = x)  \ =  \ \frac{|\nL(P^{\alpha-1})|}{|\nL(P^{\alpha})|} \cdot \frac{|\nL(P^{\alpha-2})|}{|\nL(P^{\alpha-1})|} \cdot \frac{|\nL(P^{\alpha-3})|}{|\nL(P^{\alpha-2})|} \cdot \cdots \cdot \frac{|\nL(P^{1})|}{|\nL(P^{2})|} \cdot \frac{1}{|\nL(P^{1})|}
	\ =  \ \frac{1}{|\nL(P^{\alpha})|}.
\end{eqnarray*}
Now at the point of return, we also have that $\varphi =  \varphi_0/\pr(w^0 = x)$. Thus, if $\varphi_0/\pr(w^0 = x) \leq 1$, then the probability that $x$ is output is simply $\varphi_0$. The following is then easily seen: 
%\marcelo{Does this claim work if $\varphi_0/\pr(w^0 = x) \leq 1$? As far as I can see the claim also holds if $\varphi_0/\pr(w^0 = x) = 1$.} Taken care of!
\begin{fact}\label{fact-sample-exact}
Assume that each probability $p_b$ in Algorithm \ref{fig:sample} satisfies that
\begin{align*}
p_b \ = \ \frac{|\nL(P^\beta_b)|}{|\nL(P^\beta)|}.	
\end{align*}
If $0 < \varphi_0 \leq \frac{1}{|\nL(P^{\alpha})|}$ and $w^0 \neq \fail$ is the output of Algorithm \ref{fig:sample}, then for every $x \in \nL(P^{\alpha})$, it holds that
	%		\begin{eqnarray*}
	%		\pr(w^0 = x) & = & \varphi_0.
	%		\end{eqnarray*}
	$\pr(w^0 = x) = \varphi_0$.
	Moreover, the algorithm outputs $w^0 = \fail$ with probability \mbox{$1 - |\nL(P^{\alpha})|\varphi_0$}.  
\end{fact}	
This shows that, conditioned on not failing, the above is a uniform sampler. 
%Repeating the procedure  $\poly(n)(|\nL(P^{\alpha})|\varphi_0)^{-1}$ times, we will get a sample with probability $1 - \exp(-\poly(n))$.
Repeating the procedure  $\ell \cdot (|\nL(P^{\alpha})|\varphi_0)^{-1}$ times, we get a sample with probability $1 - e^{-\ell}$ since:
\begin{eqnarray*}
	%(1 - |\nL(P^{\alpha})|\varphi_0)^{\ell \cdot (|\nL(P^{\alpha})|\varphi_0)^{-1}} & \leq & (e^{-|\nL(P^{\alpha})|\varphi_0})^{\ell \cdot (|\nL(P^{\alpha})|\varphi_0)^{-1}}\\
	%& = & e^{-|\nL(P^{\alpha})|\varphi_0 \cdot \ell \cdot (|\nL(P^{\alpha})|\varphi_0)^{-1}}\\
	%& = & e^{-\ell}.
	(1 - |\nL(P^{\alpha})|\varphi_0)^{\ell \cdot (|\nL(P^{\alpha})|\varphi_0)^{-1}} & \leq & (e^{-|\nL(P^{\alpha})|\varphi_0})^{\ell \cdot (|\nL(P^{\alpha})|\varphi_0)^{-1}} \ \ \ \ \ \ \ \  \ \ \ \  \text{\review{(by using $(1 - x)  \leq e^{-x}$)}}\\
	& = & e^{-|\nL(P^{\alpha})|\varphi_0 \cdot \ell \cdot (|\nL(P^{\alpha})|\varphi_0)^{-1}} \ =  \ e^{-\ell}.
\end{eqnarray*}
However, Fact \ref{fact-sample-exact} was obtained under the strong assumption that each probability $p_b$ can be computed exactly. Hence, in what follows we focus on showing that 
%that we show that 
with high probability the same result holds if we approximate $p_b$ with $N(P^\beta_b)/(N(P^\beta_0) + N(P^\beta_1))$ (instead of assuming that~$p_b = |\nL(P^\beta_b)|/|\nL(P^\beta)|$).

\review{
\begin{proposition}\label{prop:unifsamplecondition}
\begin{sloppypar}
	Assume that condition $\mathcal{E}(\beta)$ holds for every layer $\beta < \alpha$. 
 If $w \neq \fail$ is the output of \mbox{\rm \textbf{Sample}}$(\alpha, \{q^\alpha\}, \nemptyword, \frac{e^{-5}}{N(q^\alpha)})$,  then for every $x \in \nL(q^\alpha)$:
	\begin{eqnarray*}
		\pr(w = x) & = & \frac{e^{-5}}{N(q^\alpha)}.
	\end{eqnarray*}
	Moreover, the algorithm outputs $\fail$ with probability at most $1 - e^{-9}$. Thus, conditioned on not failing,  {\rm \textbf{Sample}}$(\alpha, \{q^\alpha\}, \nemptyword, \frac{e^{-5}}{N(q^\alpha)})$ returns a uniform element  $x \in \nL(q^\alpha)$.
\end{sloppypar}	
\end{proposition}}
\begin{proof}
	First, we show that every recursive call to \textbf{Sample} satisfies that $\varphi \in (0,1)$. Since $\phi_0 = \frac{e^{-5}}{N(q^\alpha)} > 0$ and with each call $\phi$ does not decrease (because it is divided by a probability), we know that $\phi>0$ at each subsequent call. It remains to show that $\phi<1$ for every recursive call to the \textbf{Sample} procedure. Since $\varphi$ does not decrease after each recursive call, it suffices to show this for the final value of $\varphi$. Notice that at the call $\beta$, with $\beta$ from $n$ to $1$, we have that $\phi$ is divided by a~factor 
	\begin{eqnarray*}
	%\frac{N(P^{\beta}_{w[\beta]})}{N(P^{\beta})} & = &
        \frac{N(P^{\beta}_{w[\beta]})}{N(P^{\beta}_0)+N(P^{\beta}_1)},
	\end{eqnarray*}
where $w[\beta]$ is the $\beta$-th letter of $w$. So in the final call, $\varphi$ has the value:
	\begin{eqnarray*}
		\varphi &= & \bigg(\prod_{\beta=1}^{\alpha} 
		\frac{N(P^{\beta}_{w[\beta]})}{N(P^{\beta}_0)+N(P^{\beta}_1)} \bigg)^{-1}
		 \cdot
		\varphi_0 \\ 
		&  = &  \bigg( \prod_{\beta=1}^{\alpha} 
		\frac{N(P^{\beta}_0) + N(P^{\beta}_1)}{N(P^{\beta}_{w[\beta]})} \bigg)
		\cdot
		\frac{e^{-5}}{N(q^\alpha)}
	\end{eqnarray*}
	%where $w[j]$ is the $j$-th bit of $w$ ($1$-indexed). 
%	Thus, to show that $\phi < 1$ at the end of the algorithm, it suffices to show that on any run of the algorithm, we have that
%	\begin{eqnarray*}
%		\bigg( \prod_{\beta=0}^{\alpha-1} 
%		\frac{N(P^{\beta+1}_0) + N(P^{\beta+1}_1)}{N(P^{\beta})} \bigg) & <  &1. 
%	\end{eqnarray*}
	Given that condition $\mathcal{E}(\beta)$ holds for every layer $\beta < \alpha$, by Proposition~\ref{prop:estimate-sets} we know that $N(P^{\beta}_b)$ is a $(1\pm \kc^{-2})^{\beta}$-approximation of $|\nL(P^{\beta}_b)|$ for all $b \in \{0,1\}$ and $\beta \in [1, \alpha]$ (recall that $P^{\beta}_b$ is a subset of states at layer $\beta-1$).
	It follows that at the final recursive call to \textbf{Sample}, we have that: % 
	%  	\begin{align*}
	%  	    \varphi
	%  	    &=\left(\prod_{j=0}^{\alpha-1} \frac{\widetilde{W}_{w[\alpha - j]}^{\alpha - j} }{\widetilde{W}^{\alpha - j}_0 + \widetilde{W}^{\alpha - j}_1} \right)^{-1}(e^{-4}/ R_{s_i^\alpha}) \\
	%  	    &=\left(\prod_{j=0}^{\alpha-1} \frac{(1 \pm \kc^{-1/4})^{\alpha-j} |W_{w[\alpha - j]}^{\alpha - j}|}{(1 \mp \kc^{-1/4})^{\alpha-j}(|W^{\alpha - j}_0| + |W^{\alpha - j}_1|)} \right)^{-1}(e^{-4}/ R_{s_i^\alpha}) \\
	%  	       &=\left(\prod_{j=0}^{\alpha-1} \frac{(1 \pm \kc^{-1/4})^{\alpha-j}} {(1 \mp \kc^{-1/4})^{\alpha-j}} \right)^{-1} \left(\prod_{j=0}^{\alpha-1} \frac{ |W_{w[\alpha - j]}^{\alpha - j}|}{|W^{\alpha - j}_0| + |W^{\alpha - j}_1|} \right)^{-1}(e^{-4}/ R_{s_i^\alpha}) \\
	%  	      	\end{align*}		
	\begin{align*}
	\varphi \ \ = \ \ &\bigg( \prod_{\beta=1}^{\alpha} 
	\frac{N(P^{\beta}_0) + N(P^{\beta}_1)}{N(P^{\beta}_{w[\beta]})} \bigg)
	\cdot
	\frac{e^{-5}}{N(q^\alpha)} \\
	\leq \ \ &\bigg( \prod_{\beta=1}^{\alpha} 
	\frac{(1+\kc^{-2})^\beta\cdot(|\nL(P^{\beta}_0)| + |\nL(P^{\beta}_1)|)}{(1-\kc^{-2})^\beta\cdot|\nL(P^{\beta}_{w[\beta]})|} \bigg)
	\cdot
	\frac{e^{-5}}{N(q^\alpha)}  \\
	= \ \ &
	\bigg( \prod_{\beta=1}^{\alpha} 
	\frac{(1+\kc^{-2})^\beta}{(1-\kc^{-2})^\beta} \bigg) \cdot
	\bigg( \prod_{\beta=1}^{\alpha} 
	\frac{(|\nL(P^{\beta}_0)| + |\nL(P^{\beta}_1)|)}{|\nL(P^{\beta}_{w[\beta]})|} \bigg)
	\cdot
	\frac{e^{-5}}{N(q^\alpha)}
	\end{align*}		
	Recall that $|\nL(P^{\beta}_0)| + |\nL(P^{\beta}_1)| =  |\nL(P^{\beta+1}_{w[\beta+1]})|$ for every $\beta \in [1,\alpha-1]$. Also note that $|\nL(P^1_{w[1]})| = 1$, since $\nL(P^1)$ is the set of strings $x \in \{0,1\}$ such that $x \cdot w^1 \in \nL({q^\alpha})$, and $\nL(P^1_{w[1]})$ is the subset of $\nL(P^1)$ with the last bit equal to $w[1]$ (of which there is just one).	Thus,  given that $|\nL(q^\alpha)| = |\nL(P^\alpha)| = |\nL(P^\alpha_0)| + |\nL(P^\alpha_1)|$, we have that:
	\begin{eqnarray*}
	\prod_{\beta=1}^{\alpha} 
	\frac{(|\nL(P^{\beta}_0)| + |\nL(P^{\beta}_1)|)}{|\nL(P^{\beta}_{w[\beta]})|} \ = \ |\nL(q^\alpha)|, 
	\end{eqnarray*}
	and so 
	%  	      	\begin{align*}
	%  	    \varphi &= \left(\prod_{j=0}^{\alpha-1} \frac{(1 \pm \kc^{-1/4})^{\alpha-j}} {(1 \mp \kc^{-1/4})^{\alpha-j}} %\right)^{-1} |U_{s_i^\alpha}|(e^{-4}/ R_{s_i^\alpha}) \\
	%  	      &= \left(\prod_{j=0}^{\alpha-1} (1 \pm O(\kc^{-1/4}))^{\alpha-j} \right)|U_{s_i^\alpha}|(e^{-4}/ R_{s_i^%\alpha}) \\
	%  	    &= (1 \pm O(\kc^{-1/4}))^{\alpha^2} |U_{s_i^\alpha}|(e^{-4}/ R_{s_i^\alpha}) \\
	%  	    &= (1 \pm 1/n) |U_{s_i^\alpha}|(e^{-4}/ R_{s_i^\alpha})
	%  	\end{align*}
	\begin{eqnarray*}
		\varphi & \leq & \bigg( \prod_{\beta=1}^{\alpha} 
		\frac{(1+\kc^{-2})^\beta}{(1-\kc^{-2})^\beta} \bigg) \cdot
		|\nL(P^\alpha)|
		\cdot
		\frac{e^{-5}}{N(q^\alpha)} \\
		& = & \bigg(\frac{1+\kc^{-2}}{1-\kc^{-2}}\bigg)^{\frac{\alpha (\alpha + 1)}{2}} \cdot |\nL(q^\alpha)| \cdot
		\frac{e^{-5}}{N(q^\alpha)} \\
		& \leq & \bigg(\frac{1+\kc^{-2}}{1-\kc^{-2}}\bigg)^{\alpha^2} \cdot |\nL(q^\alpha)| \cdot
		\frac{e^{-5}}{N(q^\alpha)}\\
			& \leq & \bigg(\frac{1+\kc^{-2}}{1-\kc^{-2}}\bigg)^{n^2} \cdot |\nL(q^\alpha)| \cdot
		\frac{e^{-5}}{N(q^\alpha)}
	\end{eqnarray*}
\review{
	Furthermore, 
	%since 
	%$N(q^\alpha) = (1 \pm \kc^{-2})^\alpha|\nL(q^\alpha)|$ 
	$N(q^\alpha)$ is a $(1 \pm \kc^{-2})^\alpha$-approximation of $|\nL(q^\alpha)|$ by Proposition~\ref{prop:estimate-sets2}. Then we know that $N(q^\alpha) \geq (1 - \kc^{-2})^\alpha |\nL(q^\alpha)| \geq  (1 - \kc^{-2})^{n^2} |\nL(q^\alpha)|$ and, therefore,}
	\begin{eqnarray*}
		\varphi & \leq & \bigg(\frac{1+\kc^{-2}}{1-\kc^{-2}}\bigg)^{n^2} \cdot |\nL(q^\alpha)| \cdot
		\frac{e^{-5}}{N(q^\alpha)} \\
		& \leq & \bigg(\frac{1+\kc^{-2}}{1-\kc^{-2}}\bigg)^{n^2} \cdot |\nL(q^\alpha)| \cdot
		\frac{e^{-5}}{(1 - \kc^{-2})^{n^2} |\nL(q^\alpha)|} \\
		& = & \frac{(1+\kc^{-2})^{n^2}}{(1-\kc^{-2})^{n^2} \cdot (1-\kc^{-2})^{n^2} } \cdot e^{-5} \ < \ \frac{e}{e^{-2}\cdot e^{-2}} \cdot e^{-5} \ = \ 1
	\end{eqnarray*}
	where the last inequality holds because $\kc = \lceil\frac{nm}{\neps}\rceil \geq n \geq 2$, $(1+\ell^{-1})^\ell < e$ and $(1-\ell^{-1})^\ell \geq e^{-2}$ for every $\ell \geq 2$.
	Hence, we know that under the assumptions stated for this proposition, on each call and, in particular, on  the last call, we have that $\phi \leq 1$. 	
	
	As a second step in the proof of the proposition, we show that the algorithm outputs $\fail$ with probability at most $1 - e^{-9}$. Notice that 
	 %The 
	 this probability 
	 %that the procedure outputs $\fail$ 
	 is only due to Step~(1) in Algorithm \ref{fig:sample}. That is, the probability we output fail is at most $(1-\phi)$, where $\phi$ is as computed  in the previous part of the proof. Thus, to show that the failure probability is at most $1-e^{-9}$, we compute a lower bound for $\varphi$ in a similar way \review{as} we computed an upper bound for it:
	%use the bound from above to observe  $\varphi = \big(\prod_{j=0}^{\alpha - 1} \frac{\widetilde{W}_{w[\alpha - j]}^{\alpha - j} }{\widetilde{W}^{\alpha - j}_0 + \widetilde{W}^{\alpha - j}_1} \big)^{-1} \allowbreak(e^{-4}/ R_{s_i^\alpha}) \geq  (1 -4/n) e^{-4} \geq  e^{-5}$.
	\begin{align*}
 	\varphi \ = \ \ &\bigg( \prod_{\beta=1}^{\alpha} 
	\frac{N(P^{\beta}_0) + N(P^{\beta}_1)}{N(P^{\beta}_{w[\beta]})} \bigg)
	\cdot
	\frac{e^{-5}}{N(q^\alpha)} \ \geq \\
	&\bigg( \prod_{\beta=1}^{\alpha} 
	\frac{(1-\kc^{-2})^\beta\cdot(|\nL(P^{\beta}_0)| + |\nL(P^{\beta}_1)|)}{(1+\kc^{-2})^\beta\cdot|\nL(P^{\beta}_{w[\beta]})|} \bigg)
	\cdot
	\frac{e^{-5}}{N(q^\alpha)}  \ = \\
         &\bigg( \prod_{\beta=1}^{\alpha} 
	\frac{(1-\kc^{-2})^\beta}{(1+\kc^{-2})^\beta} \bigg) \cdot
	\bigg( \prod_{\beta=1}^{\alpha} 
	\frac{(|\nL(P^{\beta}_0)| + |\nL(P^{\beta}_1)|)}{|\nL(P^{\beta}_{w[\beta]})|} \bigg)
	\cdot
	\frac{e^{-5}}{N(q^\alpha)} \ \geq\\
	& 
	 \bigg(\frac{1-\kc^{-2}}{1+\kc^{-2}}\bigg)^{n^2} \cdot |\nL(q^\alpha)| \cdot
	\frac{e^{-5}}{(1+\kc^{-2})^{n^2} \cdot |\nL(q^\alpha)|} \ = \\ 
	&\frac{(1-\kc^{-2})^{n^2}}{(1+\kc^{-2})^{n^2} \cdot (1+\kc^{-2})^{n^2} } \cdot e^{-5} \ \geq  \ \frac{e^{-2}}{e\cdot e} \cdot e^{-5} \ = \ e^{-9} 
	\end{align*}
	%The probability that the procedure outputs \texttt{FAIL} is then only due to Step 2. That is, the probability we output fail is $\phi$ where $\phi$ is as computed above. So to show that the failure probability is at most $1-e^{-5}$, we use the bound from above to observe  $\varphi = \big(\prod_{j=0}^{\alpha - 1} \frac{\widetilde{W}_{w[\alpha - j]}^{\alpha - j} }{\widetilde{W}^{\alpha - j}_0 + \widetilde{W}^{\alpha - j}_1} \big)^{-1} \allowbreak(e^{-4}/ R_{s_i^\alpha}) \geq  (1 -4/n) e^{-4} \geq  e^{-5}$.
	\review{
	Note that in the fourth step we use the fact that 
	$N(q^\alpha)$ is a $(1 \pm \kc^{-2})^\alpha$-approximation of $|\nL(q^\alpha)|$ by 
	Proposition~\ref{prop:estimate-sets2} (indeed, implied by the assumption that $\mathcal{E}(\beta)$ holds for each $\beta < \alpha$), so that 
	%$N(q^\alpha) = (1 \pm \kc^{-2})^\alpha|\nL(q^\alpha)|$ and then 
	$N(q^\alpha) \leq (1 + \kc^{-2})^\alpha |\nL(q^\alpha)| \leq  (1 + \kc^{-2})^{n^2} |\nL(q^\alpha)|$.
	}
	
	As the final step of the proof, we need to show that if the output of the algorithm is $w \neq \fail$, then 
	$\pr(w = x)  = e^{-5}/N(q^\alpha)$ for every $x\in \nL(q^\alpha)$.
	Now, the probability of 
	%the output 
	$w$ being a particular $x\in \nL(q^\alpha)$ is given by the following expression:
	\begin{eqnarray*}
	\pr(w = x) & = & \pr(w^0 = x \wedge \text{ the last call to \textbf{Sample} does not fail}) \\
	& = & \pr(\text{last call to \textbf{Sample} does not fail} \mid w^0 = x)\cdot \pr(w^0 = x) \\
	&= & \bigg( \bigg(\prod_{\beta=1}^{\alpha} 
	\frac{N(P^{\beta}_{w[\beta]})}{N(P^{\beta}_0)+N(P^{\beta}_1)} \bigg)^{-1}  \cdot\frac{e^{-5}}{N(q^\alpha)}\bigg)
	\cdot
	\bigg(\prod_{\beta=1}^{\alpha} 
	\frac{N(P^{\beta}_{w[\beta]})}{N(P^{\beta}_0)+N(P^{\beta}_1)} \bigg)\\
	&= & \frac{e^{-5}}{N({q^\alpha})},
	\end{eqnarray*}
	as desired. 
	%Note by assumption on the approximation given by $R_{s_i^\alpha}$, we have $e^{-4}/R_{s_i^\alpha} < 1/U_{s_i^\alpha}$. 
	This concludes the proof of the proposition.
\end{proof}
\review{
We would like to remark that, in order for Proposition~\ref{prop:unifsamplecondition} to be correct, we need that $\mathcal{E}(\beta)$ holds for every layer $\beta < \alpha$. Indeed, the sampling procedure uses values $N(P^\beta_b)/(N(P^\beta_0)+N(P^\beta_1))$ for approximating the real probabilities $|\nL(P^\beta_b)|/|\nL(P^\beta)|$. For this, we need that each value $N(P^\beta)$ is a good estimate for $|\nL(P^\beta)|$ and this is implied by Proposition~\ref{prop:estimate-sets2} if $\mathcal{E}(\beta)$ holds for every layer $\beta < \alpha$. In the next section, we prove that indeed this happen with exponentially high probability.
}
	%Therefore, Proposition \ref{prop:unifsamplecondition} shows that conditioned on not outputting $\fail$, our sampler returns $w \in \nL(q^\alpha)$ \textit{uniformly at random}, which is the desired result. 

%% file: approximate-bounding.tex
\review{
As it was previously discussed, the computation of the sketch composed by the estimates $N(q^\alpha)$ and sets $S(q^\alpha)$ is subject that conditions $\mathcal{E}(\alpha)$ hold for all layers $\alpha \leq n$. Therefore, this section is aimed to bound the probability that $\mathcal{E}(\alpha)$ is false for some layer $\alpha$ and show that, indeed, this probability is exponentially low. 

First, assume that we are back to a layer $\alpha$, condition $\mathcal{E}(\beta)$ holds for all layers $\beta < \alpha$, and we want to check the probability that $\mathcal{E}(\alpha)$ holds for $\alpha$. 
In other words, we want to bound $\pr(\mathcal{E}(\alpha) \mid  \bigwedge_{\beta=0}^{\alpha-1} \mathcal{E}(\beta)  )$,
%To bound this, we need
for which we need
%to recall the well-known
Hoeffding's inequality.	
\begin{proposition}[Hoeffding's inequality \cite{hoeffding1963probability}]
	Let $X_1,\dots,X_t$ be independent random variables bounded by the interval $[0,1]$ such that $\ex{X_i} = \mu$. Then for every $\ndelta >0$, it holds that
	%\begin{eqnarray*}
	%\pr(|S -  \ex{S}| \geq t ) & \leq & 2e^{-2 n t^2}.
	%\end{eqnarray*}
	\begin{eqnarray*}
		\pr\bigg(\bigg|\frac{1}{t}\sum_{i=1}^t X_i -  \mu\bigg| \geq \ndelta \bigg) & \leq & 2e^{-2 t \ndelta ^2}.
	\end{eqnarray*}
\end{proposition}

For the first layer $\alpha = 0$, the condition $\mathcal{E}(0)$ certainly holds. 
Now if we are at any layer $\alpha$ and $\bigwedge_{\beta=0}^{\alpha-1} \mathcal{E}(\beta)$ holds, then we know by Proposition~\ref{prop:unifsamplecondition} that for each $q \in Q$, it is possible to fill $S(q^\alpha)$ with $2\kc^7$ uniform samples of $\nL(q^\alpha)$. Consider the case of any subset $P \subseteq Q$, and let $S(q^\alpha) = \{w_1, \ldots, w_t\}$ be the uniform sample of $\nL(q^\alpha)$ of size $t = 2\kc^7$. For each $w_i$, consider the random variable $X_i$ such that $X_i = 1$ if $w_i \in (\cL(q^\alpha) \setminus \bigcup_{p \in P} \nL(p^\alpha))$, and $0$ otherwise. Then we have that:
%$$
%\renewcommand{\arraystretch}{1.4}
%\begin{array}{rcl}
\begin{eqnarray*}
\ex{X_i} & = & \frac{|\nL(q^\alpha) \setminus \bigcup_{p \in P} \nL(p^\alpha)|}{|\nL(q^\alpha)|} \\
\sum_{i=1}^t X_i & = & |S(q^\alpha) \setminus \bigcup_{p \in P} \nL(p^\alpha)|, \ \text{ and } \\
t & = & |S(q^\alpha)|.
\end{eqnarray*}
%\end{array}
%$$
Therefore, by Hoeffding's inequality we infer that:
\begin{equation*}
\pr\bigg(
\bigg| 
\frac{|\nL(q^\alpha) \setminus \bigcup_{p \in P} \nL(p^\alpha)|}{|\nL(q^\alpha)|} 
-
\frac{|S(q^\alpha) \setminus \bigcup_{p \in P} \nL(p^\alpha)|}{|S(q^\alpha)|}
\bigg| \ \geq \ \frac{1}{\kc^{3}} \ \ \bigg|  \ \ \bigwedge_{\beta=0}^{\alpha-1} \mathcal{E}(\beta) \bigg) \ \leq\ 
2e^{-4\kc}
\end{equation*}
Note that in the previous inequality the condition $\bigwedge_{\beta=0}^{\alpha-1} \mathcal{E}(\beta)$ does not change the assumptions of Hoeffding's inequality.
We can bound $\pr( \neg \mathcal{E}(\alpha) \mid \bigwedge_{\beta=0}^{\alpha-1} \mathcal{E}(\beta) )$ by taking the union bound over all states $q$ and all possible subsets $P \subseteq Q$:
\begin{multline*}
\pr\bigg(
\exists{q \in Q} \  \ \exists P \subseteq Q \ \
\bigg| 
\frac{|\nL(q^\alpha) \setminus \bigcup_{p \in P} \nL(p^\alpha)|}{|\nL(q^\alpha)|} 
-
\frac{|S(q^\alpha) \setminus \bigcup_{p \in P} \nL(p^\alpha)|}{|S(q^\alpha)|}
\bigg| \ \geq \ \frac{1}{\kc^{3}}  \ \ \bigg|  \ \ \bigwedge_{\beta=0}^{\alpha-1} \mathcal{E}(\beta) \bigg) \ \leq\\
m2^m \cdot 2e^{-4\kc}\  \leq \ e^{2nm} \cdot e^{-4\kc} \ \leq \ e^{-2\kc}.
\end{multline*}	
We conclude that, at layer $\alpha$, the probability $\pr(\mathcal{E}(\alpha) \mid \bigwedge_{\beta=0}^{\alpha-1} \mathcal{E}(\beta) ) \geq 1 - e^{-2\kc}$.

To extend these chain of implications over all layers, we can use that:
\begin{eqnarray*}
	\pr(\mathcal{E}(0)  \wedge \cdots \wedge \mathcal{E}(n) ) & = & \prod_{\alpha=1}^n \pr(\mathcal{E}(\alpha) \mid  \bigwedge_{\beta=0}^{\alpha-1} \mathcal{E}(\beta) )\\
	& \geq &  \prod_{\alpha=1}^n (1 - e^{-2\kc}) \ = \  (1 - e^{-2\kc})^n.
\end{eqnarray*}
Moreover, we have that:
\begin{eqnarray*}
	(1 - e^{-2\kc})^n & = & 1 + \sum_{j=1}^n \binom{n}{j}(-1)^j e^{- 2\kc \cdot j}\\
	& \geq & 1 - \sum_{j=1}^n \binom{n}{j} e^{- 2\kc \cdot j}\\
	& \geq & 1 - e^{-2\kc} \cdot \sum_{j=1}^n \binom{n}{j}\\
	& \geq & 1 - e^{-2\kc} \cdot 2^n\\
	& \geq & 1 - e^{-2\kc} \cdot e^{\kc}  = 1-e^{-\kc}
\end{eqnarray*}
We conclude
%We end this section
by stating the main purpose of this section in the following proposition.  
\begin{proposition}\label{prop:bounding-prob}
The probability that $\mathcal{E}(\alpha)$ holds for all layers $\alpha \leq n$ is bounded below by $1-e^{-\kc}$.
\end{proposition}
}

%% file: approximate-algo.tex
% !TEX root = main.tex
% !TeX spellcheck = en_US

%\marcelo{Include the proof of Theorem \ref{theo-snfa-plvug}}

\begin{figure*}
	\begin{Frame}[\textbf{Algorithm \ref{fig:mainalgo}: FPRAS to estimate $|\nL_n(\nA)|$ for an NFA $\nA = (\nQ, \{0,1\}, \Delta, \nI, \nF)$ with $m \geq 2$ states, integer $n \geq 2$ given in unary and error $\neps \in (0,1)$}]
		\label{fig:mainalgo}
		\begin{enumerate}
			\item If $\nL_n(\nA) = \emptyset$, then return 0.
			\item Else, construct the directed acyclic graph $\nAun$ from $\nA$, and set $\kc = \lceil \frac{nm}{\neps} \rceil$.
			\item For each vertex $q^\alpha$ of $\nAun$, if there is no path from a vertex in $I^0$ to $q^\alpha$, then remove $q^\alpha$ from $\nAun$.
			\item \label{alg-main-fl} For each $q^0 \in I^0$, set $N(q^0) = 1$ and $S(q^0) = \{\lambda\}$. 
			\item For layers $\alpha = 1,2,\dots,n$ and for each vertex $q^\alpha$ in $\nAun$:
			\begin{enumerate}
				\item Let $R_b = \{p^{\alpha-1} \in Q^{\alpha-1} \mid (p^{\alpha-1}, b, q^{\alpha})$ is an edge in $\nAun\}$ for $b = 0,1$.
				%and $R_1 = \{p^{\alpha-1} \mid p^{\alpha-1} \xrightarrow{1} q^{\alpha}\}$.
				\item Set $N(q^\alpha) = N(R_0) + N(R_1)$.
				
				\item Set $S(q^\alpha) = \emptyset$. Then while $|S(q^\alpha)| < 2\kc^7$:
				\begin{enumerate}
					\item Run \textbf{Sample}$(\alpha, \{q^\alpha\},\, \nemptyword,\, \frac{e^{-5}}{N(q^\alpha)})$ until it returns a string $w \neq \fail$, and at most $\Theta(\log(\kc))$ times
					% \marcelo{It could be the case that $R(s_i^\alpha) = 0$, for example if $s_i$ is not reachable from the initial state in $N$ or if $N$ does not accept a string of length $\alpha$. We need to say how to handle this case (notice that preprocessing $N$ to remove unreachable states from the initial state is not enough).} Fixed, thanks!
					\item If $w = \fail$, then terminate the algorithm and output $0$ as the estimate (failure event). 
					\item Otherwise, a sample $w \in \{0,1\}^\alpha$ was returned, and set $S(q^\alpha) = S(q^\alpha) \cup \{w\}$ (recall $S(q^\alpha)$ allows duplicates).
				\end{enumerate}
				
			\end{enumerate}
			\item Return $N(F^n)$ as an estimate for $|\mathcal{L}_n(\nA)|$.
		\end{enumerate}
	\end{Frame}
\end{figure*}

%We have all the ingredients to give the main algorithm. 
In Algorithm~\ref{fig:mainalgo}, we give all the steps of the FPRAS that has been discussed in the previous subsections. This algorithm follows the same structure of Algorithm \ref{fig:FPRAStemplate}, but now the computation of $N(q^\alpha)$ and $S(q^\alpha)$ is fully described. 
The algorithm proceeds as mentioned before, layer by layer, computing the estimates $N(q^\alpha)$ and the sample sets $S(q^\alpha)$. 
For each state $q^0$ at the initial layer $\alpha = 0$,  the pair $N(q^\alpha)$, $S(q^\alpha)$ is computed without considering any additional information (Step~(\ref{alg-main-fl}) of Algorithm \ref{fig:mainalgo}). 
%the layer $0$ is computed directly. 
Then for each layer $\alpha > 0$ and each vertex $q^\alpha$, we compute $N(q^\alpha) = N(R_0) + N(R_1)$, as was discussed in Section~\ref{subsec:sets}. After $N(q^\alpha)$ is computed, the set $S(q^\alpha)$ is filled with $2\kc^7$ uniform and independent samples. For this, \textbf{Sample}$(\alpha, \{q^\alpha\},\, \nemptyword,\, \frac{e^{-5}}{N(q^\alpha)})$ is run at most $\Theta(\log(\kc))$ times or until a string $w \neq \textbf{fail}$ is output.
If $w = \fail$, then the algorithm terminates and outputs $0$. Otherwise, the sample is added to $S(q^\alpha)$, and the computation continues. 
Finally, the algorithm computes and returns $N(F^n)$ as it was discussed in Section~\ref{subsec:sets}.

\review{
Given that the value $\kc$ is polynomial in $m$, $n$ and $\frac{1}{\neps}$, it is clear that Algorithm~\ref{fig:mainalgo} works in time polynomial in $m$, $n$ and $\frac{1}{\neps}$. 
Hence, it only remains to show that Algorithm~\ref{fig:mainalgo} is correct, namely, that $N(F^n)$ is a $(1\pm \neps)$-approximation of $|\mathcal{L}_n(\nA)|$ with probability greater than $\frac{3}{4}$. First, assume that the algorithm returns a good estimate, that is, condition
$\mathcal{E}(\alpha)$ holds for all layers $\alpha \leq n$ during the run of the algorithm and $w$ is never equal to $\fail$ after Step (i) of the algorithm. Then by Proposition~\ref{prop:final-estimate}, we obtain that  $N(F^n)$ is a $(1\pm \neps)$-approximation of $|\mathcal{L}_n(\nA)|$ as desired. 

To conclude the proof of the correctness of Algorithm~\ref{fig:mainalgo}, we need to 
bound the probability that the algorithm does not give a good estimate, namely, that either condition $\mathcal{E}(\alpha)$ is false for some layer $\alpha$ or the sampling algorithm fails 
%$\Theta(\log(\kc))$ 
$c(\kc)$ times at Step (i), where $c(\kc)$ is the number of repetitions performed in this step. 
Therefore, it remains to show that this probability of giving a wrong output is at most $\frac{1}{4}$ considering a value for $c(\kc) \in \Theta(\log(\kc))$.
Let $\mathcal{E}_{\textbf{fail}}(\alpha, q, j)$ be the event that 
%\textbf{Sample} 
the call \textbf{Sample}$(\alpha, \{q^\alpha\},\, \nemptyword,\, e^{-5}/N(q^\alpha))$
%algorithm 
fails 
%$\Theta(\log(\kc))$ 
$c(\kc)$ consecutive times at layer $\alpha$, state $q^\alpha$, and the $j$-th sample of $S(q^\alpha)$, where $j \in [1, 2\kc^7]$. 
We know by Proposition~\ref{prop:unifsamplecondition} that the probability that \textbf{Sample} fails is %bounded above by 
at most $1-e^{-9}$ and, therefore, $\pr(\mathcal{E}_{\textbf{fail}}(\alpha, q, j)) \leq (1-e^{-9})^{c(\kc)}$. 
Furthermore, by Proposition~\ref{prop:bounding-prob} we already know that $\pr( \neg \mathcal{E}(0)   \vee  \cdots \vee \neg \mathcal{E}(n) ) \leq e^{-\kc}$.  
%for the number $c$ of repetitions used  at Step (i) of the algorithm. 
%\in \Theta(\log(\kc))$. 
By the union bound, we conclude that the probability that the algorithm gives the wrong output is:
\begin{eqnarray*}
	\pr\Big(\neg  \bigvee_{\alpha=1}^n \bigvee_{q \in Q} \bigvee_{j=1}^{2\kc^7} \mathcal{E}_{\textbf{fail}}(\alpha, q, j) \vee \neg \mathcal{E}(0) \vee  \cdots \vee \neg \mathcal{E}(n)   \Big)
	 & \leq & \sum_{\alpha=1}^n \sum_{q \in Q} \sum_{j=1}^{2\kc^7} (1-e^{-9})^{c(\kc)} + e^{-\kc} \\
	 & \leq &  2 n m \kc^7 (1-e^{-9})^{c(\kc)} + e^{-2} \\
	 & \leq &  2 \kc^8 (1-e^{-9})^{c(\kc)} + e^{-2}.
 \end{eqnarray*}
Finally, if we take 
\begin{eqnarray*}
c(\kc) & = & \Big\lceil \frac{2+\log(4)+8\log(\kc)}{\log((1-e^{-9})^{-1})} \Big\rceil,
\end{eqnarray*} 
we obtain that $2 \kc^8 (1-e^{-9})^{c(\kc)} \leq \frac{1}{2} \cdot e^{-2}$. Therefore, we have that the probability that the algorithm returns a wrong estimate is at most $\frac{3}{2} \cdot e^{-2} < \frac{1}{4}$. As $c(\kc)$ is $\Theta(\log(\kc))$, this concludes the proof of the correctness of~Algorithm~\ref{fig:mainalgo} as stated in Theorem \ref{theo-snfa-fpras}.
}

\begin{sloppypar}
To finish this section, we need to prove Theorem \ref{theo-snfa-pplvug}, that is, we need to show that $\GEN(\nfa)$ admits a preprocessing polynomial-time Las Vegas uniform generator (PPLVUG).
%We now remark 
Notice that our algorithm 
%immediately 
results in a sampler satisfying the conditions of Theorem \ref{theo-snfa-pplvug}.  
\review{
Specifically, given an NFA $\nA$ and a natural number $n$, Algorithm~\ref{fig:mainalgo} builds a structure $(\nAun, \{N(p^\alpha), S(p^\alpha)\}_{p \in Q,\alpha \leq n})$. Furthermore, conditioned on all the above events (i.e., $\mathcal{E}(0)$, $\ldots$, $\mathcal{E}(n)$, and $\mathcal{E}_{\textbf{fail}}$) we can use $(\nAun, \{N(p^\alpha), S(p^\alpha))\}_{p \in Q,\alpha \leq n})$ for getting a uniform sample from $\mathcal{L}_n(\nA)$ by using Algorithm~\ref{fig:sample}, and this algorithm fails with probability $1-e^{-9}$.
In other words, Algorithm~\ref{fig:mainalgo} and Algorithm~\ref{fig:sample} are the randomized algorithms $\mathcal{P}$ and $\mathcal{G}$, respectively, and  $(\nAun, \{N(p^\alpha), S(p^\alpha))\}_{p \in Q,\alpha \leq n})$ is the string $\cD$ that is good-for-generation (i.e., the advice) from the PPLVUG's definition.}
The probability of success of the above events can be amplified to $1-\delta$ for any $\delta > 0$ by scaling $\kc$ up by a factor of $\log(1/\delta)$.
%, since then the above failure probability will be  $ 	\pr(\mathcal{E}(0)  \wedge  \wedge \cdots \wedge \mathcal{E}(n) ) \geq 1 - e^{-\kc} > 1- \delta$. 
The overall runtime is now polynomial in $(n,m,\log(1/\delta))$ as needed. Note that for the purposes of a uniform sampler, the parameter $\eps$ can be set to a constant, as it does not appear in the definition of the sampler in  Theorem \ref{theo-snfa-pplvug}. Now by Proposition \ref{prop:unifsamplecondition}, we can obtain truly uniform samples from the set $\mathcal{L}_n(\nA)$ \review{in time polynomial in} $(n,m,\log(1/\delta))$ time. The probability that $\fail$ is returned by the sampler is at most $1-e^{-9}$, which can be amplified to at most~$1/2$ \review{by repeating it a constant number of times} to satisfy the condition required in Theorem~\ref{theo-snfa-pplvug}, which completes the proof of this theorem. 
\end{sloppypar}

%% file: conclusions.tex
%!TEX root = main.tex

We consider this work as a first step towards the definition of
classes of problems with good properties in terms of enumeration,
counting and uniform generation of solutions. In this sense, there is
plenty of room for extensions and \review{improvements, and many problems need
to be studied further. First, for each one of the classes
$\rnl$ and $\rul$, we have identified a single problem that is
complete for it. An important question then is whether such classes
admit other natural and well-studied complete problems; for instance, we leave as an open problem whether $\dnf$ is complete for $\rnl$ under the notion of reduction introduced in Section~\ref{cde}. Second,}
%In particular,
the different components of the FPRAS for $\snfa$ were designed to
facilitate its proof of correctness. As such, we already know of some
optimizations that significantly reduce its runtime, and we plan
to develop more such optimizations so to make this FPRAS usable in
practice.
\review{Finally, an interesting area to explore is to extend the results of this article to the class of context-free languages.
%, and extend the computational models and relations based on them.
In particular, it will be interesting to understand if relations based on context-free languages have good properties regarding enumeration, counting, and uniform generation. Here, it is natural to ask whether the problem $\#\text{CFG}$ (i.e., to count the number of words of a given length accepted by a context-free grammar) admits an FPRAS or not.}

%% file: lemmas.tex
%!TEX root = main.tex

\subsection{Proof of Lemma \ref{lemma_ufa_from_relation}}
\label{app-lemma_ufa_from_relation}
		Let $x$ be any element in $\review{\{0,1\}}^*$. Since $R$ is in $\rul$, we know there exists a $\ul$-transducer $M$ such that $W_R(x)=M(x)$. Without loss of generality, we can assume that $M$ has only one accepting state, so it can be written as a tuple $M=(Q,\Gamma,\B,\review{\{0,1\}},\delta,q_0,\{q_F\})$. If it has more than one accepting state, say a set $F$ of accepting states, we can define a new transducer $M'$ that is identical to $M$ with one difference. It has only one final state $q_F$ and whenever it reaches a state in $F$, it makes one last transition to $q_F$ and stops. It is clear that $M(x)=M'(x)$ so we do not lose any generality with this assumption.
		
		Let $n=|x|$, $f(n)$ be the function that bounds the number of cells in the work tape that can be used, and assume that $f(n)$ is $O(\log(n))$. Consider now an execution of $M$ on input $x$. Since the input tape never changes (its content is always $x$), we can completely characterize the configuration of the machine at any given moment as a tuple $(q,i,j,w)\in Q\times \{1,\ldots,n\}\times \{1,\ldots,f(n)\}\times \Gamma^{f(n)}$ where
		\begin{itemize}
			\item $q$ stores the state the machine is in.
			\item $i$ indicates the position of the head on the input tape.
			\item $j$ indicates the position of the head on the work tape.
			\item $w$ stores the contents of the work tape.
		\end{itemize}
		With the previous notation, the initial configuration of $M$ on input $x$ is represented by $c_I=(q_0,1,1,\B^{f(n)})$, that is, $M$ is in its initial state, the heads are at the first position of their respective tapes, and the work tape is empty (that is, it only contains the blank symbol $\B$). The accepting configuration is represented by a tuple of the form $c_F=(q_F,i_F,j_F,w_F)$. Notice that without loss of generality, \review{we can assume} the accepting configuration to be unique, by changing $M$ so that it runs for a little longer in order to reach it. If $C_x$ is the set of possible configuration tuples then we have that
		\begin{align*}
		|C_x| &\leq |Q|\cdot n\cdot f(n)\cdot |\Gamma|^{f(n)} \\
		&= |Q| \cdot n\cdot f(n)\cdot |\Gamma|^{O(\log(n))} \\
		&= |Q| \cdot n\cdot f(n)\cdot O(n^\ell), \text{ where } \ell \text{ is a constant}\\
		&= O(n^{\ell+1}\log(n)),
		\end{align*}
		which is polynomial in $|x|$. Recall the notation for NFAs introduced in Section \ref{sec-approximate-technique}. We now define the NFA $\nA_x = (C_x, \review{\{0,1\}}, \Delta_x, c_I, \{c_F\})$ where $C_x$, $c_I$ and $c_F$ are defined as above and the transition relation $\Delta_x$ is constructed in the following way:
		
		\begin{itemize}
			\item Let $c,d\in C_x$. Consider any possible run of $M$ on input $x$. Suppose there is a valid transition, during that run, that goes from $c$ to $d$ while outputting symbol $\gamma\in\Gamma$. Then, $(c,\gamma,d)$ is in $\Delta_x$.
			\item Let $c,d\in C_x$. Consider any possible run of $M$ on input $x$. Suppose there is a valid transition, during that run, that goes from $c$ to $d$ while making no output. Then, $(c,\epsilon,d)$ is in $\Delta_x$.
		\end{itemize}
		
		We already showed that $C_x$ has polynomial size in $|x|$, and it clearly can be constructed explicitly in polynomial time. The same is true for $\Delta_x$. Given a pair of configurations $c,d\in C_x$, it can be checked in polynomial time 
		%quick to check 
		whether there is a possible transition from $c$ to $d$ during an execution of $M$ on input $x$ (it suffices to check $\delta$, the transition relation for $M$). And there are just $|C_x|^2$ such pairs of configurations that we need to check, so the whole construction of $\nA_x$ can be done in polynomial time. It only rests to show that $W_R(x)=\mathcal{L}(\nA_x)$ and that $\nA_x$ is unambiguous.
		
		Let $y\in W_R(x)$. That means there is an accepting \review{run} of $M$ on input $x$ that yields $y$ as output. Equivalently, there is a sequence of configurations $\{c_k\}_{k=0}^m$ \review{and a sequence $\{w_k\}_{k=0}^m$ of symbols} such that:
		
		\begin{itemize}
			\item $c_0=c_I$.
			\item $c_m=c_F$.
			\item For each $k\in\{0,\ldots,m-1\}$, the transition from $c_k$ to $c_{k+1}$ is valid on input $x$ given the 	transition relation $\delta$ of $M$.
			\item For each $k\in\{0,\ldots,m-1\}$, we have that $w_k$ is equal to the symbol output when going from 			configuration $c_k$ to $c_{k+1}$ if a symbol was output. Otherwise, $w_k=\varepsilon$.
			\item \review{$y=w_0\cdot w_1\cdot \ldots \cdot w_m$}.
		\end{itemize}
		
		By definition, \review{this means} that $y$ is accepted by $\nA_x$. That is, $y\in \mathcal{L}(\nA_x)$ and so we can conclude that $W_R(x)\subseteq \mathcal{L}(\nA_x)$. Since all the previous implications are clearly \review{equivalences}, we can also conclude that $\mathcal{L}(\nA_x)\subseteq W_R(x)$. Hence $W_R(x) = \mathcal{L}(\nA_x)$ as needed. What the previous argument is saying is that every accepting run of $M$ that outputs a string $y$ has a unique corresponding accepting run of $\nA_x$ on input $y$. That implies that $\nA_x$ is unambiguous. Otherwise, there would be some $y\in \mathcal{L}(\nA_x)$ such that two different runs of $\nA_x$ accept $y$. But that would mean that there are two different runs of $M$ on input $x$ that output $y$, which cannot occur, since $M$ is a $\ul$-transducer.
		
		Finally, notice that $\nA_x$ is actually not an NFA (under the definition given in Section \ref{nlclass}), since we explicitly allowed for the possibility of $\varepsilon$-transitions. But recall that the $\varepsilon$-transitions of any NFA can be removed in polynomial time without changing the accepted language, which is a standard result from automata theory. This concludes the proof of the lemma.

\subsection{Proof of Proposition \ref{prop_self_reducibility}}
\label{app-prop_self_reducibility}
We focus on the case of $\nfa$ (it extends easily to $\unfa$). To show this result, we need to include a little more detail in our definition of $\nfa$, to consider some corner cases. First of all, we have to consider the cases where the string in unary is empty. That is, the case where $k=0$ in input $(N,0^k)$. This just amounts to the following: if the starting state is a final state, we   consider that the automaton does accept the empty string. So, if $k=0$, and $N$ is an NFA that has all the properties stated in the definition of $\nfa$, plus its starting state is the accepting state, then $((N,0^k), \varepsilon)\in\nfa$. Also, we need to consider the cases where $N$ does not have all the properties stated in the definition of $\nfa$. In those cases, we consider that $(N,0^k)$, for any $k$, does not \review{have any solutions}. Also, and this gets more technical, we consider that any input that has an invalid encoding does \review{not have any solutions} either. We will not be completely precise about which encoding should be used (although during the proof we will mention some important points regarding that). But we will ask that the \review{correctness} of the encoding can be checked in polynomial time (this is a mild requirement as any reasonable encoding will allow for it). And it is important to have in mind that for some technical concepts like \change{self-reducibility}, the encoding of the problem is critical.

We use the notion of self-reducibility stated in \cite{schmidt2009enumeration}, 
%because we want to utilize a result from that article which is proved under that specific notion of self-reducibility. We include the definition here,
adapted to our situation, since \cite{schmidt2009enumeration} uses a slightly different framework to define an enumeration problem. We say a relation $R \subseteq \review{\{0,1\}}^*\times~\review{\{0,1\}}^*$ is self reducible if there exist polynomial-time computable functions $\psi : \review{\{0,1\}}^* \times \review{\{0,1\}}^* \to \review{\{0,1\}}^*$, $\sigma : \review{\{0,1\}}^* \to \N$ and $\ell: \review{\{0,1\}}^* \to \mathbb N$ such that for every $x,y,w\in\review{\{0,1\}}^*$:
\begin{enumerate}
	\itemsep0.5em
	\item if $(x,y)\in R$, then $|y|=\ell(x)$,
	\item if $\ell(x)=0$, it can be tested in polynomial time in $|x|$, whether the empty string \review{is a solution} for $x$.
	\item $\sigma(x) \in O(\log|x|)$,
	\item $\ell(x)>0$ if and only if $\sigma(x)>0$,
	\item $|\psi(x,w)| \leq |x|$,
	\item $\ell(\psi(x,w)) = \max \{\ell(x) - |w|, 0\}$, and
	\item $W_R(x) = \displaystyle \bigcup_{w\in\review{\{0,1\}}^{\sigma(x)}} \{\review{w\cdot y} \>|\> y\in W_R(\psi(x,w))\}$.
\end{enumerate}
\review{The last condition intuitively says that all solutions for a given input can be constructed from the solutions of (a polynomial number of) smaller instances.} It can be equivalently stated in the following way, which is how we will use it:
\begin{enumerate}
\setcounter{enumi}{7}
	\item if $y = y_1y_2\dots y_m$, it holds that $(x,y) \in R \mbox{ if and only if } (\psi(x,y_1\dots y_{\sigma(x)}), y_{\sigma(x)+1} \dots y_m) \in R$.
\end{enumerate}

As we already stated, the empty string \review{is a solution} only when the input is correctly encoded and the initial and final states of the automaton coincide. So condition (2) from \change{above} is satisfied regardless of our definition of $\ell$. We will focus from now on on the other six conditions. Let $\mathcal{N} = \{N \>|\> N \text{ is an NFA with a unique final state and no $\varepsilon$-transitions}\}$. Following the previous notation, we define the functions $\ell$, $\sigma$ and $\psi$ that characterize self-reducibility. The only interesting cases, of course, are those where the automaton in the input is in $\mathcal{N}$ (and the input is correctly encoded). In \review{all the others}, the input is not correct, so \review{the set of solutions} is empty, and we do not need to worry about self-reducibility. That said, we define
\begin{align*}
\ell((N,0^k)) &=
\begin{cases}
k & \text{if the input is correctly encoded and $N\in\mathcal{N}$} \\
0 & \text{in any other case}
\end{cases}
\\
\sigma((N,0^k)) &=
\begin{cases}
1 & \text{if the input is correctly encoded, $k>0$ and $N\in\mathcal{N}$} \\
0 & \text{in any other case}
\end{cases}
\end{align*}

Both functions are clearly computable in polynomial time. The definition of $\ell$ is just saying that on input $(N,0^k)$, \review{any solution} will have length $k$, which comes directly from the definition of $\nfa$. The definition of $\sigma$ indicates that, for any input, as long \review{as its solutions} have positive length, we can create another input that has \review{the same solutions}, but with the first character removed. Notice that with these definitions, conditions (3) and (4) for self-reducibility are trivially met. Condition (1) is also met, which is easy to see from the definitions of $\nfa$ and $\ell$. The only task left is to define $\psi$ and prove conditions (5), (6) and (8). We now proceed in that direction.

Let $N=(Q,\review{\{0,1\}},\delta,q_0,\{q_F\})$ be an automaton in $\mathcal{N}$. Notice we are making the assumption that $N$ has a unique final state, since it makes the idea clearer and the proof only has to be modified slightly for the general case. We will mention some points about the exact encoding soon (which is key for condition (5) to hold). But first, consider an input $x=(N,0^k)$ which is incorrectly encoded or where $N$ is not in $\mathcal{N}$. Then, it \review{has no solutions} and it is enough to set $\psi(x,w)=x$ for all $w \in \review{\{0,1\}}^*$ (which is clearly computable in polynomial time). In that case, notice that condition (5) is trivially true. Also, notice that since $N$ is not in $\mathcal{N}$ (or is encoded in an incorrect format), we have $\ell(x)=\sigma(x)=0$, so for any $w$ it holds that
\begin{equation*}
\ell(\psi(x,w)) = \ell(x) = 0 = \max \{-|w|, 0\} = \max \{\ell(x)-|w|, 0\}
\end{equation*}
so condition (6) is also true. And given that $\ell(x)=\sigma(x)=0$, condition (8) amounts to checking that for every $y\in\review{\{0,1\}}^*$, it holds that $(x,y) \in \nfa$ if and only if $(x,y) \in \nfa$, which is obviously true. Now, consider the case of an input $x=(N,0^k)$ that is correctly encoded and where $N$ is in $\mathcal{N}$. There are two main cases to consider.

First, the case where $k=0$. This case is also simple, because we can set $\psi(x,w)=x$ for all $w \in \review{\{0,1\}}^*$ (which is computable in polynomial time and \review{implies} that condition (5) is trivially true), and since $\ell(x)=\sigma(x)=0$, 
it is possible to prove as before that conditions (6) and (8) hold. Second, we need to consider the case where $k>0$. Then we have $\sigma(x)=1$, so $\psi(x,w)$ only needs to be defined when $w$ is a single symbol. Then, for both $w\in\review{\{0,1\}}$, we set $\psi((N,0^k),w)=(N', 0^{k-1})$, where $N'$ is defined as follows. Let $Q_w$ be the set
\begin{eqnarray*}
Q_w & = & \{q\in Q \mid (q_0,w,q)\in\delta\}.
\end{eqnarray*}
Thus, $Q_w$ is the set of states that can be reached (with one transition) from the initial state, by reading the symbol $w$. Now, we define $N'=(Q',\review{\{0,1\}},\delta',q_0',\{q_F'\})$ where $q_0'$ is a new state not contained in $Q$, and:
\begin{eqnarray*}
Q' &= &(Q \setminus Q_w) \cup \{q_0'\} \\
\delta' &=& \{(q,a,p) \mid (q,a,p)\in\delta \text{ and } q,p\in Q'\} \cup \{(q,a,q_0') \mid (q,a,p)\in\delta \text{ and } q\in Q', p\in Q_w\} \cup \\
&& \{(q_0',a,p) \mid (q,a,p)\in\delta \text{ and } q\in Q_w, p\in Q'\} \cup \{(q_0',a,q_0') \mid (q,a,p)\in\delta \text{ and } q,p\in Q_w\} \\
q_F' &= &
\begin{cases}
q_F  & \text{if } q_F\in Q' \\
q_0' & \text{if } q_F\not\in Q'
\end{cases}
\end{eqnarray*}

Notice that this construction takes only polynomial time. What we are doing, basically, is the following. Imagine $Q_w$ as a first ``layer" of states reachable from $q_0$ in one step. We want to merge all of $Q_w$ in a single new initial state $q_0'$, while ensuring that from $q_0'$ we can reach the same states as were previously reachable from $Q_w$. The definitions are a little complicated because we have to account for some special cases. For example, we would maybe want to remove $q_0$ (since now we have a new initial state) but there is the possibility that $q_0$ is part of the acceptance runs of some strings, and not only as an initial state. The same goes for the states in $Q_w$, and that is why we have many different cases to consider in the definition of $\delta'$. We have to make sure not to lose any accepting runs with the removal of $Q_w$.

\review{Now, we make some observations} about $N'$. To construct $Q'$, we are removing at least one state from $Q$\review{, but} we are adding at most one new state, $q_0'$. \review{This implies that} $|Q'|\leq |Q|$ (notation here indicates set cardinality). Similarly for the construction of $\delta'$. Notice that each transition we add to construct $\delta'$ (besides the ones that come directly from $\delta$) corresponds to a transition that already existed \review{and that involved} at least one state from $Q_w$. So, all in all, we have not really added any new transitions, just simulated the ones where states in $Q_w$ appeared. That means that $|\delta'|\leq |\delta|$. So, as a whole, $N'$ contains at most as many states and transitions as $N$, and maybe less. Does that mean that (notation here indicates encoding sizes) $|\psi((N,0^k), w)|\leq |(N,0^k)|$? It will depend on the type of encoding used, of course. So we will consider that the NFA in the input is encoded in the following (natural) way. First, a list of all states, followed by the list of all tuples in the transition relation, and at the end the initial and final states. Also, we assume that all states have an encoding of the same size (which is easy to achieve through padding). And the same goes for all transitions. With that encoding, since $N'$ has less (or equal) number of states and transitions than $N$, it is clear that $|N'|\leq |N|$. Of course, it is also true that $|0^{k-1}|\leq |0^k|$. We can then conclude that $|\psi((N,0^k), w)|=|(N',0^{k-1})|\leq |(N,0^k)|$, that is, condition (5) is satisfied. We also have by definition of $\ell$ that $\ell((N,0^k))=k$ and $\ell((N',0^{k-1}))=k-1$. Since $|w|=1$, condition (6) is also true:
\begin{multline*}
\ell(\psi((N,0^k),w)) 
\ = \  \ell((N',0^{k-1})) 
\ = \  k - 1
\ = \\  \ell((N,0^k)) - 1
\ = \  \ell((N,0^k)) - |w|
\ = \  \max \{\ell((N,0^k))-|w|, 0\}.
\end{multline*}
Finally, we turn to condition (8). Let $y=y_1y_2\dots y_m\in\review{\{0,1\}}^*$. Since $\sigma(x)=1$, condition (8) amounts to checking~that
\begin{equation*}
((N,0^k),y) \in \nfa \ \mbox{ if and only if }\  ((N', 0^{k-1}), y_2 \dots y_m) \in \nfa,
\end{equation*}
where $N'$ is constructed by considering $w=y_1$, that is, $N'=\psi((N,0^k),y_1)$. Notice that if $m\not =k$, then both sides of the \review{equivalence} above are immediately false (and thus the \review{equivalence} is true), so we need only consider the case where $m=k$. We will now prove both directions of the \review{equivalence}. First, suppose $((N,0^k),y) \in \nfa$. Then, by definition, we know there is an accepting run $\rho$ of $N$ on input $y$ such that
\begin{equation*}
\rho: p_0 \xrightarrow{y_1} p_1 \xrightarrow{y_2} p_2 \xrightarrow{y_3} \dots \xrightarrow{y_{k-1}} p_{k-1} \xrightarrow{y_k} p_k
\end{equation*}
where $p_0=q_0$, $p_k=q_F$ and $(p_{i-1},y_i,p_i)\in\delta$ for all $i\in\{1,\dots,k\}$. Now, we will show that $((N', 0^{k-1}), y_2 \dots y_k) \in \nfa$, that is, $y_2 \dots y_k$ is accepted by $N'$. To do that, we first show by induction the following property: for all $i\in\{2,\dots,k\}$ there is a valid run of $N'$ on input $y_2 \dots y_i$ (although the run is not necessarily accepting) that looks like this:
\begin{equation*}
\rho_i: s_1 \xrightarrow{y_2} s_2 \xrightarrow{y_3} s_3 \xrightarrow{y_4} \dots \xrightarrow{y_{i-1}} s_{i-1} \xrightarrow{y_i} s_i
\end{equation*}
where $s_1=q_0'$ and for all $j\in\{2,\dots,i\}$, we have that $(s_{j-1},y_j,s_j)\in\delta'$ and
\begin{equation*}
s_j =
\begin{cases}
p_j & \text{ if } p_j\not\in Q_{y_1} \\
q_0' & \text{ if } p_j\in Q_{y_1}.
\end{cases}
\end{equation*}
To prove this fact by induction, consider first the case of $i=2$. By definition, we know that $p_1\in Q_{y_1}$ and $(p_1,y_1,p_2)\in \delta$. There are now two different possibilities. First, if $p_2\not\in Q_{y_1}$, then by definition of $\delta'$, we know that $(q_0',y_2,p_2)\in\delta'$. Second, if $p_2\in Q_{y_1}$, then by definition of $\delta'$, we know that $(q_0',y_2,q_0')\in\delta'$. So the property is true when $i=2$.

Now, suppose the property holds for some $i<k$, and consider the case for $i+1$. By the induction hypothesis, we know there is a valid run $\rho_i$ such that
\begin{equation*}
\rho_i: s_1 \xrightarrow{y_2} s_2 \xrightarrow{y_3} s_3 \xrightarrow{y_4} \dots \xrightarrow{y_{i-1}} s_{i-1} \xrightarrow{y_i} s_i
\end{equation*}
where $s_1=q_0'$ and $(s_{j-1},y_j,s_j)\in\delta'$ for all $j\in\{2,\dots,i\}$. Now, by the induction hypothesis, there are four possibilities (where each possibility is represented in one of the four sets that form the definition of $\delta'$):
\begin{itemize}
	\item $s_i=p_i$ and $p_{i+1}\not\in Q_{y_1}$. In that case, if we set $s_{i+1}=p_{i+1}$, by definition we know that $(s_i,y_{i+1},s_{i+1}) \in \delta'$.
	\item $s_i=p_i$ and $p_{i+1}\in Q_{y_1}$. In that case, if we set $s_{i+1}=q_0'$, by definition we know that $(s_i,y_{i+1},s_{i+1}) \in \delta'$.
	\item $s_i=q_0'$ and $p_{i+1}\not\in Q_{y_1}$. In that case, if we set $s_{i+1}=p_{i+1}$, by definition we know that $(s_i,y_{i+1},s_{i+1}) \in \delta'$.
	\item $s_i=q_0'$ and $p_{i+1}\in Q_{y_1}$. In that case, if we set $s_{i+1}=q_0'$, by definition we know that $(s_i,y_{i+1},s_{i+1}) \in \delta'$.
\end{itemize}
All that means that we can add one more transition to $\rho_i$ to form a valid run $\rho_{i+1}$ given by
\begin{equation*}
\rho_{i+1}: s_1 \xrightarrow{y_2} s_2 \xrightarrow{y_3} s_3 \xrightarrow{y_4} \dots \xrightarrow{y_i} s_i \xrightarrow{y_{i+1}} s_{i+1}
\end{equation*}
where $s_1=q_0'$ and for all $j\in\{2,\dots,i+1\}$, we have that $(s_{j-1},y_j,s_j)\in\delta'$ and
\begin{equation*}
s_j =
\begin{cases}
p_j & \text{ if } p_j\not\in Q_{y_1} \\
q_0' & \text{ if } p_j\in Q_{y_1}.
\end{cases}
\end{equation*}
The property is thus proved. Now, consider a valid run of that type for $i=k$ that looks like
\begin{equation*}
\rho': s_1 \xrightarrow{y_2} s_2 \xrightarrow{y_3} s_3 \xrightarrow{y_4} \dots \xrightarrow{y_{k-1}} s_{k-1} \xrightarrow{y_k} s_k
\end{equation*}
where $s_1=q_0'$ and for all $j\in\{2,\dots,k\}$, we have that $(s_{j-1},y_j,s_j)\in\delta'$. Now, by the property just proved, we know there are two possibilities. First, if $p_k\not\in Q_{y_1}$, we know $s_k=p_k=q_F$. Since $q_F=p_k\not\in Q_{y_1}$, we have that $q_F'=q_F$ and thus $\rho'$ is an accepting run, which means that $y_2 \dots y_k$ is accepted by $N'$. Second, if $p_k\in Q_{y_1}$,  we know $s_k=q_0'$. And since $q_F=p_k\in Q_{y_1}$, we have that $q_F'=q_0'$ and thus $\rho'$ is again an accepting run, which means that $y_2 \dots y_k$ is accepted by $N'$. All this proves that if $((N,0^k),y) \in \nfa$, then $((N', 0^{k-1}), y_2 \dots y_k) \in \nfa$. The proof for the other direction is analogous.

We conclude this section by pointing out that the same proof works for the case of $\unfa$. That is, the same definitions of $\ell$, $\sigma$ and $\psi$ work for that proof. The only difference is that we also need to show that $\psi$ produces a valid automaton for the relation, that is, an unambiguous NFA. But that is not hard to show from the previous proof. Making similar use of the notation of valid runs, it can be shown that if $\psi((N,0^k),w)$ had two different accepting runs for some word $y$, then $N$ would have two different accepting runs for $w\circ y$, and so it would not be unambiguous.